\documentclass[11pt]{article}

\usepackage{latexsym,color,graphicx,amsmath,amssymb,amsthm,url,cite}
\usepackage{fullpage}
\usepackage[noend,linesnumbered,boxed]{algorithm2e}
\graphicspath{{./pictures/}}
\usepackage{multirow}
\usepackage{xstring}
\usepackage[extdef]{delimset}

\def\reals{{\mathbb R}}
\def\eps{{\varepsilon}}
\def\bd{{\partial}}
\def\A{{\cal A}}
\def\C{{\cal C}}
\def\B{{H}}
\def\D{{\cal D}}
\def\DD{{\cal B}}
\def\E{{\cal E}}
\def\G{{\cal G}}
\def\I{{\cal I}}
\def\M{{\cal M}}
\def\R{{\cal R}}

\def\F{{\cal F}}
\def\otau{{\overline{\tau}}}
\def\oLambda{{\overline{\Lambda}}}

\def\sballs{6}
\def\swed{6}

\def\ovarphi{{\overline{\varphi}}}
\def\oLambda{{\overline{\Lambda}}}
\def\oT{{\overline{T}}}

\DeclareMathOperator{\CL}{CL}
\DeclareMathOperator{\VD}{VD}
\DeclareMathOperator{\UD}{UD}
\DeclareMathOperator{\N}{N}
\DeclareMathOperator{\NN}{NN}

\def\marrow{\marginpar[\hfill$\longrightarrow$]{\textcolor{red}{$\longleftarrow$}}}
\newcommand{\remark}[3]{\marrow\textcolor{blue}{\textsc{#1 #2:}} \textcolor{red}{\textsf{#3}}}
\newcommand{\micha}[2][says]{\remark{Micha}{#1}{#2}}
\newcommand{\haim}[2][says]{\remark{Haim}{#1}{#2}}
\newcommand{\liam}[2][says]{\remark{Liam}{#1}{#2}}
\newcommand{\wolfgang}[2][says]{\remark{Wolfgang}{#1}{#2}}
\newcommand{\paul}[2][says]{\remark{Paul}{#1}{#2}}

\def\micha#1{}
\def\haim#1{}
\def\liam#1{}
\def\paul#1{}
\def\wolfgang#1{}

\newtheorem{theorem}{Theorem}[section]
\newtheorem{lemma}[theorem]{Lemma}

\newtheorem{corollary}[theorem]{Corollary}

\newcommand{\etal}{et al.\xspace}

\newcommand{\textfrac}[2]{\IfSubStr{#1}{+}{\brk*{#1}}{\IfSubStr{#1}{-}{\brk*{#1}}{#1}}/#2}

\begin{document}

\title{Dynamic Planar Voronoi Diagrams for General
  Distance Functions and their Algorithmic Applications\thanks{%
A preliminary version appeared as H.\@ Kaplan, W.\@ Mulzer, 
L.\@ Roditty, P.\@ Seiferth, and M.\@ Sharir.
\emph{Dynamic Planar Voronoi Diagrams for General Distance 
Functions and their Algorithmic Applications}, 
Proc.\@ 28th SODA, pp. 2495--2504, 2017.
Work by Haim Kaplan, Wolfgang Mulzer, Liam Roditty, and Paul
Seiferth has been supported
by grants 1161/2011 and (with Micha Sharir) 
1367/2016 from the German-Israeli Science Foundation.
Work by Haim Kaplan has also been supported
by grants 822-10 and 1841-14 from the Israel Science Foundation,
and by the Israeli Centers for Research Excellence (I-CORE) program 
(center no.~4/11).
Work by Wolfgang Mulzer and Paul Seiferth has also been supported by
grant MU/3501/1 from Deutsche Forschungsgemeinschaft (DFG) and by
ERC StG 757609.
Work by Micha Sharir has been supported by
Grant 2012/229 from the U.S.-Israel Binational Science Foundation,
by Grants 892/13 and 260/18 from the Israel Science Foundation, by the Israeli 
Centers for Research Excellence (I-CORE) program (center no.~4/11), 
and by the Hermann Minkowski--MINERVA Center for Geometry at Tel 
Aviv University.
}}

\author{
Haim Kaplan\thanks{%
Blavatnik School of Computer Science,
Tel Aviv University, Tel~Aviv 69978, Israel;
\texttt{haimk@tau.ac.il}.}
\and
Wolfgang Mulzer\thanks{%
Institut f\"ur Informatik,
Freie Universit\"at Berlin, Berlin 14195, Germany;
\texttt{mulzer@inf.fu-berlin.de}.}
\and
Liam Roditty\thanks{%
Department of Computer Science,
Bar-Ilan University, Ramat Gan 52900, Israel;
\texttt{liamr@macs.biu.ac.il}.}
\and
Paul Seiferth\thanks{%
Institut f\"ur Informatik,
Freie Universit\"at Berlin, Berlin 14195, Germany;
\texttt{pseiferth@inf.fu-berlin.de}.}
\and
Micha Sharir\thanks{%
Blavatnik School of Computer Science,
Tel Aviv University, Tel~Aviv 69978, Israel;
\texttt{michas@tau.ac.il}.}
}

\maketitle

\begin{abstract}

We describe a new data structure for dynamic nearest 
neighbor queries in the plane with respect to a general 
family of distance functions. These include $L_p$-norms 
and additively weighted Euclidean distances. Our data
structure supports general (convex, pairwise disjoint) sites
that have constant description complexity (e.g., points, 
line segments, disks, etc.). Our structure uses $O(n \log^3 n)$ 
storage, 
and requires polylogarithmic update and query time, 
improving an earlier data structure of 
Agarwal, Efrat and Sharir that required $O(n^\eps)$ time for an
update and $O(\log n)$ time for a query [SICOMP, 1999].
Our data structure has numerous applications. In all of them, 
it gives faster algorithms, typically reducing an $O(n^\eps)$ 
factor in the previous bounds to polylogarithmic. 
In addition, we give here two new applications:
an efficient construction of a spanner in a disk intersection 
graph, and a data structure for efficient connectivity queries 
in a dynamic disk graph.

To obtain this data structure, we combine and extend various 
techniques from the literature. Along the way, we obtain several 
side results that are of independent interest.
Our data structure  depends on the existence and an
efficient construction of \emph{``vertical'' shallow cuttings}
in arrangements of bivariate algebraic functions. We prove that 
an appropriate level in an arrangement of a random sample 
of a suitable size provides such a cutting.
To compute it efficiently, we develop a randomized incremental 
construction algorithm for computing the lowest $k$ levels 
in an arrangement of bivariate algebraic functions
(we mostly consider here collections of functions whose lower 
envelope has linear complexity, as is the case in the dynamic 
nearest-neighbor context, under both types of norm).
To analyze this algorithm, we also improve a longstanding bound on 
the combinatorial complexity of the vertical decomposition of 
these levels. Finally, to obtain our structure, we combine 
our vertical shallow cutting construction with
Chan's algorithm for efficiently maintaining the
lower envelope of a dynamic set of planes in $\reals^3$.
Along the way, we also revisit Chan's technique and present
a variant that uses a single binary counter, with a
simpler analysis and improved amortized deletion time
(by a logarithmic factor; the insertion and query costs remain asymptotically the same).
\end{abstract}

\noindent
\emph{In loving memory of Ricky Pollack, one of the founding 
fathers of the field, and a dear friend.}

\section{Introduction}

Nearest neighbor searching in the plane is one 
of the most fundamental problems in computational 
geometry~\cite{dBCvKO}. Given a finite set $S$ of 
\emph{sites} in $\reals^2$, the goal is to construct 
a data structure that can find the 
``closest'' site for any given query object. If 
$S$ is fixed, Voronoi diagrams and 
their many variants provide a simple and 
well-understood solution~\cite{dBCvKO,vor-book}, 
with linear storage and logarithmic query time. 
However, in many applications, the set $S$ may 
change dynamically as sites get inserted and
deleted. Now, we want to answer nearest-neighbor 
queries interleaved with the updates. 
This setting is much less understood. 

If $S$ consists of singleton points and distances are 
measured in the Euclidean metric, we can achieve polylogarithmic 
update and query time~\cite{Cha10,Chan19}, with $O(n \log^3 n)$ 
storage. However, we are often confronted with more general 
distance functions (e.g., $L_p$-norms or additively 
weighted Euclidean distances). Examples
include the dynamic maintenance of a bichromatic closest pair of 
sites, constructing a Euclidean minimum-weight red-blue matching, 
constructing a Euclidean minimum spanning tree, computing the 
intersection of unit balls in three dimensions, or computing the 
smallest stabbing disk of a family of simply shaped compact 
strictly-convex sets in the plane; as well as computing a 
single-source shortest-path tree in a unit-disk graph (see 
Section~\ref{sec:application} for details and references). Despite 
the numerous motivating applications, there has been virtually no 
progress on the basic problem since the 1990s. The state of the 
art is work by Agarwal et al.~from 1999~\cite{AES}. 
It provides $O(n^\eps)$ update and $O(\log n)$ query time, 
for any fixed $\eps > 0$, while using $O(n^{1+\eps})$ 
storage.\footnote{%
Here and later, the constants in such bounds depend on $\eps$.}
We present a new solution that gives polylogarithmic update and 
query time, while using $O(n \log^3 n)$ storage,
for a wide range of distance functions. We 
assemble a broad set of techniques, such 
as randomized incremental construction, relative 
$(p, \eps)$-approximations, shallow cuttings for 
$xy$-monotone surfaces in $\reals^3$, and several advanced 
data structuring techniques.

We now describe our notions more thoroughly.
Let $S$ be a set of $n$ pairwise disjoint \emph{sites}. Each site is 
a simply-shaped compact convex region in the plane 
(points, line segments, disks, etc.). Let 
$\delta: \reals^2 \times \reals^2 \rightarrow \reals_{\geq 0}$ 
be a continuous \emph{distance function}
between points in the plane. For a site $s \in S$, 
define the \emph{distance to $s$},
$f_s: \reals^2 \rightarrow \reals_{\geq 0}$, as
$f_s(x,y) = \delta((x,y),s)= \min_{p \in s}\delta((x,y),p)$ 
(the minimum exists since $s$ is compact and 
$\delta$ is continuous). We 
assume that $\delta$ and the sites in $S$ have 
\emph{constant description complexity}. This means that they 
are defined by a constant number of polynomial 
equations and inequalities of constant maximum degree.
Set $F = \{f_s \mid s \in S\}$. The \emph{lower 
envelope} $\E_F$ of $F$ is the pointwise minimum 
$\E_F(x,y) = \min_{f \in F} f(x,y)$, and
its $xy$-projection is called the \emph{minimization 
diagram} of $F$, denoted by $\M_F$.
The \emph{combinatorial complexity} of $\E_F$ or of $\M_F$ 
is the total number of their vertices, edges and faces.
The book by Sharir and Agarwal~\cite{SA} provides a 
comprehensive treatment of these concepts.

Now, given a query point $q \in \reals^2$,
in order to find a $\delta$-nearest neighbor for $q$ in $S$,
we must identify a site $s$ with $\E_F(q) = f_s(q)$. 
This translates to a \emph{vertical ray shooting
query} in $\E_F$: find the intersection 
of $\E_F$ and the $z$-vertical line through $q$,
or, alternatively, locate $q$ in the planar map $\M_F$, 
where each two-dimensional face $\varphi \in \M_F$ is
labeled with the site $s$ for which $f_s$ attains 
the minimum over $\varphi$. (Edges and vertices can 
be labeled by the set of labels of their adjacent faces.)

The structure and the complexity of $\E_F$ and of 
$\M_F$, as well as algorithms for their construction
and manipulation, have been studied for 
several decades (again, see~\cite{SA}). To summarize,
under the above assumptions, the combinatorial 
complexity of $\E_F$ (or of $\M_F$)
is $O(n^{2+\eps})$, for any fixed $\eps>0$.\footnote{%
Again, the 
constant of proportionality depends on $\eps$.}
However, in many interesting cases, including the case where
the functions $f_s$ are linear (i.e., their graphs 
are non-vertical planes), the complexity of $\E_F$ 
is $O(n)$. The case of planes arises, 
after simple algebraic manipulations, for 
point sites under the Euclidean distance. Then
$\M_F$ is the Euclidean Voronoi diagram of $S$.
There are many variants of Voronoi diagrams, for 
other classes of sites and distance functions,
for which the complexity of $\E_F$ remains linear; 
see, e.g., the book by Aurenhammer, Klein, 
and Lee~\cite{vor-book}.

Coming back to nearest neighbor search, 
if we assume that $\E_F$ has linear complexity and 
can be constructed efficiently, all we need to do, 
in the so-called ``static'' 
case, is to preprocess $\M_F$ for fast planar point
location. Then, a query takes $O(\log n)$ time.
If sites in $S$ can be inserted or deleted, i.e., if $F$ 
changes dynamically, then $\E_F$ may change rather 
drastically after an update. 
Maintaining an explicit representation of $\M_F$ thus
becomes overwhelmingly expensive. Hence, the goal, 
in this paper and in earlier work, 
is to store an \emph{implicit} representation 
of $\E_F$ that still supports efficient 
vertical ray shooting in the current 
envelope $\E_F$ (or point location in the current $\M_F$).

In all the applications of dynamic nearest neighbor search that are
studied in this paper, the lower envelope $\E_F$ has
linear complexity. This is typically the case when $S$ consists
of point sites, and the distance functions 
are typically $L_p$-metrics, for $1 \leq p \leq \infty$, 
or additively weighted Euclidean metrics, where each 
point site $s \in S$ has a weight
$w_s \in \reals$, and $\delta(q,s) = |qs| + w_s$, where 
$|qs|$ is the Euclidean distance between $q$ and $s$.
See, e.g.,~\cite{Lee80,vor-book} for details concerning 
the linear complexity of $\E_F$ in these cases.
The lower envelope also has linear complexity
for general classes of 
pairwise-disjoint compact convex sites of
constant description complexity, and for fairly general metrics.

Our main result is an efficient data structure that 
dynamically maintains a set $F$ of bivariate functions
as above, under insertions and deletions 
of functions, and supports efficient vertical
ray shooting queries into $\E_F$. 
Assuming, as above, that the complexity of $\E_F$
is linear, the worst-case cost of a query, as 
well as the amortized expected cost of an update, is polylogarithmic,
and the storage used by the structure is $O(n \log^3 n)$.
As a consequence, we obtain faster 
solutions to all the applications mentioned above, and
many more, essentially reducing an $O(n^\eps)$ factor in the 
complexity to polylogarithmic.
Our results also generalize to the case where the lower 
envelope complexity is not linear.

\paragraph{A brief context.}
Suppose first that all 
functions in $F$ are linear. (As 
noted, this applies to point sites
in the Euclidean metric.) A classical 
solution for this case is due to Agarwal 
and Matou\v{s}ek~\cite{AM95}. They show how to 
maintain dynamically an implicit representation
of $\E_F$, 
with amortized update time $O(n^\eps)$, for
any fixed $\eps > 0$;
vertical ray shooting
queries take $O(\log n)$ worst-case time. Here, $n$
denotes an upper bound on $|F|$. The 
case of more general bivariate functions, as
described above, was studied by Agarwal~\etal~\cite{AES}.
If $\E_F$ has linear complexity, 
their technique has amortized 
update time $O(n^\eps)$, for any fixed $\eps > 0$,
and worst-case query time $O(\log n)$, matching
(asymptotically)  the result for planes~\cite{AM95}.

For more than ten years after the work of
Agarwal and Matou\v{s}ek~\cite{AM95}, it was open 
whether the $O(n^\eps)$ update time can be improved.
In SODA~2006, Chan~\cite{Cha10}
presented an ingenious construction for the case of planes where both 
the (amortized) update time and the (worst-case) query 
time are polylogarithmic.
More precisely, Chan's structure (combined 
with the recent deterministic construction of 
shallow cuttings by Chan and Tsakalidis~\cite{CT15}) 
supports insertions in $O(\log^3 n)$ amortized time, 
deletions in $O(\log^6 n)$ amortized time, and queries in 
$O(\log^2 n)$ worst-case time. 
However, before our work, it remained unknown whether a 
similar result (with polylogarithmic update and query time)
is possible for arbitrary bivariate functions 
with constant description complexity and linear 
envelope complexity. We 
provide an algorithm that 
meets all these performance goals.
Along the way, we also improve the deletion
time for Chan's data structure for planes by
a factor of $\log n$, and the bound of Agarwal~\etal~\cite{AES}
for the complexity of the vertical decomposition of the
$(\leq k)$-level in an arrangement of surfaces in $\reals^3$ by a 
factor of $k^\eps$.
Very recently, after the original submission of our paper, by combining
a faster cutting construction with our observations,
Chan~\cite{Chan19} further improved
the amortized deletion time for the case of
planes to $O(\log^4 n)$ and the amortized insertion time 
to $O(\log^2 n)$.

\section{Our results and techniques}

Our data structure combines a multitude of techniques. 
We first give a broad overview of how these 
techniques play together;
see Figure~\ref{fig:overview}, where the terms in the 
figure are explained below.

\begin{figure}[htb]
 \centering
 \includegraphics[scale=1.0]{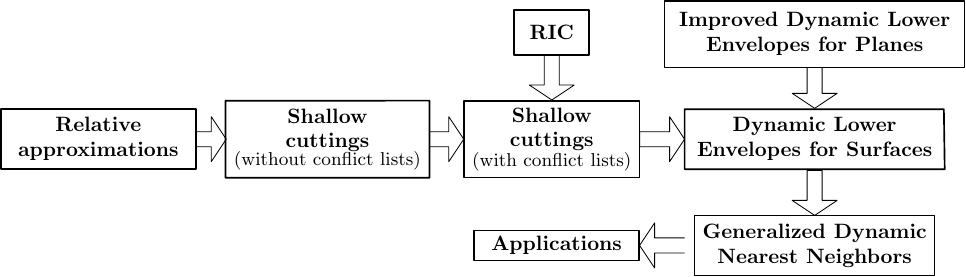}
\caption{The main tools used for 
  our data structure.}\label{fig:overview}
\end{figure}

Maybe the most crucial observation is that  the whole
geometric part in Chan's data structure~\cite{Cha10,Chan19} is in the
construction of small \emph{shallow cuttings} for 
planes. Thus, once we have 
an analogous result for surfaces, we can
maintain their dynamic lower envelope, 
or equivalently, solve the generalized dynamic planar nearest 
neighbor problem. It turns out that 
random sampling and the theory of
relative $(p, \eps)$-approximations easily yield a construction for 
the required cuttings.  However, we must also 
find the corresponding conflict lists quickly. 
For this, we present an algorithm that uses randomized 
incremental construction (RIC)  for the 
$(\leq k)$-level in an arrangement of surfaces. 
Together with an improved variant of Chan's result, 
this gives the generalized
nearest neighbor data structure. We show the impact
of this structure by presenting numerous applications thereof,
both old and new. In what follows, we describe the specific 
parts in more detail. 

The geometric core of Chan's data 
structure consists of an efficient construction of 
small-sized \emph{vertical shallow cuttings}~\cite{Cha05,CT15}. 
Let $F$ be a set of $n$ functions in $\reals^3$, 
identified in this paper with their graphs, and
let $\A(F)$ denote the \emph{arrangement} of $F$.
We recall the notion of a \emph{$k$-level} in $\A(F)$, for
a parameter $0 \leq k \leq n - 1$.
It is the closure of the set of points $q$ such that $q$ lies on some 
function graph and exactly $k$ graphs pass 
strictly below $q$.

Roughly speaking (more details follow below), for suitable 
parameters $k$ and $r \approx n/k$, a \emph{vertical 
$k$-shallow $(1/r)$-cutting} is a collection of pairwise 
openly disjoint semi-unbounded vertical prisms, 
where each prism consists of all points that lie 
vertically below some triangle. Furthermore, (i) these 
top triangles form a polyhedral terrain that is sandwiched  
between the $k$-level and the $k'$-level of the arrangement,
for a suitable parameter $k'$ close to $k$; (ii) the 
number of prisms is close to $O(r)$; and (iii) each 
prism is crossed by approximately $k$ function graphs.  

Once a fast construction of vertical 
shallow cuttings of sufficiently small size is available, 
we can plug it into Chan's machinery for planes~\cite{Cha10}, 
in almost 
black-box fashion. This gives a fast data structure for 
dynamic maintenance of the lower envelope in the general setting. 
Agarwal~\etal~\cite{AES} prove the existence of shallow 
cuttings of optimal size for general functions, but 
their cuttings are not ``vertical'', in the above sense, 
and a direct algorithmic implementation of their 
ideas yields an additional $O(n^\eps)$ factor for 
both the size and the construction time of the cutting.  
When applied to the dynamic maintenance 
problem, this gives (amortized) update cost 
$O(n^\eps)$ rather than polylogarithmic. Refining 
this bound is one of the main goals of the present paper. 

Thus, we design a different algorithm for
computing a vertical shallow cutting. For this,
we develop several technical results that 
we believe to be of independent interest.
We use \emph{relative approximations}~\cite{HPS} 
to show that, with high probability,
we get an \emph{$\eps$-approximation} of the $k$-level of $\A(F)$ 
by choosing a random sample $S_k$ of $(\textfrac{cn}{\eps^2k})\log n$ 
functions from $F$ and by taking the $t$-level of $\A(S_k)$, 
for 
$t \in \left[(1 + \textfrac{\eps}{3})\lambda,\; 
(1+\textfrac{\eps}{2})\lambda\right]$,
$\lambda = c\eps^{-2}\log n$, and  $c$ 
a suitable constant. This means that any such $t$-level  
of $\A(S_k)$ lies between levels $k$ and $(1+\eps)k$ of 
$\A(F)$.  We show that for random $t$, the expected complexity 
of the $t$-level is 
$O\left((\textfrac{n}{\eps^5 k}) \log^2 n \right)$.

Having computed such a $t$-level, we project it onto the 
$xy$-plane, construct the standard planar vertical 
decomposition of the faces of the projection, lift each 
resulting trapezoid $\varphi$ back to a trapezoidal 
subface $\varphi^*$ embedded in a surface on the original level, 
and associate it with the semi-unbounded vertical prism that extends
below $\varphi^*$. We show 
that this collection of prisms is a vertical $k$-shallow 
$(1/r)$-cutting in $\A(F)$ (with $k$ and $r \approx n/k$ 
as above). We denote it by $\Lambda_k$.

The last hurdle is to efficiently compute 
$\Lambda_k$, together with the \emph{conflict lists} of 
its prisms. The conflict list 
$\CL(\tau)$ of a prism $\tau \in \Lambda_k$ is the set of all 
functions $f \in F$ whose graphs cross the interior of $\tau$. 
(Although the construction of $\Lambda_k$ is performed with respect 
to the sample $S_k$, the 
conflict lists are defined with respect to the whole set $F$.)

This leads us to the classical problem of computing 
the $t$ lowest levels in an arrangement of $n$ 
bivariate functions of constant description complexity. 
A standard approach for this
goes via \emph{randomized incremental construction} (RIC), 
see, e.g.,~\cite{dBCvKO,mulmuley1994computational}. 
Here, one adds the functions one by one, in 
random order, while maintaining some representation of 
the first $t$ levels on the functions inserted so far. Following 
previous work, we maintain a cell decomposition of the region below 
the $t$-level of the function graphs inserted so far, and we 
associate with each cell a conflict list 
consisting of all the remaining functions that cross it.
If we run this process to completion, we get a 
suitable decomposition of the $t$ shallowest levels of 
the ``final'' $\A(F)$.  If we stop after inserting 
the first $(\textfrac{cn}{\eps^2k})\log n$ functions, 
which serve as the desired random sample 
$S_k$, we obtain, in addition to (a suitable decomposition 
of) the $t$ shallowest levels of $\A(S_k)$, the conflict 
lists of its cells (with respect to the whole $F$).

Our decomposition of choice is (a 
suitable shallow portion of) the standard \emph{vertical 
decomposition} of an arrangement of surfaces in $\reals^3$ 
(see~\cite{CEGS,SA} for details). Each prism 
extends between two consecutive levels of the 
current arrangement, so this decomposition differs from the 
vertical shallow cutting that we are after.\footnote{%
  We distinguish between the two kinds of cuttings by
  referring to the previous one as vertical, and to the
  one just introduced simply as a cutting.} 
Nevertheless, we show how to
transform this decomposition into a vertical shallow 
cutting, including the construction of the desired
conflict lists of its semi-unbounded prisms. 
A fairly intricate analysis shows that the shallow cutting 
has expected complexity $O\left((\textfrac{n}{k}) \log^2 n \right)$. 

The implementation of such a RIC for the shallowest $t$ levels of 
$\A(F)$ is far from trivial. It has been considered before for the 
case of planes.  Mulmuley~\cite{Mulmuley91} described a RIC of the 
first $t$ levels, when the lower envelope of the planes corresponds 
to the Voronoi diagram of a set of points in the $xy$-plane 
(under the standard algebraic manipulations alluded to 
above). Mulmuley's procedure needs $O(nt^2\log(n/t))$ 
expected time.\footnote{$O(nt^2)$ is a tight bound on the
complexity of the $t$ shallowest levels in an arrangement
of $n$ planes.}  Agarwal~\etal~\cite{Agarwal_etal98} used 
a somewhat less standard randomized incremental algorithm
and obtained a bound of $O(n\log^3 n + nt^2)$ expected 
time. Their algorithm works for any set of planes.
It maintains a point $p$ in each prism, such that the level of $p$ in 
$\A(F)$ is known, and it uses this information to prune 
away prisms that can be ascertained not to intersect the 
shallowest $t$ levels of $\A(F)$. Finally, 
Chan~\cite{Chan2000} obtained a bound of $O(n\log n + nt^2)$ 
expected time with an algorithm that can be viewed as 
a batched randomized incremental construction.
Unfortunately, it is not clear how to apply some 
crucial components of these algorithms when $F$ is 
a set of nonlinear functions.

We present and analyze a standard randomized incremental
construction algorithm for the shallowest $t$ levels of 
an arrangement $\A(F)$ of a set $F$ of $n$ bivariate functions with 
constant description complexity and linear envelope 
complexity. Our algorithm runs in 
$O\left( nt \log(n/t) \log n \,\lambda_s(t)\right)$ 
expected time, where $s$ is a constant that depends on
the surfaces and $\lambda_s(t)$ is the maximum length of
a Davenport-Schinzel sequence on $t$ symbols of order 
$s$~\cite{SA}.\footnote{As is well known~\cite{SA}, the function 
$\lambda_s(t)$ is ``almost'' linear, i.e.,
$\lambda_s(t) = t \beta_s(t)$ for some extremely slow-growing
function $\beta_s(t)$ of inverse-Ackermann type.}
To get this result, we improve a bound of Agarwal~\etal~\cite{AES} 
on the complexity of the vertical 
decomposition of the $t$ shallowest levels in $\A(F)$.
Agarwal~\etal~proved that this 
complexity is $O(nt^{2+\eps})$, for any fixed $\eps > 0$, via 
a fairly complicated charging scheme. We improve 
this to $O(nt\lambda_s(t))$, with a simpler argument,
where $s$ is a constant that depends on the algebraic 
complexity of the functions of $F$ (see below for a precise 
definition).

Using our randomized incremental algorithm, we construct 
a vertical shallow cutting of the first $k$ levels in 
$\A(F)$, consisting of $O\left((\textfrac{n}{k}) \log^2 n\right)$ 
prisms, each with a conflict list of size $O(k)$. The 
construction time is $O(n\log^3 n\,\lambda_s(\log n))$.

Once we have an efficient mechanism for constructing
vertical shallow cuttings, we apply it, following and adapting the
technique of Chan for the case of planes, to obtain our dynamic
data structure. Before that, we re-examine Chan's data structure, 
and we present it in a way that is easier to understand (in our 
opinion) and, at the same time, slightly faster than the original
version.  Our variant follows a standard route:
we begin with a static data structure and extend it for 
insertions, using a (somewhat non-standard) variant of the well-known
Bentley-Saxe binary counter technique~\cite{BS}. Then, 
we show how to perform deletions via re-insertions
of planes, using a \emph{deletion lookahead 
mechanism}, the major innovation in Chan's work.
We believe that our analysis sheds additional 
light on the inner workings of Chan's structure.
We improve the amortized deletion time
to $O(\log^5 n)$, i.e., by a logarithmic factor.
Deletions are the costliest operations in Chan's structure 
and constitute the bottleneck in most of its applications.
As mentioned, in recent work, Chan~\cite{Chan19}
achieved a further improvement, building upon our analysis, 
reducing the amortized
deletion time to $O(\log^4 n)$ and the amortized insertion time 
to $O(\log^2 n)$.

We finally combine our shallow cutting construction with 
our improved version of Chan's data structure, extended to more
general functions, to obtain
a dynamic data structure for vertical ray shooting into 
the lower envelope of a dynamically changing set of bivariate 
functions, as above. Our (worst-case, deterministic) query time is 
$O(\log^2 n)$, the (amortized, expected)
time for an insertion is $O(\log^5 n \,\lambda_s(\log n))$,
and the (amortized, expected) time for a deletion is 
$O(\log^9 n\,\lambda_s(\log n))$.  
The larger polylogarithmic factors are a consequence of slightly
weaker bounds on the complexity of an approximating level.

Plugging our new bounds into the
applications in Agarwal~\etal~\cite{AES} and in 
Chan~\cite{Cha10}, we immediately improve several 
running times, replacing a factor of $n^\eps$ by a polylogarithmic 
factor.  Some prominent examples are shown in the following table; 
details follow in Section~\ref{sec:application}. 
(Constants of proportionality are suppressed in the table.)
The parameter $s$ depends on the precise metric, 
and is defined in more detail later in the paper. 
Concrete values of $s$ are given in the table for the 
specific respective applications.

\begin{center}
\begin{tabular}{| c | c | c | }
\hline
  Problem & Old Bound & New Bound \\
\hline
\hline
\parbox{5cm}{Dynamic bichromatic closest pair
in general planar metric} & 
\parbox[c]{4.5cm}{$n^\eps$ update~\cite{AES}} & 
\parbox[c][1.2cm]{5cm}{$\log^{5} n \,\lambda_s(\log n)$ insertion,\\
$\log^{9} n \,\lambda_s(\log n)$ deletion} \\
\hline

\parbox{5cm}{Minimum planar bichromatic Euclidean 
matching} &
\parbox[c]{4.5cm}{$n^{2+\eps}$~\cite{AES}} & 
\parbox[c][1.2cm]{5cm}{$n^2 \log^{9} n 
  \,\lambda_\swed(\log n)$}\\
\hline

\parbox{5cm}{Dynamic minimum spanning tree 
in $L_p$-metric} & 
\parbox[c]{4.5cm}{$n^\eps$ update~\cite{AES}} &
\parbox[c][1.2cm]{5cm}{$\log^{11} n 
  \,\lambda_s(\log n)$ update}\\ 

\hline
\parbox{5.2cm}{Dynamic intersection of unit balls in $\reals^3$} & 
\parbox[c]{4.5cm}{$n^\eps$ update~\cite{AES} \\ 
queries in $\log n$ and $\log^4 n$\\
(depending on the precise query)} & 
\parbox[c][1.9cm]{5cm}{%
$\log^5 n \,\lambda_\sballs(\log n)$ insertion,\\ 
$\log^{9} n \,\lambda_\sballs(\log n)$ deletion, \\
queries in $\log^{2} n$ and $\log^5 n$ \\
(depending on precise query)}\\
\hline
\end{tabular}
\end{center}

A particularly fruitful application domain for our data 
structure can be found in \emph{disk intersection graphs}.
These are defined as follows: Let $S \subset \reals^2$
be a finite set of point sites, each with an associated
weight $w_p > 0$, $p \in S$; a site $p$ with weight $w_p$
represents the disk of radius $w_p$ centered at $p$. 
The \emph{disk intersection graph} for $S$, denoted $D(S)$, 
has the sites in $S$ as vertices,
and there is an edge $pq$ between two sites $p$, $q$ in $S$ if and
only if $|pq| \leq w_p + w_q$, i.e., if the disk around 
$p$ with radius $w_p$ intersects the disk around
$q$ with radius $w_q$. If all weights are $1$, we call
$D(S)$ the \emph{unit disk graph} for $S$.
Disk intersection graphs are a popular model for 
geometrically defined graphs and networks, and enjoy an increasing
interest in the research community, in particular due to applications 
in wireless
sensor networks~\cite{RodittySe11,CabelloJejcic15,FurerKa12,
ChanPaRo11,KaplanMuRoSe18,KaplanEtAl15}. 
The following table gives an overview of our results on disk graphs.

\begin{center}
\begin{tabular}{|c | c| c| }
\hline
  Problem & Old Bound & New Bound \\
\hline
\hline
\parbox{5cm}{Shortest-path tree in a unit disk graph } & 
\parbox[c]{4.5cm}{$n^{1+\eps}$~\cite{CabelloJejcic15}} & 
\parbox[c][1.2cm]{5cm}{$n\log^{9} n \,\lambda_\swed(\log n)$} \\
\hline

\parbox{5cm}{Dynamic connectivity in disk intersection 
graphs with radii in $[1,\Psi]$} &
\parbox[c]{4.5cm}{$n^{20/21}$ update \\
$n^{1/7}$ query~\cite{ChanPaRo11}} &
\parbox[c][1.2cm]{5cm}{$\Psi^2 \log^{9} n\, \lambda_\swed(\log n)$ 
update\\ 
$\log n/\log\log n$ query}\\
\hline

\parbox{5cm}{BFS tree in disk intersection graphs} &
\parbox[c]{4.5cm}{$n^{1+\eps}$~\cite{RodittySe11}} & 
\parbox[c][1.2cm]{5cm}{$n\log^9 n\, \lambda_\swed(\log n)$}\\
\hline

\parbox{5cm}{$(1+\rho)$-spanners for disk intersection graphs} &
\parbox[c]{4.5cm}{$n^{4/3+\eps}\rho^{-4/3}\log^{2/3} 
\Psi$~\cite{FurerKa12}} & 
\parbox[c][1.2cm]{5cm}{$n\rho^{-2}\log^{9}n\,  
\lambda_\swed(\log n)$} \\
\hline
\end{tabular}
\end{center}
Two of the applications listed above concern finding shortest-path 
trees in unit disk graphs, and BFS-trees in disk intersection 
graphs. Our new structures give improved bounds almost 
in a black-box fashion, using the respective techniques 
of Cabello and Jej\^ci\^c~\cite{CabelloJejcic15} 
and of Roditty and Segal~\cite{RodittySe11}.
Very recently, Wang and Xue~\cite{WangXu19} presented
a deterministic algorithm to find the shortest-path tree
in a unit disk graph in $O(n \log^2 n)$ time.
The other two applications are a bit more involved.  
First, we give a data structure 
for the dynamic maintenance of the connected components 
in a disk intersection graph, as disks are
inserted or deleted, where we assume that all disks 
have radii from the interval $[1, \Psi]$.
Then, we can apply our data structure in a 
grid-based approach that gives an update time that 
depends on $\Psi$ and is polylogarithmic if
$\Psi$ is constant. The previous bound of Chan, 
P{\v{a}}tra{\c{s}}cu, and Roditty~\cite{ChanPaRo11} 
is only slightly sublinear (albeit independent of $\Psi$). 
Queries are faster in both approaches, but the bound 
in~\cite{ChanPaRo11} is a power of $n$ whereas here it is only 
sub-logarithmic.  Very recently, Kauer and Mulzer~\cite{KauerMu19} 
presented method that improves the dependence on $\Psi$.
Finally, we give an algorithm for computing a $(1 + \rho)$-spanner 
in a disk intersection graph, for any $\rho > 0$. 
A \emph{$(1 + \rho)$-spanner} for $D(S)$ is a subgraph
$H$ of $D(S)$ such that the shortest path distances in $H$ 
approximate the shortest path distances in $D(S)$ up to a factor of 
$(1 + \rho)$. The previous construction by F\"urer and 
Kasiviswanathan~\cite{FurerKa12} has a running time that depends on 
the radius ratio $\Psi$, as defined above.
Our new algorithm is independent of $\Psi$ and achieves 
almost linear running time, improving the previous algorithm 
by a factor of at least $n^{1/3}$.

\paragraph{Paper outline.}
Section~\ref{sec:prelim} gives further background and precise 
definitions. In Section~\ref{sec:levelapproximation}, we describe 
how to obtain a terrain that approximates the $k$-level 
of $\A(F)$ by random sampling, via 
\emph{relative $(p,\eps)$-approximations} (see~\cite{HPS} and below).
In Section~\ref{sec:cutting}, we define a 
\emph{vertical shallow cutting}, based on our level approximation,
and show how to compute it with a randomized 
incremental construction of the shallowest $t$ 
levels in $\A(F)$. In Section~\ref{sec:ric}, we 
describe in detail the randomized incremental construction 
and analyze it. Section~\ref{sec:chan} gives our 
improved variant of Chan's structure for maintaining the
lower envelope of planes. Combining our 
cuttings with Chan's machinery as presented in 
Section~\ref{sec:chan}, we obtain, in 
Section~\ref{sec:connect}, an efficient procedure 
for dynamically maintaining the lower envelope of a 
collection of algebraic surfaces of constant 
description complexity and linear lower envelope complexity. 
Finally, in Section~\ref{sec:application}, we present several 
known applications, for which we obtain 
better bounds, and our new applications for disk graphs.
Along the way, we also consider the case where the lower 
envelope complexity of $F$ is superlinear, and extend our 
analysis to this more general setup.

\section{Preliminaries}\label{sec:prelim}

Let $\F$ be a family of bivariate functions 
$f: \reals^2 \rightarrow \reals$, and let $F$ be a finite subset of 
$\F$. Throughout the paper, we assume that the 
functions in $\F$ are continuous, totally 
defined, and algebraic, and that they have 
\emph{constant description complexity}. This 
means that the graph of each function is a 
semialgebraic set, defined by a constant number 
of polynomial equalities and inequalities of 
constant maximum degree. 
We will generally make no distinction between 
a function $f \in \F$ and its graph 
$\{(x,y,f(x,y)) \mid x, y \in \reals\}$, which is 
a continuous $xy$-monotone surface, also called a \emph{terrain}.

The \emph{lower envelope} $\E_F$ of $F$ is the graph of the 
pointwise minimum of the functions of $F$. The $xy$-projection of 
$\E_F$ yields a subdivision of the $xy$-plane called the 
\emph{minimization diagram} $\M_F$ of $F$.  It can be represented 
by a standard doubly connected edge list (DCEL, see, 
e.g.,~\cite{dBCvKO}). Each (two-dimensional) face $f$ of $\M_F$ 
is labeled by the function in $F$ that attains the pointwise 
minimum over $f$.

When $\M_F$ consists of $O(|F|)$ faces, vertices, 
and edges, for any finite $F \subseteq \F$, we say 
that $\F$ has lower envelopes of \emph{linear 
complexity}. We will mostly assume this to be the case for
the families $\F$ considered here.
In particular, this assumption holds when $\F$ is the family of
all nonvertical planes, or when $\F$ is a family 
of distance functions under some metric (or under some so-called 
\emph{convex distance function}~\cite{ChewDr85}), each of which 
measures the distance of a point in the $xy$-plane to some 
given site (cf.~Section~\ref{sec:application}), where
the sites are assumed to be pairwise disjoint closed convex sets.

For simplicity, we also assume that $F$ is in 
\emph{general position}, i.e., no more than three 
functions meet at a common point, no more than 
two functions meet in a one-dimensional curve, 
no pair of functions are tangent to each other, and
no function is tangent to the intersection curve of two other 
functions. For example, this holds if the coefficients of the 
polynomials defining the functions in $F$ are 
algebraically independent over $\reals$~\cite{SA}.
Furthermore, we assume that the coordinate frame is 
generic, so that the $xy$-projections of the 
intersection curves of any pair of functions in $F$
are also in general position, defined in an analogous sense.

\paragraph{Model of computation.}
We assume a (by now fairly standard) algebraic 
model of computation, in which primitive operations 
that involve a constant number of functions of $\F$ 
take constant time. Such operations include: computing 
the intersection points of any three functions, 
computing (a suitable representation of) the intersection curve 
of any two functions,
decomposing it into connected components, finding a 
representative point on each such component, computing 
the points of intersection between the $xy$-projections
of two intersection curves, testing whether a point 
lies below, on, or above a function graph, and so on. 
This model is reasonable, because there are standard 
techniques in computational algebra (see, 
e.g.,~\cite{BPR,SS2}), and actual packages (such as 
the one described by Boissonnat and Teillaud~\cite{ECG}), that 
perform such 
operations exactly in constant time. Technically, 
these methods and packages determine the truth 
value of any Boolean predicate of constant description 
complexity. That is, they are not expected to provide 
precise values of roots of polynomial equations, but 
they can determine, exactly and in constant time, 
any algebraic relation between such roots and/or similar 
entities, expressed by a constant number of polynomial 
equations and inequalities of constant maximum degree.

\paragraph{Shallow cuttings.}
Let $\A(F)$ be the arrangement of a set $F$ of $n$ bivariate 
functions from $\F$ in $\reals^3$. The 
\emph{level} of a point $q\in\reals^3$ in 
$\A(F)$ is the number of functions of $F$ that 
pass strictly below $q$. For $0\le k \le n-1$,
the \emph{$k$-level} $L_k(F)$ of $\A(F)$ is the 
closure of the set of all points at level $k$ that lie 
on a function in $F$.
We denote by $L_{\leq k}(F)$ the union of the first 
$k$ levels of $\A(F)$. For given parameters 
$k \in \{0, \dots, n - 1\}$, $r \in \{1, \dots, n\}$,
a \emph{$k$-shallow $(1/r)$-cutting} in $\A(F)$ is
a collection $\Lambda$ of pairwise openly disjoint 
regions $\tau$, each of constant description complexity, 
so that the union of all $\tau \in \Lambda$ covers $L_{\le k}(F)$,
and so that the interior of each $\tau \in \Lambda$ 
is intersected by at most $n/r$ functions in $F$. 
The \emph{size} of $\Lambda$ is the number of regions in $\Lambda$.

In addition, we call $\Lambda$ \emph{vertical} if
every region $\tau \in \Lambda$ is a (semi-unbounded) 
\emph{pseudo-prism}. A pseudo-prism $\tau$ of this kind consists of 
all points that lie vertically below some \emph{pseudo-trapezoid} 
$\otau$ on a function $f \in F$. Such a pseudo-trapezoid is defined 
as the set 
\[
\otau = \{(x,y,f(x,y)) 
\mid x^-\le x\le x^+, \psi^-(x) \le y\le \psi^+(x)\} ,
\]
for real numbers $x^- < x^+$ and (semi-)algebraic functions 
$\psi^-$, $\psi^+: \reals \rightarrow \reals$ of constant 
description complexity; some of these boundary constraints might be 
absent. For planes, $\otau$ will simply be a planar $y$-vertical 
trapezoid, and we do not insist that $\otau$ be contained in one of 
the input planes. Since the interior of a pseudo-prism in a 
vertical $k$-shallow $(1/r)$-cutting is intersected by at least 
$k$ functions, we must have $k \leq n/r$. In our setting, we will 
set $r = \Theta(n/k)$, with a sufficiently small constant of 
proportionality, which is the case most relevant 
for all our applications.

Matou{\v s}ek~\cite{Ma:rph} proved 
that for any $n$ hyperplanes in $\reals^d$,
there is a $k$-shallow $(1/r)$-cutting of size 
$O\big({q^{\lceil d/2\rceil}r^{\lfloor d/2\rfloor}}\big)$, 
for $q = k(r/n)+1$. For the most relevant case 
$k=\Theta(n/r)$, we get $q=O(1)$ and a cutting of size
$O\big(r^{\lfloor d/2\rfloor}\big)$, 
a significant improvement over the general bound 
$O(r^d)$ for a cutting that covers all of $\A(F)$,
rather than just $L_{\le k}(F)$~\cite{Ch}.\footnote{%
  These bounds are not known for general surfaces.}
For example, for planes in three dimensions, we get cuttings of size
$O(r)$ instead of $O(r^3)$.
This has led to improved solutions of many range
searching and related problems (see, e.g.,~\cite{CT15}
and the references therein).
Matou{\v s}ek~\cite{Ma:rph} presented a deterministic 
algorithm to construct a shallow cutting in polynomial 
time. If $r < n^\delta$, the running time improves to $O(n \log r)$,
where $\delta > 0$ is a sufficiently small constant
depending on $d$.  Later,
Ramos~\cite{Ram99} presented a (rather complicated) 
randomized algorithm for $d = 2,3$, that constructs 
a hierarchy of shallow cuttings for a geometric 
sequence of $O(\log n)$ values of $r$ (and $k=\Theta(n/r)$), 
in $O(n\log n)$ overall expected time.  Recently, 
Chan and Tsakalidis~\cite{CT15} gave a
\emph{deterministic} algorithm for the same task.
Their algorithm can be stopped early, to obtain
an $O(n/r)$-shallow $(1/r)$-cutting in $O(n \log r)$ time.
Interestingly, the analysis of Chan and Tsakalidis uses 
Matou{\v s}ek's theorem on the existence of an $O(n/r)$-shallow 
$(1/r)$-cutting of size $O(r)$ as a black box.

Chan~\cite{Cha05} was the first to point out the 
existence of vertical shallow cuttings for planes 
in $\reals^3$. Such a cutting originates
from a polyhedral triangulated $xy$-monotone terrain
that lies entirely above $L_{\leq k}(F)$,
so that each triangle $\otau$ of the 
terrain generates a semi-unbounded triangular prism $\tau$
with $\otau$ as its top face. These shallow cuttings 
have many applications, in particular in Chan's data structure
for dynamic maintenance of lower envelopes~\cite{Cha10,Chan19}, 
as reviewed above.
The deterministic construction of Chan and 
Tsakalidis~\cite{CT15} constructs vertical shallow 
cuttings. Recently, Har-Peled, Kaplan, and Sharir~\cite{HKS}
gave an alternative construction with additional 
favorable properties.\footnote{%
One significant difference is that the ``top terrain'' 
in~\cite{HKS} approximates the corresponding level 
$k$ up to any specified accuracy, whereas the structure 
in~\cite{CT15} does not.}

Things become technically more involved when we 
allow more general algebraic functions.
For example, decomposing cells of the arrangement into subcells 
of constant description complexity is easy for 
hyperplanes (where the subcells are simplices), using, 
e.g., a \emph{bottom-vertex 
triangulation}~\cite{Chazelle00,Clarkson88}.  For general curves or 
surfaces, the only known general-purpose cell 
decomposition technique is \emph{vertical 
decomposition}~\cite{CEGS,SA}. In the plane, the 
complexity of such a decomposition is proportional to 
the complexity of the original arrangement, and 
in three and four dimensions, near (but not quite) 
optimal upper bounds are known~\cite{CEGS,Kol}. However,
in dimension five and higher, there are 
significant gaps between the known upper and lower 
bounds~\cite{CEGS}. Regarding shallow cuttings for general surfaces, 
we are aware only of the aforementioned result of 
Agarwal~\etal~\cite{AES}, and of no work that considers 
\emph{vertical} shallow cuttings for this general setup.

\section{Approximate $k$-levels}\label{sec:levelapproximation}

Here and in the following section, we 
show the existence of vertical shallow cuttings 
for surfaces. Later, we address the issue of how to 
efficiently compute these cuttings and the conflict lists
of their pseudo-prisms.

Let $\F$ be a family of functions as in Section~\ref{sec:prelim}, 
and let $F$ be a collection of $n$ functions from $\F$. 
Recall that we assume that the lower envelope complexity of
$\F$ is linear. Agarwal~\etal~\cite{AES} provide 
a shallow cutting for $\A(F)$, and show that, for any fixed 
$\eps > 0$, $0\le k \le n-1$, and $1\le r \le n$, 
there is a $k$-shallow $(1/r)$-cutting of size $O(q^{2+\eps}r)$, 
where $q=kr/n+1$ (and the constant of proportionality depends on 
$\eps$).  This is slightly sub-optimal when $q$ is large. However, 
we are interested in the case $r\approx n/k$, so $q=O(1)$ and the 
bound becomes $O(r)$, which is optimal.  Nonetheless, 
the cutting of Agarwal~\etal~\cite{AES} is not vertical, 
and is therefore useless for our purposes.

Known techniques for computing vertical shallow 
cuttings for planes, and the conflict lists
of their prisms~\cite{CT15,HKS},
crucially rely on the fact that if a plane
intersects a semi-unbounded prism $\tau$, it must also intersect 
a vertical edge of $\tau$. This is not true
for general functions. Thus, we 
use a somewhat different 
approach that results in cuttings whose size
is a (small) polylogarithmic factor off optimal.
It is an interesting challenge to tighten the bound. 
For the time being, though, we are not aware
of any alternative construction
that meets our specific needs.

Let $0 < \eps \le 1/2$ be a specified error parameter,\footnote{%
  If $\eps$ is constant (say, $\eps=1/2$), the dependence 
  on $\eps$ can be suppressed. Nonetheless, we include it in the 
  interest of precision, and in anticipation of future 
  applications that might require closer level approximations.}
and let $0\le k \le n-1$. We will approximate 
the level $L_{k}(F)$ of $\A(F)$ by a terrain 
$\oT_{k}$ (which will actually be a level in 
the arrangement of some sample of $F$), with the following properties
(for simplicity, we ignore in this paper the trivial issue of 
rounding).
\begin{enumerate}
\item\label{prop:1}
$\oT_{k}$ fully lies above $L_k(F)$ and below 
$L_{(1+\eps) k}(F)$.
\item\label{prop:2}
The complexity $|\oT_{k}|$ of $\oT_{k}$
is $O\left((\textfrac{n}{\eps^5 k})\log^2 n \right)$.
\end{enumerate}
To construct $\oT_{k}$, we use the notion of 
\emph{relative $(p,\eps)$-approximation} (see 
Har-Peled and Sharir~\cite{HPS} for more details): for a range 
space $(X,\R)$ of finite VC-dimension, and for 
given parameters $p,\eps \in (0,1)$, a set $A\subseteq X$ 
is called a relative $(p,\eps)$-approximation, 
if, for each range $R\in\R$, we have
\begin{equation} \label{rel-apx}
\left| \frac{|R \cap X|}{|X|} - \frac{|R\cap A|}{|A|} \right| \le
\begin{cases}
\eps\, \frac{|R|}{|X|}, & \text{if $|R|\ge p|X|$} \\
\eps p, & \text{if $|R| < p|X|$}.
\end{cases}
\end{equation}
As shown by Har-Peled and Sharir~\cite{HPS} (following 
Li~\etal~\cite{LLS}, see also Har-Peled's 
book~\cite{sariels-book}), for any $q \in (0,1)$, 
a random sample of size
\[
O\left( \frac{1}{\eps^2p} 
\left( \log \frac{1}{p} + \log \frac{1}{q} \right) \right) 
\]
is a relative $(p, \eps)$-approximation with 
probability at least $1-q$,
where the constant of proportionality 
depends linearly on the VC-dimension, 
but is independent of $\eps$ and $p$.

We apply this general machinery to the range 
space $(\F,\R)$ defined as follows. 
An \emph{object} $o$ can be either 
a straight line, a segment, a ray, or an edge in the arrangement of 
a constant number of functions of 
$\F$, a face in such an arrangement,  
a connected portion of such a face cut off 
by vertical planes orthogonal to the $x$-axis, 
or a connected component of the intersection 
of such a face with a plane orthogonal to the $x$-axis. 
Each range $R \in \R$ corresponds to an object $o$ as above, 
and is the set of functions of $\F$ that intersect $o$.
The fact that $(\F,\R)$ has finite 
VC-dimension, follows by standard
arguments (see, e.g.,~\cite{SA,sariels-book}).
Thus, let $S_k\subseteq F$ be a random sample of size 
\begin{equation} \label{eq:rk}
r_k = \frac{cn}{\eps^2k} \log n ,
\end{equation}
where $c > 0$ is a suitable constant, proportional to 
the VC-dimension of $(\F,\R)$. By the previous discussion, 
we can chose $c$ such that $S_k$ is a relative 
$\left(\textfrac{k}{2n},\textfrac{\eps}{3}\right)$-approximation 
for $(\F,\R)$, with probability at least $1-1/n^b$, for some
sufficiently large constant $b \geq 1$.
Note that for this choice of $r_k$ to make 
sense, we need $k = \Omega\left( \eps^{-2}\log n\right)$. 
The case of smaller $k$ is simpler and will be treated below.

Set $\oT_k$ to a \emph{random} level $L_t(S_k)$, 
where $t$ is chosen uniformly in the range 
\[
\left[\Big(1+\frac{\eps}{3}\Big)\lambda,
  \;\Big(1+\frac{\eps}{2}\Big)\lambda\right] , 
\;\text{for}\; \lambda = c\eps^{-2} \log n. 
\]
We refer to $\oT_k$ as an \emph{$\eps$-approximation}
to level $L_k(F)$. This terminology is justified
in the following lemmas.  From now on, we will assume
that $S_k$ is indeed a relative 
$\left(\textfrac{k}{2n}, \textfrac{\eps}{3}\right)$-approximation.
The bounds in Lemmas~\ref{lem:prop2x} 
and~\ref{lem:prop2} hold notwithstanding, 
since the assumption fails with probability at most $1/n^b$, 
so that, by making $b$ sufficiently large (as we can),
the event of failure contributes a negligible amount to
the relevant expectation.

\begin{lemma}\label{lem:prop1}
The terrain $\oT_k$ lies between levels $k$ and 
$(1+\eps)k$ of $\A(F)$.
\end{lemma}
\begin{proof}
Let $p$ be a point of level $k$ in $\A(F)$, and 
let $R^{(p)}$ denote the range of those functions 
that pass below $p$. 
By assumption, $S_k$ is a relative 
$\left(\textfrac{k}{2n},\textfrac{\eps}{3}\right)$-approximation
for a range space that includes $R^{(p)}$. Since
\[
\frac{k}{2n} < \frac{k}{n} = \frac{|R^{(p)}|}{n} ,
\]
the first case in (\ref{rel-apx}) implies that
\[
|R^{(p)}\cap S_k| \le \left(1+\frac{\eps}{3}\right) \frac{k}{n} r_k 
\le \left(1+\frac{\eps}{3}\right)\lambda.
\]
Thus, at most 
$\left(1+\frac{\eps}{3}\right)\lambda$ functions 
of $S_k$ pass below $p$, i.e.,
$p$ lies on or below $\oT_k$. Similarly, let 
$q$ be a point of level $(1+\eps)k$ in $\A(F)$. 
By a symmetric argument, at least 
\[
\left(1-\frac{\eps}{3}\right)(1+\eps) \frac{k}{n} r_k \ge 
\left(1+\frac{\eps}{2}\right)\lambda
\]
functions of $S_k$ pass below
$q$ (using $\eps \leq 1/2$). Hence, $q$ lies on or above $\oT_k$, and 
the lemma follows.
\end{proof}

\begin{lemma}\label{lem:prop2x}
The expected number of vertices $p$ of $\A(S_k)$
whose level in $\A(S_k)$ is 
between $\left(1+\textfrac{\eps}{3}\right)\lambda$ 
and
$\left(1+\textfrac{\eps}{2}\right)\lambda$ is
$O\left((\textfrac{n}{\eps^6 k}) \log^3 n \right)$.
\end{lemma}
\begin{proof}
Let $p$ be a vertex of $\A(F)$, and let $\ell_{S_k}(p)$ 
denote the level of $p$ in $\A(S_k)$.
As follows from the proof of Lemma~\ref{lem:prop1}, 
the vertex $p$ can satisfy
$\left(1+\textfrac{\eps}{3}\right)\lambda \le \ell_{S_k}(p) \le 
\left(1+\textfrac{\eps}{2}\right)\lambda$ only if the
level of $p$ in $\A(F)$ lies
lies between $k$ and $(1+\eps)k$.
The probability that $p$ 
shows up in $\A(S_k)$ is\footnote{%
Here we use the model where we sample a subset of 
the prescribed size, where all such subsets are 
equally likely to be drawn. One could also use 
an alternative common model, in which each 
function is independently chosen to be in $S_k$ 
with probability $r_k/n$.  The calculations are 
slightly different in the latter model, but they 
lead to the same conclusions and asymptotic 
bounds.}
\[
\frac{\binom{n-3}{r_k-3}}{\binom{n}{r_k}} 
\approx \left( \frac{r_k}{n} \right)^3 
= O\left( \frac{\log^3 n}{\eps^6k^3} \right).
\]
As shown by Clarkson and Shor~\cite{CS89}, 
there are $O(n((1+\eps)k)^2) = O(nk^2)$ 
vertices in $L_{\leq(1+\eps)k}(F)$.
Hence, the expected number of vertices $p$ 
of $\A(S_k)$ with
$\left(1+\textfrac{\eps}{2}\right)\lambda \le \ell_{S_k}(p) \le 
\left(1+\textfrac{\eps}{2}\right)\lambda$ is at most
$O\left((\textfrac{n}{\eps^6 k}) \log^3 n \right)$, as claimed.
\end{proof}

\begin{lemma}\label{lem:prop2}
The expected complexity of $\oT_k$, over the random choices 
of $S_k$ and of the level in
$\left[\left(1+\textfrac{\eps}{3}\right)\lambda,\;
\left(1+\textfrac{\eps}{2}\right)\lambda \right]$, is 
\begin{equation} \label{lk}
O\left(\frac{n}{\eps^5 k} \log^2 n \right).
\end{equation}
\end{lemma}
\begin{proof}
Since a level in $\A(S_k)$ is an $xy$-monotone terrain,
and since each vertex of $\A(S_k)$ appears in only three
(consecutive) levels, the sum of the complexities of 
all the $L_j(S_k)$, for 
$j\in \left[ \left(1+\frac{\eps}{3}\right)\lambda,\; 
\left(1+\frac{\eps}{2}\right)\lambda \right]$,
is proportional to the number of vertices in
$\A(S_k)$ with level between 
$\left(1+\frac{\eps}{3}\right)\lambda$ 
and
$\left(1+\frac{\eps}{2}\right)\lambda$.
Thus, by Lemma~\ref{lem:prop2x}, 
the expected complexity of a random 
level in this range is
\[
\frac{1}{\eps\lambda/6} \cdot
O\left(\frac{n}{\eps^6 k} \log^3 n \right) =
\frac{1}{\frac{\eps}{6}\cdot c\eps^{-2} \log n } \cdot
O\left(\frac{n}{\eps^6 k} \log^3 n \right) =
O\left(\frac{n}{\eps^5 k} \log^2 n \right) ,
\]
as claimed.
\end{proof}

Finally, we discuss the case $k = O(\eps^{-2} \log n)$. 
In this case, we pick $t$ uniformly at random in the 
interval $[k,(1+\eps)k]$, and we set $\oT_k$ to $L_t(F)$.
By construction, it is clear that
$\oT_k$ approximates the $k$-level in $\A(F)$. Furthermore, 
the same Clarkson-Shor bound used in the proofs of 
Lemma~\ref{lem:prop2x} and~\ref{lem:prop2} shows 
that $\oT_k$ has expected complexity 
$O(\textfrac{nk}{\eps}) = O\big((\textfrac{n}{\eps^5 k})\log^2 n\big)$, 
using our assumption on $k$.

\paragraph{Remark.} 
The same result holds for general lower envelope complexity. Suppose 
that every set of $m$ functions in $\F$ has lower envelope complexity 
at most $\psi(m)$, where we assume (or require) that
$m \mapsto \psi(m)/m$ is monotonically increasing.
Then, given a set $F$ of $n$ functions from $\F$, for every 
$\eps \in (0, 1/2]$ and for every 
$0\le k \le n-1$, we can find a terrain $\oT_k$ that 
lies fully between $L_k(F)$ and $L_{(1+\eps)k}(F)$ and that 
has complexity $O(\psi(n/k)\eps^{-5}\log^2 n)$.

Indeed, the argument proceeds as above, with slightly adjusted 
bounds.  The Clarkson-Shor bound in the proof of Lemma~\ref{lem:prop2x}
now shows that there are 
\[
O((1+\eps)^3 k^3 \psi(n/(1+\eps)k) ) =
O(k^3\psi(n/k))
\]
vertices in $L_{\leq(1+\eps)k}(F)$, so the 
expected number of vertices in $\A(S_k)$ with level between 
$(1+\eps/3)\lambda$ and $(1+\eps/2)\lambda$ is
$O(\psi(n/k)\eps^{-6}\log^3 n)$. Dividing by $\eps\lambda/6$,
we obtain the claimed bound.

\section{From approximate levels to shallow 
  cuttings}\label{sec:cutting}

Having obtained an approximate level $\oT_k$
as in Section~\ref{sec:levelapproximation},
we would like to turn $\oT_k$ into a shallow cutting for
$L_{\leq k}(F)$ by creating for each face 
$\ovarphi$ of $\oT_k$ a semi-unbounded
vertical pseudo-prism $\varphi$ that consists of 
the points vertically below $\ovarphi$.
For brevity, we
will refer to these pseudo-prisms simply as 
\emph{prisms}, and we denote them by $T_k$.
The only issue is that the faces 
$\ovarphi$ need not have constant complexity, 
so that the corresponding prisms 
might be crossed by too many functions in $F$.

Thus, we decompose each 
face $\ovarphi$ of $\oT_k$ into sub-faces of constant 
complexity, using two-dimensional vertical decomposition. More 
precisely, we project each face $\ovarphi$ onto the 
$xy$-plane, and we decompose the resulting 
projection $\ovarphi^*$ into $y$-vertical 
pseudo-trapezoids by erecting $y$-vertical 
segments from each vertex of 
$\ovarphi^*$ and from each point of vertical 
tangency on its boundary, extending them either
into infinity or until 
they hit another edge of $\ovarphi^*$. 
By planarity, the number of pseudo-trapezoids is proportional to 
the complexity of $\ovarphi$. We lift 
each resulting pseudo-trapezoid $\tau^*$ into 
a prism $\tau$,  consisting 
of all the points vertically below 
$\ovarphi$ that project to 
$\tau^*$.\footnote{%
A significant difference between the machinery 
used here and that for the case of planes, as
in~\cite{HKS}, say, is that in the case of 
planes we only lift the vertices of the $xy$-map
to the appropriate level (or to an approximation 
of the level), and each triangular face is 
lifted to the convex hull of its vertices, which 
in general is not contained in the level. In 
contrast, here we lift each pseudo-trapezoidal 
face from the $xy$-plane to lie fully on the 
level.}
Our cutting $\Lambda_k$ consists of 
all these prisms $\tau$, and we denote by 
$\oLambda_k$ the terrain formed by the ceilings 
$\otau$, for $\tau\in \Lambda_k$.
Then, $\oLambda_k$ is a refinement of 
$\oT_k$.  As we will shortly show, $\Lambda_k$ 
is indeed a shallow cutting for $L_{\leq k}(F)$.
For each prism $\tau\in \Lambda_k$, the 
\emph{conflict list} 
$\CL(\tau)$ is the set of 
functions of $F$ that intersect $\tau$.

\begin{lemma}\label{lem:shallow-cutting}
$\Lambda_k$ is a shallow cutting of the first 
$k$ levels of $\A(F)$. It consists (in 
expectation) of
\[
O(|\Lambda_k|) =
O\left(\frac{n}{\eps^5 k} \log^2 n \right)
\]
prisms, and each prism in $\Lambda_k$ intersects 
at least $k$ and at most $(1+2\eps)k$ graphs of 
functions of $F$.
\end{lemma}
\begin{proof}
Let $\tau \in \Lambda_k$, and let $p$ be vertex
of the ceiling $\otau$ of $\tau$. By 
Lemma~\ref{lem:prop1}, 
the level of $p$ in $\A(F)$ is 
in $[k,(1+\eps)k]$. 
Thus, at most $(1+\eps)k$ functions
in $F$ pass below all vertices of $\tau$.
Furthermore, since
$\overline{\tau}$ does not intersect 
any function in $S_k$, since $S_k$ is a relative 
$\left(\textfrac{k}{2n},\textfrac{\eps}{3}\right)$-approximation 
for $(\F,\R)$, and since $\otau$ induces a range in $\R$, 
by the second bound in (\ref{rel-apx}), it follows that at most
\[
\frac{\eps pn}{3} = \frac{\eps k}{6} 
\]
functions of $F$ cross $\otau$.
For any function
$f\in F$ that intersects $\tau$ either  
passes below all vertices of $\otau$ 
or crosses $\otau$, we get
\[
|\CL(\tau)| \le \left(1+\frac{7\eps}{6}\right)k < (1+2\eps)k.
\]
The construction of $\Lambda_k$ ensures that 
$|\Lambda_k|$ is proportional to 
the complexity of $\oT_k$, so, by 
Lemma~\ref{lem:prop2}, it satisfies (in expectation) 
the bound asserted in the lemma.
\end{proof}

\paragraph{Remark.}
More generally, Agarwal~\etal~\cite{AES} show
the following: 
let $\psi: \mathbb{N} \rightarrow \mathbb{N}$ be such
that any $m$ functions in $\F$ have a lower envelope of 
complexity $\psi(m)$. Let $F \subset \F$ be a
set of $n$ functions in $F$. Then, for any $k \in \{1, \dots, n-1\}$,
there exists a $k$-shallow $\Theta(k/n)$-cutting for $F$
of size $O(\psi(n/k))$. Our techniques also generalize to this case.
In particular, we obtain the following result.
\begin{lemma}\label{lem:shallow-cutting-general}
For any $k \in \{1, \dots, n-1\}$, our sampling procedure
yields
a shallow cutting $\Lambda_k$ of the first 
$k$ levels of $\A(F)$. It consists (in 
expectation) of
\[
O(|\Lambda_k|) =
O\left(\eps^{-5}\psi(n/k) \log^2 n \right)
\]
prisms, and each prism in $\Lambda_k$ intersects 
at least $k$ and at most $(1+2\eps)k$ graphs of 
functions of $F$.
\end{lemma}
\begin{proof}
We only need a more general bound on the complexity of
$\oT_k$.
By Clarkson-Shor, in general, there are $O(\psi(n/k)k^3)$
vertices in $L_{\leq(1+\eps)k}(F)$, so we get the bound 
$O(\eps^{-6}\psi(n/k)\log^3 n)$
in Lemma~\ref{lem:prop2x} 
and 
$O(\eps^{-5}\psi(n/k)\log^2 n)$ in Lemma~\ref{lem:prop2}. 
Now, the result
follows as before.
\end{proof}

\section{Randomized incremental construction of the $\leq t$ level}
\label{sec:ric}

Again, let $\F$ be a family of bivariate functions in 
$\reals^3$ with constant description complexity 
and with linear lower envelope complexity. Let 
$F$ be a subset of $n$ members of $\F$, which we
assume to be in general position, and let 
$0\le t \le n-1$. Our goal is to 
construct the first $t$ levels of $\A(F)$. 
We describe an algorithm with expected running time 
$O(nt\lambda_s(t)\log(n/t)\log n)$ and with expected 
storage $O(nt\lambda_s(t))$, where $s$ is a 
constant that depends on $\F$, and $\lambda_s(t)$ is the
familiar Davenport-Schinzel bound~\cite{SA}. 
Our algorithm can be used to compute a vertical shallow 
cutting as prescribed in Section~\ref{sec:cutting},
together with the conflict lists of its prisms.\footnote{%
  In more detail, as will be described later, we run the
  algorithm on the entire set $F$, but stop it after inserting
  $r_k$ functions (where $r_k$ is as given in (\ref{eq:rk})).
  This will provide us with the conflict lists of the final
  prisms with respect to the entire set $F$.}

We follow the standard \emph{randomized incremental 
construction} (RIC) paradigm: we insert 
the surfaces of $F$ one at a time, in random 
order, and maintain, after each insertion, the 
first $t$ levels in the arrangement of the 
surfaces inserted so far ($t$ stays fixed 
during the process). Number the elements 
of $F$ in the random insertion order as 
$f_1, f_2, \dots, f_n$, and put $F_i = \{f_1, \dots, f_i\}$, 
for $i = 1, \dots,n$. As is standard in the RIC approach, 
the algorithm maintains a decomposition (the standard
vertical decomposition in our case) of $L_{\le t}(F_i)$
into cells of constant description complexity (these
are not necessarily the semi-unbounded prisms of the
vertical shallow cutting that we are after---see below),
and keeps the \emph{conflict list} for each cell 
$\tau$, i.e., the set of all functions in $F$ that 
cross $\tau$. When the next function $f_{i+1}$ is inserted,  
the conflict lists can be used to retrieve the cells that 
are crossed by $f_{i+1}$.  These cells are ``destroyed'', as
they no longer appear in the new decomposition, and are 
partitioned by $f_{i+1}$ into fragments. These fragments 
are not necessarily valid prisms for the vertical 
decomposition of $L_{\le t}(F_{i+1})$, and may need 
to be merged and refined into the correct new cells. 
In addition, we have to construct the conflict 
lists of the new cells, which are obtained 
from the conflict lists of the destroyed cells.

\subsection{Computing the first $t$ levels} \label{sec:complk}

After each insertion, we maintain the \emph{vertical decomposition} 
$\VD_{\le t}(F_i)$ of $L_{\le t}(F_i)$, the first $t$ levels of 
$\A(F_i)$. We obtain $\VD_{\le t}(F_i)$ by applying two 
decomposition stages to each cell of $L_{\le t}(F_i)$. 
(We reiterate that this decomposition differs from the vertical 
shallow decomposition used above, in the sense that its
prisms are in general not semi-unbounded; see below.)

We call a cell $C$ of $L_{\le t}(F_i)$ a \emph{stage-0 cell}.
In the first stage, we erect a vertical 
\emph{wall} within $C$ through each edge $e$ 
of $L_{\leq t}(F_i)$ on $\bd C$.
Each such wall is the union of all 
maximal vertical segments that lie in (the 
closure of) $C$ and pass through the 
points of $e$. These walls partition $C$ into \emph{stage-1 cells}. 
Every stage-1 cell $C'$ has a unique \emph{ceiling}
(``top'' surface) and a unique \emph{floor} (``bottom'' surface).
The ceiling and/or floor may be undefined if
$C'$ is unbounded. All other facets of $C'$ lie on the vertical 
walls. The complexity of $C'$ may however still be arbitrarily large.
Thus, in the second stage, we take each stage-1 cell $C'$, project it 
onto the $xy$-plane, and apply a two-dimensional 
vertical decomposition to the projection $(C')^*$. 
That is, as in Section~\ref{sec:cutting},
we erect a $y$-vertical segment through each vertex 
of $(C')^*$ and through each locally $x$-extremal point
on its boundary. This partitions $(C')^*$ into $y$-vertical
pseudo-trapezoids, and we lift each such pseudo-trapezoid 
$\tau$ back into $\reals^3$ by forming the intersection 
$(\tau\times\reals)\cap C'$. This yields a decomposition 
of $C'$ into prism-like \emph{stage-2 cells} of constant 
description complexity, referred to, for simplicity, just
as \emph{prisms}.\footnote{%
  Each prism is bounded by at most six surfaces of 
  constant desciption complexity---see below.}
More details can be found in \cite{CEGS,SA}.
Collectively, all the stage-2 cells, over all cells $C$
and all subcells $C'$, constitute the vertical decomposition
$\VD_{\le t}(F_i)$.

\subsubsection{Complexity of the vertical decomposition}

As is well known, the complexity of 
$\VD_{\le t}(F_i)$ is proportional to the number 
of tuples $(q, e, e')$, for $e$, $e'$ edges of $L_{\le t}(F_i)$ 
and for $q\in\reals^2$, so that $q$ belongs to the 
$xy$-projections of $e$ and $e'$, and the $z$-vertical line 
$\ell_q$ through $q$ meets no other surface of $F_i$ between 
$e$ and $e'$. We call $(e,e')$ a \emph{vertically visible} pair,
and refer to $(q,e,e')$ as a \emph{triple of vertical visibility}.
We assume that the pair $(e,e')$ is 
ordered such that $e$ lies above $e'$, 
i.e., we encounter $e$ before $e'$ as we travel
along $\ell_q$ from $z = \infty$ to $z = -\infty$.

The following crucial lemma, which we regard as one of the 
main contributions of the paper, bounds the complexity of
$\VD_{\leq t}(F_i)$. It improves an earlier bound of 
$O(nt^{2+\eps})$ by Agarwal~\etal~\cite{AES}. We define 
a parameter $s$ as follows: For any $f_1, f_2, f_3, f_4 \in F$,
we let $s(f_1, f_2, f_3, f_4)$ denote the number of co-vertical 
pairs of points $q \in f_1\cap f_2$, $q'\in f_3\cap f_4$. 
Then $s=s_0+2$, where $s_0$ is the maximum of $s(f_1,f_2,f_3,f_4)$, 
over all quadruples $f_1,f_2,f_3,f_4 \in F$.
By our assumptions on $\F$ (including general position), 
we have $s=O(1)$, where the constant depends\footnote{%
  In certain applications we can, and will, derive concrete bounds on $s$.}
on the complexity of the family $\F$. We use 
$\lambda_s(t)$ to denote the maximum length of
a Davenport-Schinzel sequence of order $s$
on $t$ symbols~\cite{SA}.

\begin{lemma} \label{lem:vdpairs}
Let $F$ be a set of $n$ functions of $\F$, and let 
$1\le t \le n-1$. The complexity of 
$\VD_{\le t}(F)$ is $O(nt\lambda_s(t))$.
\end{lemma}
\begin{proof}
Let $e$ be an edge of $L_{\le t}(F)$, and  let $F_e \subseteq F$
be the functions in $F$ that pass vertically below some point on $e$.
Since $e$ is not crossed by any function of $F_e$, each $f \in F_e$ 
appears below every point of $e$, implying that $|F_e| \leq t$.
Let $V_e$ be the vertical wall erected downward from $e$, 
all the way to $z=-\infty$. Then, the complexity of the upper 
envelope of $F_e$, clipped to $V_e$, is at most $\lambda_{s_0}(t)$. 
Indeed, using a suitable parametrization of $e$, 
the cross-sections of the functions in $F_e$ with 
$V_e$ are totally defined univariate continuous functions, 
each pair of which intersect at most $s_0$ times. 
This follows from the definition of $s_0$, since the vertices 
of the arrangement of these functions are exactly
the intersection points of $V_e$ with edges $e'$ 
of $L_{\le t}(F)$ that form co-vertical pairs $(e,e')$ with $e$.
Since the vertices of the upper envelope of these functions
stand in a 1-1 correspondence with the triples of vertical 
visibility pairs with $e$ as the top edge, the number of these pairs
is at most $\lambda_{s_0}(t)$, as claimed.

A standard application of the Clarkson-Shor technique implies 
that $L_{\le t}(F)$ has $O(t^3\cdot(n/t))=O(nt^2)$ edges. 
This follows by charging the edges to their endpoints 
and by using the fact that there are $O(m)$ vertices 
on the lower envelope of any $m$ functions of $F$. 
This already gives a (weak) bound of 
$O(nt^2\lambda_{s_0}(t)) \approx nt^3$ on the 
complexity of $\VD_{\le t}(F)$.

The arguments so far follow the initial part of the
analysis of Agarwal~\etal~\cite{AES}, but the next part
is new and gives a sharper bound. Fix two functions 
$f, f'\in F$, and let $\gamma = f \cap f'$ be their 
intersection curve. We cut $\gamma$ at each singular 
and locally $x$-extremal point. This decomposes $\gamma$
into $O(1)$ connected $x$-monotone Jordan subarcs. 
Recall that, in addition to general position of $F$, 
we also assume a generic coordinate frame, so that 
no resulting piece lies within some $yz$-parallel plane.

We cut these arcs further at their 
intersections with the level $L_{t}(F)$, and we keep those 
portions that lie in $L_{\le t}(F)$. 
To control the number of such portions, we
relax the problem a bit, replacing the level 
$t$ by a larger level $t'$ with  $t\le t'\le 2t$,
for which the complexity of $L_{t'}(F)$ is 
$O(nt)$. Since the overall complexity of $L_{\le 2t}(F)$ is 
$O(nt^2)$ (as just noted), the average complexity of a 
level between $t$ and $2t$ is indeed $O(nt)$. 
Thus, there is a level $t'$ with the 
above properties. We 
will establish the asserted upper bound for 
$\VD_{\le t'}(F)$, which then also
applies to $\VD_{\leq t}(F)$.
To keep the notation simple, we
continue to denote the top level $t'$ as $t$.

Let $\Gamma$ be the set of all Jordan
subarcs of some intersection 
curve that lies in $L_{\le t}(F)$ (now with 
the new, potentially larger, index $t$). 
If $\gamma \in \Gamma$ does not fully lie 
below $L_{t}(F)$, it ends in at least one vertex 
of $L_{t}(F)$, so the number of these
$\gamma \in \Gamma$ is $O(nt)$. 
Any other $\gamma \in \Gamma$ is charged either to one of 
its endpoints, or, if it is unbounded (and $x$-monotone), 
to its intersection 
with a plane at infinity, say $V_\infty: x = +\infty$.
If $\gamma$ reaches $V_\infty$, it
appears there as a vertex of the first $t$ 
levels of the cross-sections of the functions 
in $F$ with $V_\infty$. An application of 
the Clarkson-Shor technique to this planar 
arrangement shows that there are
$O(nt)$ such vertices, so this also bounds the 
number of these arcs in $\Gamma$.
Finally, we bound the number of $\gamma \in \Gamma$ with a 
singular or locally $x$-extremal endpoint 
by charging $\gamma$ to this 
endpoint. 
The number of these points lying in $L_{\le t}(F)$ is bounded by yet another 
application of the Clarkson-Shor technique.
Noting that each such point is now defined by 
only two functions of $F$, this leads to the 
upper bound $O(nt)$. 
Thus, $|\Gamma|  = O(nt)$.

Fix an arc $\gamma \in \Gamma$,  and
let $\mu(\gamma)$ be the number of edges
in $L_{\leq t}(F)$ on $\gamma$. In general,
$\mu(\gamma) \geq 1$. We decompose $\gamma$
into $\xi(\gamma) := \lceil \mu(\gamma)/t\rceil$ 
pieces, each consisting of at most $t$ 
consecutive edges. By general 
position, if $e_1$ and $e_2$ are consecutive
edges along $\gamma$, the set 
of functions of $F$ that appear
below $e_1$ and the set of functions that appear
below $e_2$ differ exactly by the third function incident to the common 
endpoint of $e_1$ and $e_2$. 
This implies that, for a piece $\delta$ of 
$\gamma$, the overall number of functions 
that appear below $\delta$ is at 
most $2t$. Some of these functions 
are now only partially defined. Arguing as 
above, the number of vertically visible pairs 
whose top edge lies on $\delta$ 
is at most $\lambda_{s_0 + 2}(2t) = O(\lambda_s(t))$. 
Hence, the overall number of triples of vertical 
visibility in $L_{\le t}(F)$ is
\begin{equation*}
  \left(\sum_{\gamma\in\Gamma} \xi(\gamma) \right)\cdot 
     O(\lambda_s(t))  \le
   \left( \sum_{\gamma\in\Gamma} \left( \frac{\mu(\gamma)}{t} + 1 \right)
   \right) \cdot O(\lambda_s(t)) 
    =
  \left(\frac{1}{t} \sum_{\gamma\in\Gamma} \mu(\gamma) + |\Gamma| \right)
  \cdot O(\lambda_s(t)).
\end{equation*}
We have already seen that $|\Gamma|=O(nt)$. Furthermore, 
$\sum_{\gamma\in\Gamma} \mu(\gamma)$ is simply 
the number of edges in $L_{\le t}(F)$, which, as 
already argued, is $O(nt^2)$. It follows that the 
number of vertically visible pairs in 
$L_{\le t}(F)$ is $O\left( nt\lambda_s(t) \right)$. 
\end{proof}

\paragraph{Remark.}
Our analysis also works if the lower envelope complexity is 
not necessarily linear. Let
$\psi: \mathbb{N} \rightarrow \mathbb{N}$ be such
that any $m$ functions in $\F$ have a lower envelope of 
complexity at most $\psi(m)$. Then we obtain the following bound.
\begin{lemma} \label{lem:vdpairs_general}
Let $\F$ be a family of functions with lower envelope 
complexity bounded by $\psi(m)$, let $F$ be a set of 
$n$ functions of $\F$, and let $1\le t \le n-1$. 
The complexity of $\VD_{\le t}(F)$ is 
$O(t^2\psi(n/t)\lambda_s(t))$.
\end{lemma}
\begin{proof}
The proof proceeds exactly as the proof of 
Lemma~\ref{lem:vdpairs}, but with more general bounds 
for the various structures associated with $\A(F)$. 
In particular, by the Clarkson-Shor technique,
the overall complexity of $L_{\leq 2t }(F)$ now
is $O(t^3\psi(n/t))$, and hence we can find a level
$t\le t'\le 2t$ such that $L_{t'}$ has complexity
$O(t^2\psi(n/t))$. The arguments for bounding
$|\Gamma|$ remain valid, but now the complexity
of the first $t$ levels of the planar arrangement 
on $V_\infty$, and the number of singular or locally
$x$-extremal points in $L_{\leq t}$, are both
$O(t^2\psi(n/t))$.\footnote{%
To be precise, the bound on the complexity of the
first $t$ levels of the planar arrangement 
on $V_\infty$ is only
$O(t^2 \lambda_{s'}(n/t))$, where
$\lambda_{s'}$ is a Davenport-Schinzel-factor
that accounts for the maximum number of times that two
surfaces intersect. However, this is dominated by
the other contributions, so we ignore this refinement here.}
Proceeding with these bounds 
as before, we obtain the claimed result.
\end{proof}

\paragraph{The randomized incremental construction.}
Although the high-level description of the randomized incremental
construction is fairly routine, the finer details are somewhat
intricate, and their description is rather lengthy. We present
the construction, with full details, in Appendix~\ref{app:ric}.
As we will show, the expected running time of the procedure is 
proportional to
the expectation of
\begin{equation} \label{ric-compl}
\sum_{\tau \in \Pi^*} (1 + |\CL(\tau)|)\log n,
\end{equation} 
where $\Pi^*$ is the set of all prisms that
are generated during the incremental process, and 
where $\CL(\tau)$ is the conflict list of prism $\tau$.

\subsection{Analysis}

We now bound the expected value of (\ref{ric-compl}).
Let $\Pi$ be the set of all possible pseudo-prisms. 
That is, we consider all possible sets $F_0$ of at most 
six surfaces in $F$, and for each such $F_0$, we construct 
the entire vertical decomposition $\VD(F_0)$ and add the 
resulting prisms to $\Pi$.

We associate two \emph{weights} with each prism 
$\tau\in\Pi$. The first weight $w_0(\tau)$ 
is the size of its conflict list, that is $w_0(\tau) = |\CL(\tau)|$. 
The second weight $w^-(\tau)$ equals the 
number of surfaces that pass fully below $\tau$. 
For simplicity, we focus on prisms that are 
defined by exactly six functions; the treatment 
of prisms defined by fewer functions is 
fully analogous. Let $\Xi(\tau)$ denote the set of defining
functions of $\tau$. As just said, we only consider the
case $|\Xi(\tau)| = 6$.

Following a standard approach to the 
analysis of RICs, we proceed in two steps. First, 
we estimate the probability that a prism with given weights 
ever appears in $L_{\leq t}(F_i)$. 
Then, we estimate the number of prisms $\tau$ with weights 
$w^-(\tau) \le a$ and $w_0(\tau) \le b$, 
using the Clarkson-Shor technique and 
several other considerations. Finally, we  combine 
the bounds to get the desired estimate on the 
expected running time and storage of the algorithm.

\paragraph{Estimating the probability of a prism to appear.}
For the first step, let $\tau \in \Pi$ be a prism 
with six defining functions and
with $w^-(\tau)=a$, $w_0(\tau)=b$. That is, $\tau$ 
has $w^-(\tau) = a$  \emph{lower surfaces} 
and $w_0(\tau) = b$ \emph{crossing surfaces}.
Neither of $w^-(\tau)$, $w_0(\tau)$ counts any of the
defining functions of $\tau$, although some of these
functions might pass below $\tau$. The number of such `lower'
defining functions is always between $1$ and $5$, because the 
floor is always such a lower function, and the ceiling is always 
excluded; see Figure~\ref{fig:lowdef}.

\begin{figure}[htb]
\centering
\includegraphics[scale=0.6]{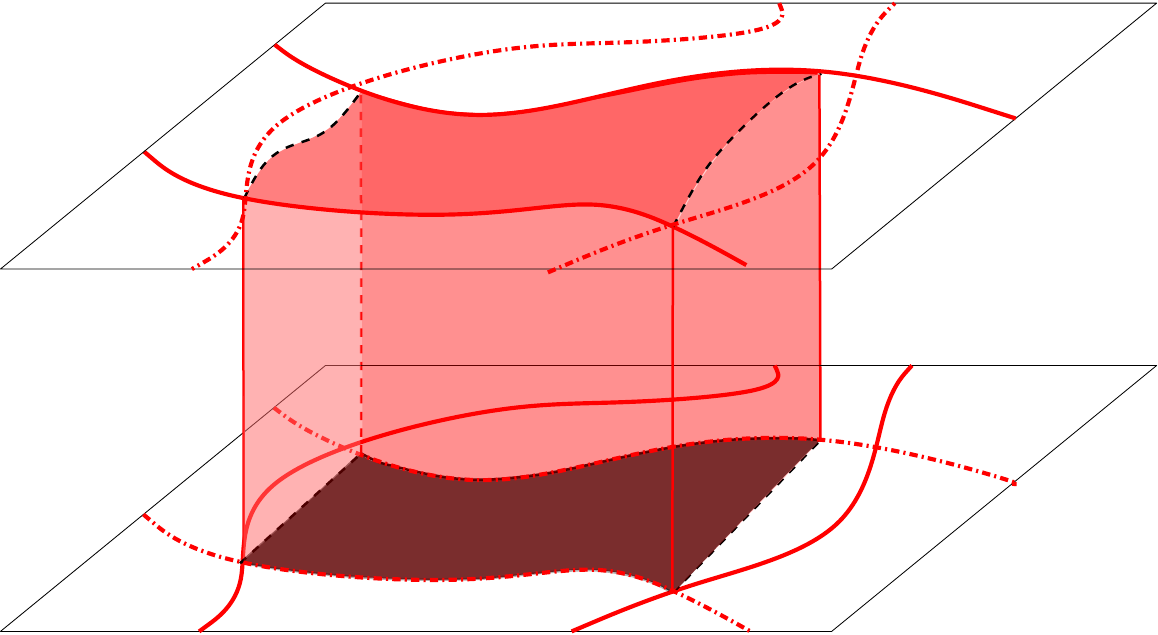}
\caption{A schematic illustration of a prism with three lower defining
functions. The solid lines represent actual intersections 
between surfaces, the dotted lines represent shadow 
edges or vertical edges. To belong to $L_{\le t}(F_i)$, at 
most $t-4$ lower
non-defining functions may belong to $F_i$.}
\label{fig:lowdef}
\end{figure}

The prism $\tau$ appears in some 
$\VD_{\le t}(F_i)$ if and only if (i) the last function
in $\Xi(\tau)$, denoted as $f_6$, 
is inserted before any of the $b$ crossing surfaces; and
(ii) at most $t'=t-\xi$ of the $a$ lower non-defining surfaces 
are inserted before $f_6$, where $\xi\ge 1$ is the number 
of defining functions of $\tau$ that pass below $\tau$,
one of which is the floor of $\tau$.

The probability $p_\tau$ of this event can be 
calculated as follows.  Restrict 
the random insertion permutation to the $a+b+6$ 
relevant surfaces (the $a$ surfaces below $\tau$, 
the $b$ surfaces crossing $\tau$, and the six
surfaces in $\Xi(\tau)$). To get a restricted 
permutation that satisfies (i) and (ii), we first
choose which function in $\Xi(\tau)$ is 
$f_6$, then we choose some 
$j\le \min\{a,t'\}$ of the $a$ lower surfaces 
to precede $f_6$, then we mix these $j$ surfaces 
with the five in $\Xi(\tau) \setminus \{f_6\}$, 
and finally we place the
remaining $a-j$ lower surfaces and all $b$ crossing 
surfaces after $f_6$.  We thus get
\begin{equation} \label{equ:p-tau-bound}
p_\tau = \sum_{j=0}^{\min\{a,t'\}} \frac{6 \binom{a}{j} (j+5)! (a+b-j)!} 
{(a+b+6)!}.
\end{equation}
We rewrite and bound each summand  in (\ref{equ:p-tau-bound})
as follows.
\begin{align*}
\frac{6 \binom{a}{j} (j+5)! (a+b-j)!} {(a+b+6)!} & =
\frac{6a!(j+5)! (a+b-j)!}{j!(a-j)!(a+b+6)!} \\
& = \frac{6}{a+b+6}\cdot \frac{(a-j+1)\cdots(a-j+b)}{(a+1)\cdots(a+b)} 
\cdot
\frac{(j+1)\cdots(j+5)}{(a+b+1)\cdots(a+b+5)} \\
& \le \frac{6}{a+b+6}\cdot \left(\frac{a-j+b}{a+b}\right)^b \cdot
\left(\frac{j+5}{a+b+5}\right)^5 \\
& \le  \frac{6}{a+b+6}\cdot \left(\frac{j+5}{a+b+5}\right)^5 
\cdot e^{-jb/(a+b)}.
\end{align*}
Let
\[
\varphi_{a,b}(j) = 
\frac{6}{a+b+6}\cdot \left(\frac{j+5}{a+b+5}\right)^5 \cdot 
e^{-jb/(a+b)}
\]
be our bound  on the $j$th summand 
of (\ref{equ:p-tau-bound}).
Note that with $a,b$ fixed, $\varphi_{a,b}(x)$ 
peaks at $x = 5\textfrac{a + b}{b}-5 = \textfrac{5a}{b}$
(the positive root of the derivative of $\varphi_{a,b}(x)$ satisfies 
$5(x+5)^4 - \textfrac{b(x+5)^5 }{(a + b)} = 0$).
We estimate $p_\tau$ by replacing the sum by an integral. That is, 
\begin{align} \label{equ:integral}
p_\tau &= \sum_{j=0}^{\min\{a,t'\}} 
\frac{6 \binom{a}{j} (j+5)! (a+b-j)!} {(a+b+6)!} \nonumber \\
& \le \sum_{j=0}^{\min\{a,t'\}} \varphi_{a,b}(j)
\le e \int_{0}^{\min\{a,t'\} + 1} \varphi_{a,b}(x) dx \\
& = \frac{6e}{a+b+6}\cdot \int_{0}^{\min\{a,t'\} + 1} 
\left(\frac{x+5}{a+b+5}\right)^5 \cdot e^{-xb/(a+b)} dx ; \nonumber
\end{align}
to justify the inequality between the sum and 
the integral in (\ref{equ:integral}), it 
suffices to note that, for $x\in [j,j+1]$,
\[
\frac{\varphi_{a,b}(j)}{\varphi_{a,b}(x)} = 
\left(\frac{j+5}{x+5}\right)^5e^{-b(j-x)/(a+b)} \leq e^{b/(a+b)} \leq e,
\]
for every $j\ge 0$. 
To estimate the integral, we substitute
$y=xb/(a+b)$. The upper limit in the integral becomes
\[
c := (\min\{a,t'\} + 1) \cdot \frac{b}{a+b} ,
\]
and we get
\begin{align} \label{(**)}
p_\tau & \le \frac{6e(a+b)}{b(a+b+6)}\cdot \int_{0}^{c} 
\left(\frac{y(a+b)/b+5}{a+b+5}\right)^5 \cdot e^{-y} dy \nonumber \\
& =\frac{6e(a+b)}{b(a+b+6)}\cdot \left(\frac{a+b}{b(a+b+5)}\right)^5 
\int_{0}^{c} \left(y+5b/(a+b)\right)^5 e^{-y} dy \nonumber \\
& \le \frac{6e}{b^6} \int_{0}^{c} \left(y+5b/(a+b)\right)^5 e^{-y} dy .
\end{align}
The integral in (\ref{(**)}) is at most
\[
\int_{0}^{\infty} \left(y+5\right)^5 e^{-y} dy = O(1).
\]
Thus,\footnote{%
  Technically, we should write this as $O(1/(b+1)^6)$, 
  to cater also for the case $b=0$.
  We gloss over this trifle issue, as is 
  common in other works too, to simplify the notation.}
$p_\tau = O(1/b^6)$. For large $c$,
we cannot improve this bound. However, if $c$ is 
sufficiently small, bounding the integral in (\ref{(**)})
by an absolute constant may be wasteful.
For $a\le t$ we will not refine the bound and use 
$p_\tau = O(1/b^6)$. Consider now the case $a > t > t'$,
so $c=(\min\{a,t'\} + 1)\textfrac{b}{(a+b)} = 
\textfrac{b(t' + 1)}{a+b}$.
As is easily checked, the function 
$\varphi(y) = \left(y+5b/(a+b)\right)^5 e^{-y}$
is increasing on  $[0, 5a/(a+b)]$, so when
\[
c =  \frac{b(t' + 1)}{a+b} \le \frac{5a}{a+b} 
,\quad\quad\text{or}\quad\quad
t' \le \frac{5a}{b}  -1 \ ,
\]
we bound the integral in (\ref{(**)}) by $c \cdot \varphi(c)$, 
and get\footnote{%
  Again, we should write $a+1$ in the final expression.}
\[
p_\tau \le
\frac{6e}{b^6} \cdot \frac{b(t' + 1)}{a+b} \left(\frac{b(t'+6)}{a+b}\right)^5 \cdot 
e^{-b(t'+1)/(a+b)} =
O\left( \frac{t^6}{(a+b)^6} e^{-bt/(a+b)} \right) =
O\left( \frac{t^6}{a^6} \right).
\]
Thus, we can bound $p_\tau$ in terms of $a$ and $b$. 
Denoting this bound by $p(a,b)$, we have
\begin{equation} \label{pab}
p(a,b) = \begin{cases}
O\left( 1/b^6 \right), & \text{for $a \le bt/5$ or $a\le t$} \\
O\left( t^6/a^6 \right), & \text{for $a > bt/5$
and $a>t$} .
\end{cases}
\end{equation}
(Unless $b$ is very small, the constraint 
$a\le t$ or $a>t$ is subsumed by the
other respective constraint.)

\paragraph{Bounding the number of prisms of small weights.}
We next estimate the number of prisms $\tau \in \Pi$ 
with $w^-(\tau)\le a$ and $w_0(\tau)\le b$.
We denote this quantity as $N_{\le a,\le b}$.
We also use the notation $N_{a,b}$ for
the number of prisms $\tau \in \Pi$ 
with $w^-(\tau) = a$ and $w_0(\tau) = b$.

\begin{lemma} \label{lem:ab}
The number of prisms $\tau$ with $w^-(\tau)\le a$ 
and $w_0(\tau)\le b$ is $O(nb^5)$, for $a\le b$,
and $O(nab^4\lambda_s(a/b))$, for $a>b$.
\end{lemma}
\begin{proof}
Set $p = 1/b$, and let $R \subseteq F$ be a random sample of size $pn$.

\paragraph{Case 1:} $a \leq b$.
We assume that $b \leq n/10$, so that $p = 1/b \geq 10/n$.
Fix a prism $\tau \in \Pi$, defined 
by six functions, with $w^-(\tau) = i$ 
and $w_0(\tau) = j$, with $i\le a$, $j\le b$. 
Let $q_\tau$ be the probability that $R$ contains (i) the 
defining set $\Xi(\tau)$;
(ii) none of the $j$ crossing 
functions; and (iii) none of the $i$ lower 
functions. By elementary probability,
\begin{align*}
q_\tau = \frac{\binom{n - 6 - i - j}{np - 6}}{\binom{n}{np}}
&\geq \frac{\binom{n - 6 - a - b}{np - 6}}{\binom{n}{np}}\\
&= \prod_{k = 0}^{a+b-1}\frac{n- np - k} {n - 6 - k}\cdot
\prod_{l = 0}^{5} \frac{np - l}{n - l}\\
&\geq \prod_{k = 0}^{a+b-1}\left(1 - \frac{np}{n-k}\right)
\cdot \prod_{l = 0}^{5} \frac{np - l}{n}\\
&\geq \left(1 - \frac{p}{2}\right)^{a+b} \cdot
\left( \frac{p}{2} \right)^6 \geq \left(1 - \frac{1}{2b}\right)^{2b}
\cdot \left( \frac{1}{2b} \right)^6  = \Omega(1/b^6) ;
\end{align*}
this follows since we set $p = 1/b$ and  assumed 
that $a \leq b$ and $b \leq n/10$, so that 
$n - k \geq n - a - b \geq n - 2b \geq n/2$ and 
$np - l \geq np - 5 \geq np/2$.
If the event holds, $\tau$ becomes a prism in $\VD_{\leq 6}(R)$.
By Lemma~\ref{lem:vdpairs}, the number of such prisms is 
$O(|R|) = O(n/b)$. This yields, as a variant of the Clarkson--Shor theory,
$N_{\le a,\le b} = O(nb^5)$.
This bound also holds trivially if $ b > n/10$, since there
are at most $O(n^6)$ prisms in total.

\paragraph{Case 2:} $a > b$.
Again, we assume that $b \leq n/10$, so we have $p = 1/b \ge 10/n$. 
Also, we require that $n$ is more than a large enough constant.
We put $\xi_0 = 2a/b$ and $\xi = \xi_0 + \nu$, where
$\nu$ is the number of defining functions of $\tau$ that pass below $\tau$.
As before, fix a prism $\tau \in \Pi$, defined 
by six functions, with $w^-(\tau) = i$ 
and $w_0(\tau) = j$, with $i\le a$, $j\le b$. 
The probability $q_\tau$ that $\tau$ 
appears in $\VD_{\leq \xi}(R)$
is the probability that $R$ contains (i) the 
defining set $\Xi(\tau)$; (ii) none of the $j$ crossing 
functions; and (iii) at most $\xi_0$ 
of the $i$ lower functions. Similarly to Case~1,
the probability that (i) and (ii) hold is at least
$(p/2)^6(1-p/2)^b = \Omega(1/b^6)$. 
Conditioned on (i) and (ii) holding, 
(iii) is the event that when choosing
$np -6$ out of $n - 6 - j$ functions,
we obtain at most $\xi_0$ of the $i$ lower functions.
The number of the lower functions in  the sample follows 
a hypergeometric distribution, with expected value 
\[
i \frac{np-6}{n-6-j} \leq i \frac{np}{n-6-b} \leq i\frac{11p}{9}  \leq 
\frac{11}{9}\frac{a}{b},
\]
using our assumption that $b \leq n/10$ and that $n$ is large enough.
Now, the Chernoff bound for the hypergeometric 
distribution (see~\cite{Chvatal79} and 
\cite[Theorem~5.2 and Corollary~4.4]{Mulzer18}) implies that 
the failure probability, of choosing more than $\xi_0 = 2a/b$ 
lower functions, is at most $e^{-a/(10b)} \le e^{-1/10}$.
Hence, we have $q_\tau = \Omega(1/b^6)$.
To complete the Clarkson-Shor analysis, 
we need an upper bound on the 
number of prisms in $\VD_{\leq \xi}(R)$. By 
Lemma~\ref{lem:vdpairs}, this is 
$O(|R|\xi\lambda_s(\xi))$. The analysis thus yields
\[
N_{\le a,\le b} = O\left(b^6 np\cdot \xi\lambda_s(\xi)\right)
= O\left(b^6 (n/b)(a/b)\lambda_s(a/b)\right)
= O\left( nab^4 \lambda_s(a/b)\right) ,
\]
as asserted. Again, the bound holds trivially if $b > n/10$ or
if $n$ is constant.
\end{proof}

\paragraph{Remark.}
As usual, the bounds generalize to superlinear lower envelope complexity.
Let $\psi: \mathbb{N} \rightarrow \mathbb{N}$ be such
that any $m$ functions in $\F$ have a lower envelope of 
complexity at most $\psi(m)$, and suppose that
$m \mapsto \psi(m)/m$ is monotonically increasing. 
Then we obtain the following bound.
\begin{lemma} \label{lem:ab_general}
The number of prisms $\tau$ with $w^-(\tau)\le a$ 
and $w_0(\tau)\le b$ is $O(b^6\psi(n/b))$ for $a\le b$,
and $O(a^2b^4\psi(n/a)\lambda_s(a/b))$ for $a>b$.
\end{lemma}
\begin{proof}
The reasoning is analogous to that in the proof
of Lemma~\ref{lem:ab}, but we replace the bounds from Lemma~\ref{lem:vdpairs}
with the bounds from Lemma~\ref{lem:vdpairs_general}.
In particular, in Case~1, the vertical decomposition
$\VD_{\leq 6}(R)$ has $O(\psi(n/b))$ prisms, and
in Case~2, the vertical decomposition $\VD_{\leq \xi}(R)$ has
\[
O\left((a/b)^2\psi((n/b)(b/2a)) \lambda_s(2a/b)\right)
= O\left((a/b)^2\psi(n/a) \lambda_s(a/b)\right)
\]
prisms. (Since we assumed that $\psi(m)/m$ is increasing,
we have $\psi(n/2a) \leq \psi(n/a)$.)
\end{proof}

\begin{figure}[htb]
\begin{center}
\includegraphics{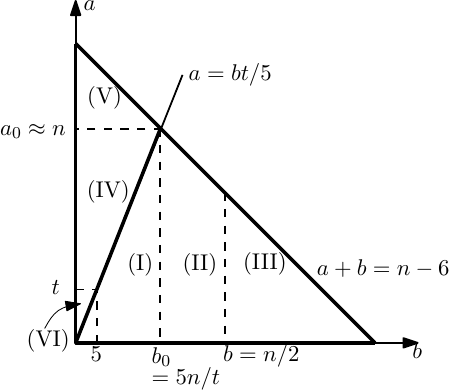}
\caption{The decomposition of the $(a,b)$-range into subranges.}
\label{levranges}
\end{center}
\end{figure}

We can now combine all the bounds derived so far, 
and bound (i) the expected number of prisms
that are ever generated in the RIC, and (ii) the 
expected overall size of their conflict lists,
which, as explained above, dominates the running 
time of the algorithm (with an additional logarithmic factor).
The expected number of prisms is simply
\begin{equation} \label{sumptau}
\sum_{\tau\in\Pi} p_\tau = \sum_a \sum_b p(a,b) N_{a,b}.
\end{equation}
Similarly, the expected overall size of the 
conflict lists is
\begin{equation} \label{sumbptau}
\sum_{\tau\in\Pi} w_0(\tau)p_\tau = \sum_a \sum_b bp(a,b) N_{a,b}.
\end{equation}
We bound these sums separately for pairs 
$(a,b)$ within each of the six regions 
depicted in Figure~\ref{levranges}. Together, 
these regions cover the entire range $a,b\ge 0$, $a+b\le n-6$.  
As will follow from the forthcoming analysis, the most 
expensive prisms are those for which 
$(a,b)$ lies in region (I) or region (IV).

\paragraph{Region (I).}
In this region, $5\le b\le 5n/t$ and $0\le a\le bt/5$.
We cover the region by vertical slabs of the form 
$S_j := \{(a,b) \mid b_{j-1} \le b\le b_j\}$, 
for $j\ge 1$, where $b_j=5\cdot 2^j$, for $j=0,\dots,\lceil \log (n/k) \rceil$. 
Within each slab $S_j$, the maximum 
value of $p_\tau$ is $O(1/b_{j-1}^6) = O(1/2^{6j})$, 
and we bound $\sum_{(a,b)\in S_j} N_{a,b}$ by
$N_{\le b_jt/5,\le b_j}$ which, by Lemma~\ref{lem:ab}, is
\[
O\left(n(b_jt/5)b_j^4\lambda_s(t/5) \right) =
O\left(nb_j^5 t\lambda_s(t) \right) =
O\left(2^{5j}nt\lambda_s(t) \right).
\]
Hence, the contribution of $S_j$ to (\ref{sumptau}) is at most
\[
O\left( \frac{nt\lambda_s(t)}{2^j} \right) ,
\]
and, summing this over $j$, we get that the 
contribution of region (I) to (\ref{sumptau})
is $O(nt\lambda_s(t))$.
Similarly, the contribution of $S_j$ to (\ref{sumbptau}) 
is at most
\[
O\left( b_j \cdot \frac{nt\lambda_s(t)}{2^j} \right) =
O\left( nt\lambda_s(t) \right).
\]
We need to multiply this bound by the number 
of slabs, which, as is easily 
checked, is $O(\log(n/t))$.
Hence, the contribution of region (I) to (\ref{sumbptau})
is $O(nt\lambda_s(t)\log(n/t))$.

\paragraph{Region (II).}
In this region, $5n/t \le b \le n/2$ and $0\le a\le n-6-b$.
Here too we cover the region by vertical slabs of the form 
$S'_j := \{(a,b) \mid b'_{j-1} \le b\le b'_j\}$, for $j\ge 1$, 
where $b'_j=(5n/t)\cdot 2^j$, for $j=0,\dots,\lceil\log(t/10)\rceil$. 
Within each $S'_j$, the maximum value of $p_\tau$ is 
$O(1/(b'_{j-1})^6) = O(t^6/(n^62^{6j}))$, and 
we bound $\sum_{(a,b)\in S'_j} N_{a,b}$ by
$N_{\le n-b'_{j-1}-6,\le b'_j}$ which, by 
Lemma~\ref{lem:ab}, is (upper-bounding $n-b'_{j-1}-6$ by $n$)
\[
O\left(n^2(b'_j)^4\lambda_s(n/b'_j) \right) =
O\left(n^2(n/t)^4 2^{4j} \lambda_s(t/2^j) \right) =
O\left(2^{3j}n^6\lambda_s(t)/t^4 \right).
\]
Hence, the contribution of $S'_j$ to (\ref{sumptau}) is at most
\[
O\left( t^2\lambda_s(t)/2^{3j} \right),
\]
and, summing this over $j$, we get that the 
contribution of region (II) to (\ref{sumptau})
is $O(t^2\lambda_s(t)) = O(nt\lambda_s(t))$.
Similarly, the contribution of $S'_j$ to 
(\ref{sumbptau}) is at most
\[
O\left( b'_jt^2\lambda_s(t)/2^{3j} \right) =
O\left( nt\lambda_s(t)/2^{2j} \right) ,
\]
and, summing this over $j$, we get that 
the contribution of region (II) to 
(\ref{sumbptau})
is $O(nt\lambda_s(t))$.

\paragraph{Region (III).}
In this region, $n/2 \le b \le n$ and 
$0\le a\le n-6-b$.  We treat this region 
as a single entity. The maximum value of 
$p_\tau$ here is $O(1/n^6)$, and we bound 
$\sum_{(a,b)\in {\rm (III)}} N_{a,b}$ by the 
overall number of prisms, which is $O(n^6)$,
getting a negligible contribution to 
(\ref{sumptau}) of only $O(1)$. A similar 
argument shows that the contribution of 
this region to (\ref{sumbptau}) is $O(n)$, again 
negligible compared with
the other regions.

\paragraph{Region (IV).}
In this region, $t\le a\le a_0\approx n$ and 
$0\le b\le 5a/t$.  
We cover the region by horizontal slabs of 
the form $S''_j := \{(a,b) \mid a_{j-1} \le a\le a_j\}$,
for $j\ge 1$, where $a_j=t\cdot 2^j$, for 
$j=0,\dots,\lceil \log(a_0/t)\rceil = O\left(\log (n/t)\right)$.
Within each slab $S''_j$, the maximum value of 
$p_\tau$ is $O(t^6/a_{j-1}^6) = O(1/2^{6j})$, 
and we bound $\sum_{(a,b)\in S''_j} N_{a,b}$ by
$N_{\le a_j,\le 5a_j/t}$ which, by 
Lemma~\ref{lem:ab}, is
\[
O\left(na_j(5a_j/t)^4\lambda_s(t/5) \right) =
O\left(na_j^5 \lambda_s(t)/t^4 \right) =
O\left(2^{5j}nt\lambda_s(t) \right).
\]
Hence, the contribution of $S''_j$ to 
(\ref{sumptau}) is at most
\[
O\left( \frac{nt\lambda_s(t)}{2^j} \right) ,
\]
and, summing this over $j$, we get that 
the contribution of region (IV) to 
(\ref{sumptau})
is $O(nt\lambda_s(t))$.
Similarly, the contribution of $S''_j$ to 
(\ref{sumbptau}) is at most
\[
O\left( (5a_j/t)\frac{nt\lambda_s(t)}{2^j} \right) =
O\left( nt\lambda_s(t) \right).
\]
Since the number of slabs is $O(\log(n/t))$, 
region (IV)  contributes
$O(nt\lambda_s(t)\log(n/t))$
to (\ref{sumbptau}).

\paragraph{Region (V).}
In this region, $a_0\le a\le n$ and $0\le b\le n-a-6$.
We treat this region as a single entity. 
The maximum value of $p_\tau$ in this region 
is $O(t^6/n^6)$, and we bound $\sum_{(a,b)\in {\rm (V)}} 
N_{a,b}$ by
\[
N_{\le n,\le 5n/t} = O\left( n^2(5n/t)^4\lambda_s(t/5)\right) = 
O(n^6\lambda_s(t)/t^4).
\]
Hence, the contribution of region (V) 
to (\ref{sumptau}) is at most
\[
O\left( t^2 \lambda_s(t) \right) = O(nt\lambda_s(t)).
\]
For the contribution to (\ref{sumbptau}), 
we multiply this bound by $O(n/t)$, an 
upper bound on $b$ in this region, and get
\[
O\left( (n/t)\cdot t^2 \lambda_s(t) \right) = O(nt\lambda_s(t)).
\]

\paragraph{Region (VI).}
Finally, we consider this region, 
which is given by $0\le a\le t$ and $0\le b\le 5$.
Here we upper-bound $p_\tau$ simply by $1$, 
and bound $\sum_{(a,b)\in {\rm (VI)}} N_{a,b}$ by 
$N_{\le t,\le 5}$, which is $O(nt\lambda_s(t))$. 
Hence, the contribution of region (VI) to (\ref{sumptau}) 
is at most
$O(nt\lambda_s(t))$.
Since $b$ is bounded by a constant in this 
region, the same expression also bounds
the contribution of region (VI) to (\ref{sumbptau}).

In conclusion, taking the additional logarithmic factor into
account, we have the following main result of this section.
\begin{theorem} \label{thm:ric}
The first $t$ levels of an arrangement of 
the graphs of $n$ continuous totally 
defined algebraic functions of constant 
description complexity, for which the 
complexity of the lower envelope of any $m$ 
functions is $O(m)$, can be constructed 
by a randomized incremental algorithm, 
whose expected running time is 
$O(nt\lambda_s(t)\log(n/t)\log n)$,
and whose expected storage is $O(nt\lambda_s(t))$.
\end{theorem}

\paragraph{Remark.}
The result for superlinear lower envelope complexity is as follows.
\begin{theorem} \label{thm:ric_general}
The first $t$ levels of an arrangement of 
the graphs of $n$ continuous totally 
defined algebraic functions of constant 
description complexity, for which the 
complexity of the lower envelope of any $m$ 
functions is at most $\psi(m)$, where $m \mapsto \psi(m)/m$ increases
monotonically, can be constructed 
by a randomized incremental algorithm, 
with expected running time 
\[
O(t^2\psi(n/t)\lambda_s(t)\log(n/t)\log n) ,
\]
and expected storage 
\[
O(t^2\psi(n/t)\lambda_s(t)).
\]
\end{theorem}
\begin{proof}
We again consider the six regions for $(a, b)$,
but with the bounds from Lemma~\ref{lem:ab_general}.
In Region~(I), the bound on $N_{\leq{b_jt/5},\leq b_j}$ 
within a slab $S_j$ is
$O(2^{6j} t^2 \psi(n/(2^j t)) \lambda_s(t))$,
so the contriubtion of $S_j$ is at most
$O(t^2\psi(n/(2^j t)) \lambda_s(t))$,
resulting in a total contribution of
$O(t^2\psi(n/t) \lambda_s(t))$ for Region~(I)
to the number of prisms, and 
$O(t^2\psi(n/t) \lambda_s(t)\log(n/t))$ to the conflict size.
In Region~(II), the bound on $N_{\leq{n},\leq b_j'}$ 
within a slab $S_j'$ is
$O(\psi(1)\cdot n^6 2^{3j} \lambda_s(t)/t^4)$.
Thus, the bound for Region~(II) follows as before.
The argument for Region~(III) is also the same.
In Region~(IV), the bound on $N_{\leq a_j,\leq 5a_j/t}$
is $O(t^22^{6j}\psi(n/(t2^j))\lambda_s(t))$,
so the contribution of the slab $S_j''$ is 
$O(t^2\psi(n/(t2^j)) \lambda_s(t))$.
By our assumption on $\psi$, this sums to a
total contribution of $O(t^2\psi(n/t) \lambda_s(t))$
to the number of prisms and of 
$O(t^2\psi(n/t) \lambda_s(t)\log(n/t))$ to the 
conflict size. In Region~(V), the
bound on $N_{\leq n, \leq 5n/t}$ is
$O(\psi(1) n^6 \lambda_s(t)/t^4)$, and the
argument proceeds as before. Finally, in Region~(VI),
the bound on $N_{\leq t, \leq 5}$ is
$O(t^2\psi(n/t)\lambda_s(t))$. The claim follows.
\end{proof}

\section{Improved Dynamic Maintenance of Lower Envelopes for 
  Planes}\label{sec:chan}

We present our interpretation
of Chan's technique for dynamically maintaining the
lower envelope of a set $H$ of non-vertical planes
in $\reals^3$ under insertions and deletions.
In the next section, we will combine this structure with the results
from the previous sections to obtain a data structure that supports 
polylogarithmic update and query time for maintaining the lower 
envelope of general surfaces. As before, maintaining the lower 
envelope means that we can efficiently answer \emph{point location} 
queries, each specifying a point $q\in\reals^2$ and asking for the 
lowest plane above $q$.

Our exposition proceeds in three
steps: We begin with a static structure, then develop a
simple variant (with a not so simple analysis) of a standard 
technique for handling insertions, and finally describe how to 
perform deletions.  With the help of a simple charging argument
in the analysis of the static
structure, we also improve Chan's original
(amortized) deletion time by a logarithmic factor, from $O(\log^6n)$
to $O(\log^5n)$, the first improvement in more than ten years.
Subsequently to this work, Chan found an additional improvement,
resulting in amortized deletion time $O(\log^4 n)$ and 
amortized insertion time $O(\log^2 n)$~\cite{Chan19}.
This further improvement is based on the improvement presented in 
this section, combined with an improved algorithm for the data 
structure used in the static case.

\subsection{A static structure}\label{sec:static}

Let $H$ be a fixed set of $n$ non-vertical planes in $\reals^3$.
We fix a constant $k_0 \geq 1$, and define the
sequence $k_j = 2^jk_0$, for $j=0,1,\dots,m$, where
$m = \lfloor \log(n/k_0)\rfloor$. We have $k_m \in (n/2, n]$.

Next, we construct a sequence $\{\Lambda_j\}_{j\ge 1}$ of vertical
shallow cuttings, for $j = m + 1, \dots, 0$, as follows:
the shallow cutting $\Lambda_{m+1}$ consists of a single
prism $\tau$ that covers all of $\reals^3$ and has conflict
list $\CL(\tau) = H$ (the shallowness is vacuously satisfied here).
Next, we set $H_m = H$ and $n_m = n$, and we start with $j = m$.
In round $j$, we have a set $H_j \subseteq H$ of $n_j \le n$
\emph{surviving} planes, and we construct a vertical $k_j$-shallow
$(\alpha k_j/n_j)$-cutting $\Lambda_j$ for $\A(H_j)$,
for a suitable constant $\alpha > 1$
(the same constant for all rounds; see Section~\ref{sec:prelim}, and
also~\cite{CT15,HKS}).  The reason for using this hierarchy of 
cuttings will become clear when we discuss deletions, in 
Section~\ref{sec:delete}.  That is, $\Lambda_j$ consists of 
$O(n_j/k_j)$ semi-unbounded vertical prisms whose ceilings
form the faces of a polyhedral terrain $\oLambda_j$.
The terrain $\oLambda_j$ lies fully above level $k_j$,
and the \emph{conflict list} $\CL(\tau)$ of each prism
$\tau \in \Lambda_j$, consisting of the planes from $H_j$
that cross $\tau$, is of size at most $\alpha k_j$.
After constructing $\Lambda_j$, we
identify the set $H^\times \subseteq H_j$ of planes
that belong to more than $c\log n$ conflict
lists in \emph{all} cuttings $\Lambda_m, \dots, \Lambda_j$
constructed so far;\footnote{%
  We note, for the expert reader, that this is the point where
  our construction improves over Chan's original result,
  since Chan's pruning strategy considers each cutting individually.
  Our approach ensures that each plane appears in
  $O(\log n)$ conflict lists in a substructure,
  whereas in Chan's structure the bound is $O(\log^2 n)$.
  Lemma~\ref{lem:discarded} shows that the more aggressive
  pruning does not remove too many planes.}
here $c$ is some sufficiently large constant that we will specify
later. We remove the planes in $H^\times$ from the conflict
lists of $\Lambda_j$ (but not from those of the higher levels 
$\Lambda_{j'}$, for $j'>j$), and we set 
$H_{j-1} = H_j \setminus H^\times$ and $n_{j-1} = |H_{j-1}|$ for the
next round. This \emph{pruning mechanism}
ensures that each plane of $H$ appears in at most $c\log n$
(pruned) conflict lists, within the current hierarchy of cuttings,
which is crucial for the efficiency of
the dynamic algorithm, to be presented in the next two subsections.
Note that $\oLambda_{j-1}$ is not necessarily ``lower'' than
$\oLambda_j$, since the former cutting is constructed with respect
to a potentially smaller set of planes (and can therefore contain 
points that lie above $\oLambda_j$, even though it approximates a 
lower-indexed level).  We stop after performing the pruning step for
$\Lambda_0$.  The conflict lists of $\Lambda_0$ will be used for 
answering queries.  (To support deletions, we will also need the 
conflict lists of $\Lambda_j$, for $j \geq 1$; see below.)
We denote by $\D^{(1)}$ the structure consisting of the
shallow cuttings $\Lambda_{m+1},  \Lambda_m, \dots, \Lambda_0$ and 
the pruned conflict lists of their prisms.
We write $S(\D^{(1)})$ for the set of surviving
planes in $\D^{(1)}$ (those that were not pruned at any stage),
and $C(\D^{(1)}) = H$ for the initial set of
planes. We say that $S(\D^{(1)})$ is \emph{stored} in
$\D^{(1)}$ and that $C(\D^{(1)})$ was used to
\emph{construct} $\D^{(1)}$. The planes in
$C(\D^{(1)})\setminus S(\D^{(1)})$ are the
\emph{pruned} planes in $\D^{(1)}$.
We emphasize again that in general a pruned plane still shows up in 
some conflict lists of some of the cuttings $\Lambda_j$ (for 
indices $j$ higher than the one at which it was pruned).
The difference is that the stored planes are kept only
in $\D^{(1)}$, whereas the pruned planes are also passed
to subsequent substructures, as we will soon describe.
The following lemma bounds the size of $S(\D^{(1)})$.
\begin{lemma}\label{lem:discarded}
For any $\zeta\in (0,1)$ there exists a sufficiently large (but 
constant) choice of $c$ (the coefficient in the threshold $c\log n$,
beyond which we prune planes), so that 
$|S(\D^{(1)})| \ge (1-\zeta)n$.
\end{lemma}

\begin{proof}
Define the potential $\Phi(j)$ as the
total ``stored size'' of the conflict lists
of $\Lambda_m, \dots, \Lambda_j$, after the pruning step in
round $j$, and set $\Phi(m+1) = 0$.
By the stored size of a conflict list $\CL(\tau)$, of some prism 
$\tau$ of some cutting, we mean
the number of planes in $\CL(\tau)$ that have not been pruned yet,
at any of the steps $m,m-1,\dots,j$.\footnote{%
  Some of these planes might be pruned, later, from the conflict 
  lists of some lower-indexed cuttings, and then they will no longer
  be counted in the stored size of $\tau$.}
Since $\Lambda_j$ has $O(n_j/k_j)$ prisms, and each conflict list
of $\Lambda_j$ has $O(k_j)$ planes, the overall size of the conflict
lists of $\Lambda_j$ is $O(n_j)=O(n)$. Hence,
in round $j$, we increase $\Phi$ by at most $\gamma n$,
for some fixed $\gamma > 0$, where the increase is caused by
the conflict lists of the prisms of the new cutting.
If a plane $h$ is pruned at this stage, then $h$ lies in at least 
$c\log n$ conflict lists of all stages processed so far (including 
those at the present stage), and has to be removed
from the stored size count of these prisms, so this decreases $\Phi$
by at least $c\log n$. Since $\Phi$ is initially $0$ and is never 
negative, and since we increase it by at most $\gamma n\log n$ units
throughout the construction, it follows that we prune at most 
$\zeta n$ planes, if we choose $c \geq \gamma/\zeta$.
\end{proof}

We fix the fraction $\zeta = 1/32$, and use the corresponding 
coefficient $c$ in the construction. We set $H^{(1)} = H$ and
$H^{(2)} = C(\D^{(1)})\setminus S(\D^{(1)}) = H\setminus S(\D^{(1)})$
(that is, $H^{(2)}$ consists
of the planes that we have pruned). As just argued, we have
$|H^{(2)}| \le \zeta n$. We repeat the process with the set
$H^{(2)}$, obtaining an analogous structure
$\D^{(2)}$ and a remainder set
$H^{(3)} = C(\D^{(2)})\setminus S(\D^{(2)}) = 
H^{(2)}\setminus S(\D^{(2)})$ of at most $\zeta^2 n$ planes that 
were pruned in $\D^{(2)}$.  Proceeding in this manner, for 
$s\le \log_{1/\zeta} n = \frac15 \log n$ steps, we obtain the 
complete target (static) structure 
$\D = (\D^{(1)},\D^{(2)},\dots,\D^{(s)})$, where $\D^{(s)}$ involves
only a constant number of planes.  Thus, in the last step, the 
overall size of the conflict lists is (much) smaller than $c\log n$,
so all of these planes are stored, and the process can `safely' 
terminate.  By construction, the sets $S(\D^{(i)})$ are pairwise 
disjoint and their union is $H$. For each $i$, let 
$m_i = \lfloor \log(|H^{(i)}|/k_0)\rfloor = O(\log n)$ be the number
of cuttings in $\D^{(i)}$.  The overall size of $\D^{(i)}$, including
the conflict lists, is
\[
O\left( \sum_{j=0}^{m_i} \frac{|H^{(i)}|}{k_j}\cdot k_j \right) =
O\left( |H^{(i)}|m_i \right) =
O\left( |H^{(i)}|\log n \right).
\]
Since by Lemma~\ref{lem:discarded}, we have $|H^{(i)}|\le 
\zeta^{i-1}n$, the overall size of $\D$ is $O(n\log n)$.\footnote{%
  Note that if we did not have to account for the conflict lists,
  the sum would have been $O(n)$. Later, we will reduce the storage,
  including the conflict lists, to linear, by
  representing the conflict lists implicitly.}
Using the algorithm of Chan and Tsakalidis~\cite{CT15},
we construct each cutting $\Lambda_j$, including its associated
conflict lists, in any $\D^{(i)}$, in $O(|H^{(i)}|\log n)$
time. Summing over $j$, we can construct
$\D^{(i)}$ in $O(|H^{(i)}|\log^2 n)$ time, and summing over $i$,
recalling that $|H^{(i)}|$ decreases geometrically with $i$,
we get a total running time $O(n\log^2 n)$.\footnote{We note, again,
that the recent improvement of Chan~\cite{Chan19} comes
from decreasing the running time of this preprocessing
step to $O(n \log n)$, while the rest of the
algorithm remains essentially unchanged.}
We write this
bound as $an\log^2n$, for a suitable concrete coefficient $a$.

Answering a query is easy: Given $q \in \reals^2$, we iterate over
all $O(\log n)$ substructures $\D^{(i)}$, and we find the prism
$\tau$ of the cutting $\Lambda_0$ in $\D^{(i)}$ whose
$xy$-projection contains $q$ (or possibly the $O(1)$ prisms, if the
query $q$ is not in general position and falls on the boundary of
several such prisms).
This takes $O(\log n)$ time, with a suitable point-location
structure for the $xy$-projection of $\oLambda_0$.
We then search the conflict list $\CL(\tau)$ of $\tau$,
in brute force, for the lowest plane over $q$ (or possibly
planes, in case that $q$ is not in general position).
This requires $O(k_0) = O(1)$ time.
We return the plane (or planes) lowest over $q$ among all $O(\log n)$
candidates. The query time is thus $O(\log^2 n)$.
The correctness of this procedure is obvious, and follows from
the property that the ceilings of the prisms in $\Lambda_0$
(for any $\D_i$) pass above level $k_0\ge 1$ of $\A(H^{(i)})$,
so the lowest plane of $H^{(i)}$ over $q$ belongs to the conflict
list of the corresponding prism.

\subsection{Handling insertions}\label{sec:insert}

Here, we explain how to maintain the lower envelope of
a set $H$ of non-vertical planes in $\reals^3$ under insertions,
where the notion of maintenance is as defined in the preceding 
section.  For simplicity, at each point in time, we denote the 
current number of planes in the data structure by $n$. Furthermore,
we use $N$ to denote the power of $2$ satisfying
$n \in [N, 2N)$. Whenever $n$ becomes too large, we double
the value of $N$. Our structure uses a variant of a standard 
technique, introduced by Bentley and Saxe~\cite{BS} and later 
refined in Overmars and van Leeuwen~\cite{OvL} (see also Erickson's
notes~\cite{Er}).  Specifically, we maintain a sequence
$\I = (\DD_{i_0},\DD_{i_1},\dots,\DD_{i_k})$ of structures, 
where $0\le i_0 < i_1 < \dots < i_k $. The indices $i_j$ are not
fixed, are not necessarily contiguous, and may change after each 
insertion. Informally, we have an infinite sequence of bins, indexed
by $0,1,2,\dots$, and $i_j$ is the index of the bin that stores 
$\DD_{i_j}$, for $j=0,1,\dots,k$. We refer to $\DD_{i}$ as the 
structure at \emph{location} (or bin) ${i}$. Lemma~\ref{lem:lenI} 
below shows that $i_k$, and thus also $k$, are $O(\log n)$.
Each $\DD_{i_j}$ is a \emph{substructure} $\D^{(u)}$ of some static 
structure $\D$, as in Section~\ref{sec:static}, constructed over 
some subset of $H$.
We maintain the following invariants.
\begin{itemize}
\item[(I1)]\label{inv:inv1}
For each occupied index $i$, we have $2^{i-1} < |S(\DD_{i})| 
\leq 2^{i}$.
\item[(I2)]\label{inv:inv2}
The sets $S(\DD_{i})$, over the occupied indices $i$, are pairwise 
disjoint, and their union is $H$.
\end{itemize}
For each plane $h \in H$, we
say that $h$ \emph{is stored} at location ${i}$
if $h\in S(\DD_{i})$.

\paragraph{Inserting a plane.} To insert a plane $h$, we determine 
the smallest index of an empty bin,
i.e., the smallest integer $j \geq 0$ that is not in the sequence
$(i_0, i_1, \dots, i_k)$. If $j = 0$, we store $h$ in a trivial
structure $\DD_0$ of size $1$. Otherwise, we have
$(i_0, \dots, i_{j-1}) = (0,\dots, j-1)$, and we set
$\B_j := \left(\bigcup_{i=0}^{j-1} S(\DD_{i}) \right)\cup\{h\}$.
Assuming inductively that Invariants~(I1) and (I2) hold prior
to the insertion of $h$, we have
\begin{align}
|\B_j| &= 1 + \sum_{i=0}^{j-1} |S(\DD_{i})|
> 1 + \sum_{i=0}^{j-1}2^{i-1} > 2^{j-1} ,
\label{eq:eq1a}
\\
\intertext{and}
|\B_j|&=1 + \sum_{i=0}^{j-1}
|S(\DD_{i})| \le  1+\sum_{i=0}^{j-1}
2^{i} = 2^j.
\label{eq:eq1b}
\end{align}
We construct over $\B_j$ a static structure $\D$, as in 
Section~\ref{sec:static}.
Recall that $\D$ is a sequence of a logarithmic number of 
substructures, $\D^{(1)},\D^{(2)},\dots,\D^{(s)}$, where (recall 
Lemma~\ref{lem:discarded})
\[
s \le \log_{1/\zeta} |\B_j| = \frac{\log |\B_j|}{\log (1/\zeta)} 
\le \frac{j}{\log (1/\zeta)} = \frac{j}{5} < j ,
\]
by~(\ref{eq:eq1b}) and our choice of $\zeta = 1/32$.
We remove the current structures $\DD_{0},\dots, \DD_{j-1}$ from 
$\I$.  Then, for each substructure $\D^{(u)}$ constructed for $\B_j$,
for $u=1,\dots,s$, we set $\DD_{i_u} := \D^{(u)}$ for 
$i_u  = \lceil \log |S(\D^{(u)})| \rceil$, and add $\DD_{i_u}$ 
to $\I$.

We have $C(\D^{(1)}) = H_j$. By Lemma~\ref{lem:discarded} 
and~(\ref{eq:eq1a}), $|S(\D^{(1)})| \ge (1-\zeta) |H_j| 
\ge (1-\zeta) 2^{j-1} > 2^{j-2}$ (recall that $\zeta = 1/32$),
so it follows that $\D^{(1)}$ is placed in bin $j$ or $j-1$.
Moreover, by Lemma~\ref{lem:discarded}, we have
$|S(\D^{(u+1)})| \le \zeta|C(\D^{(u)})| \le 
\frac{\zeta}{1-\zeta} |S(\D^{(u)})|$,
so the corresponding indices satisfy
\begin{align*}
i_{u+1} & = \lceil \log |S(\D^{(u+1)})| \rceil \le
1 + \log |S(\D^{(u+1)})| \\
& \le 1 + \log \Bigl( \frac{\zeta}{1-\zeta} |S(\D^{(u)})| \Bigr)
= 1 + \log |S(\D^{(u)})| - \log 31 \\
& = \log |S(\D^{(u)})| - \log 15.5 \le 
\lceil \log |S(\D^{(u)})|\rceil - \log 15.5 = i_u - \log 15.5 ,
\end{align*}
since $\zeta = 1/32$. That is, since both $i_u$ and $i_{u+1}$ are 
integers, $i_{u+1} \le i_u - 4$.
Hence, each structure $\D^{(u)}$ is assigned a different index
$i \leq j$ (with a gap of at least three empty bins between 
consecutive occupied ones), and Invariants (I1) and (I2) hold by 
construction ((I1) follows from the definition of the indices $i_u$).

\paragraph{Answering a query.}
To answer a query $q$, we find in each substructure $\DD_{i}$ of $\I$
the prism (or prisms, in case of a non-generic query) 
of the corresponding lowest-index cutting $\Lambda_0$
whose $xy$-projection contains the query point $q$, and we search 
over the at most $k_0$ planes of its conflict list for the lowest 
ones over $q$. We output the lowest among all these planes, over 
all substructures.  This takes $O(\log n)$ time per structure, as 
in the static case.

The correctness of this procedure follows from Invariant (I2). 
Indeed, if $h$ is a lowest plane over a query point $q$, then 
$h \in S(\D^{(i)})$ for some unique $i$ (by Invariant (I2)). Since 
$h$ is stored at $\D^{(i)}$, it has not been pruned from any conflict
list of this substructure. Let $\tau$ be a prism of the 
lowest-indexed cutting $\Lambda_0$ of $\D^{(i)}$ whose 
$xy$-projection contains $q$. By construction, the ceiling of 
$\tau$ lies above the $k_0$-level of the corresponding arrangement, 
which implies that $h$ must lie in $\tau$ over $q$.
This, and the fact that $h$ has not been pruned, implies that $h$
belongs to the (pruned) conflict list $\CL(\tau)$, so the query will
encounter $h$, and thus will output it as the correct answer.
The following lemma bounds the size of $\I$.
\begin{lemma}\label{lem:lenI}
The largest index of an occupied bin is at most $\log N  + 1$, and 
thus the number of structures
in $\I$ is at most $\log N  + 2$.
\end{lemma}
\begin{proof}
By definition, the number of planes in the structure 
satisfies $n < 2N$, and by Invariant~(I1),
the largest index $i\in \I$ satisfies
\begin{align*}
&2^{i-1} < |S(\DD_i)| \le n < 2N, \\
\text{or}\quad &i < 1+ \log (2N) = \log N + 2 ,
\end{align*}
that is $i \le \log N + 1$ (since $\log N$ is an integer).
\end{proof}

\paragraph{Running Time.}
We now bound the running time of our insertion-only
structure. In particular, we are going to show
that the deterministic amortized cost of an insertion 
is $O(\log^3 n)$ and the deterministic worst-case
cost of a query is still $O(\log^2 n)$, as in the static
case.

Indeed, the bound on the query time is immediate by Lemma~\ref{lem:lenI}
and the observation preceding it, because $\log N = O(\log n)$.
To analyze the amortized insertion cost, we use a charging argument.
Each plane $h$ that is currently stored in $\I$ holds $b(w-i)$
\emph{credits}, each worth $a\log^2 N$ \emph{units}, where $i$
is the current unique location (bin index) of the structure $\DD_i$
for which $h\in S(\DD_i)$, and $w:= \log N + 2$ bounds,
by Lemma~\ref{lem:lenI}, the maximum length of $\I$.
Here $a$ is the coefficient of the bound $an\log^2n$ for the
construction time of the static structure on $n$ planes (see
Section~\ref{sec:static}), and $b$ is some absolute constant
to be fixed shortly. Note that the number of credits held by
a plane $h$ is larger when the bin index where $h$ is stored is 
smaller.

We define the \emph{potential} $\Psi$ of the structure as
the overall number of units of credit that its planes hold.
The amortized cost of an insertion is defined to be $bw$
credits, that is, $abw\log^2N$ units.

When a new plane is inserted, we give it $bw$ credits, that is, 
$abw\log^2N$ units of credit. The unit size depends on $N$, 
so the whole charging scheme has to be updated every time $N$ is 
doubled. Specifically, when $N$ is doubled, the increase in the size
of a credit is
\[
a\log^2(2N) - a\log^2N = a(1+\log N)^2 - a\log^2N = a(1+2\log N).
\]
We have $2N$ planes in the structure at this moment, and each of them
carries at most $bw$ credits. Hence, updating the overall amount of
credit in the structure in this
doubling costs $O(Nw\log N)$ units, which is $O(N\log^2N)$.
There are only $O(\log n)$ doubling steps, and the $N$'s that
they involve are powers of $2$, implying that the overall additional 
cost of updating the credit distribution during doublings of $N$ is
$O(n\log^2n)$. In what follows we ignore this issue of updating
the credit size, whose cost will be subsumed by the overall cost 
of the insertions, and just use the number of credits in our charging
scheme.\footnote{%
  Again, Chan's improvement~\cite{Chan19} is reflected in this scheme
  by reducing the size of a credit to $O(\log N)$, leaving the rest 
  of the analysis unchanged.}

We recall that, when inserting a new plane $h$, we destroy a prefix 
of $j$ substructures in $\I$, put all their planes, including $h$,
in a subset $\B_{j}$, compute a new static structure $\D$ for $\B_j$,
and spread its substructures $\D^{(u)}$ 
into some subset of bins with indices from $j$ downwards.
Set $t = |\B_{j}|$. As noted above, by the analysis in
Section~\ref{sec:static}, the real cost of the insertion is at most
$at\log^2 t$, or, in other words, at most $t$ credits.
The main claim in the complexity analysis of insertions is the 
following lemma.
\begin{lemma}\label{lem:amortized}
With the above notations, for a sufficiently large choice of the
constant $b$, we have
\begin{equation} \label{stored:pot}
at\log^2 t + \Delta\Psi \le abw\log^2N ,
\end{equation}
where $\Delta\Psi = \Psi^+ - \Psi^-$, where $\Psi^+$ and $\Psi^-$
are, respectively, the values of the potential just after 
and just before the insertion.
\end{lemma}
\begin{proof}
Consider the sequence $\D^{(1)},\D^{(2)},\dots,\D^{(s)}$, of 
substructures that we construct over $\B_j$, where 
$s=\lfloor\log_{1/\zeta} |\B_j|\rfloor$. 
Recall that by (\ref{eq:eq1a}) and (\ref{eq:eq1b}), we have that 
$2^{j-1} < t=|\B_{j}|  \le 2^j$. In particular, 
$s\le \frac{1}{5} \log t \le j/5 < j$.
As already noted, by Lemma~\ref{lem:discarded}, 
$|S(\D^{(1)})| \ge (1-\zeta)t \ge (1-\zeta)2^{j-1}$,
and consequently (since $\zeta = 1/32$), the structure 
$\D^{(1)}$ is placed either in 
bin $j$ or bin $j-1$. Furthermore, the lower bound in Invariant 
(I1) shows that 
before the insertion, at least $\sum_{i=1}^{j-2}2^{i-1} + 1 
> 2^{j-2}\geq t/4$ planes of $\B_j$ were stored at bins 
$i=0, 1, \dots, j - 2$ (the $ + 1$ is for bin $0$, which contains 
exactly one plane).  Since $|S(\D^{(1)})|\ge (1-\zeta)t$, 
it follows that at least $t/4 -\zeta t$ among the planes that 
were stored in bins 
$i=0,1,\dots, j-2$, end up at $\DD_j$ or $\DD_{j-1}$ following the 
insertion. These planes release at least $bt(1/4 -\zeta)$ credits, 
that is, they decrease $\Psi$ by at least these many credits.

The technical issue that we need to address is that some 
planes in $\D^{(2)},\dots,\D^{(s)}$
may require more credits than what they had before the insertion, 
if they end up in smaller-indexed bins than the bins in which they 
were stored before the insertion.  We claim that the overall amount
of this extra credit is much smaller than the amount
of released credit, so the released credit (more than) suffices to 
fill in the required extra credit, thereby making each plane hold 
the correct amount of credit, with change to spare. That is, the 
insertion does cause $\Psi$ to decrease.

To show this, let $(j-1)-j_i$ be the bin index in which we put 
$\D^{(i)}$, for $2\le i \le s$.
In the worst case, each plane of $\D^{(i)}$ was stored in 
bin $j-1$ before the insertion and
now requires $b j_i$ additional credits. (Note that 
$h$ does not participate in this argument: 
it cannot release any credit since it did not exist 
before the insertion; nevertheless,
the credits that it gets as it is being inserted more 
than suffice for any bin $h$ is 
eventually placed in.) Summing up, we get that the number 
of additional credits that 
we need to give the planes is at most $b\sum_{i=2}^s j_i |
S(\D^{(i)})|$. 
From the definition of the insertion mechanism, we get that
$(j-1)-j_i = \lceil \log |S(\D^{(i)})| \rceil$, so
the total number of additional credits that we need to give 
these planes is at most
\begin{equation} \label{eq:credit-increase}
b\sum_{i=2}^s j_i |S(\D^{(i)})| = 
b \sum_{i=2}^s \Bigl(j-1-\lceil\log |S(\D^{(i)})|\rceil \Bigr) 
|S(\D^{(i)})| \le b \sum_{i=2}^s \Bigl(j-1-\log |S(\D^{(i)})|\Bigr) 
|S(\D^{(i)})|.
\end{equation}
Since the function $x \mapsto (j-1-\log x )x$ is increasing for 
$0<x\le 2^{j-1-1/\ln 2} \approx 2^{j-2.442}$,
and $|S(\D^{(i)})| \le \zeta^{i-1} t < 2^{j-3}$ for every 
$i=2,\dots,s$ and for $\zeta = 1/32$, 
we can bound the last expression in~(\ref{eq:credit-increase})
by replacing $|S(\D^{(i)})|$ with $\zeta^{i-1} t$, 
for every $i=2,\dots,s$, so we get that
\[
b\sum_{i=2}^s \Bigl(j-1-\log |S(\D^{(i)})|\Bigr) | S(\D^{(i)}) | 
\le b\sum_{i=2}^s \Bigl(j- 1-\log(\zeta^{i-1} t)\Bigr) \zeta^{i-1} t.
\]
Since $t\ge 2^{j-1}$, we can bound the right-hand side by
\[
-bt\sum_{i=2}^s \log(\zeta^{i-1}) \zeta^{i-1} 
= bt\log\frac{1}{\zeta}\sum_{i=2}^s (i-1) \zeta^{i-1} 
\le 5bt\sum_{i=1}^\infty i\zeta^{i} 
= 5bt\frac{\zeta}{(1-\zeta)^2}.
\]
We conclude that, for $\zeta=\frac{1}{32}$, the sum is smaller than 
$bt/6$, so we still have more than
\[
bt\left(\frac14 - \frac{1}{32} \right) - \frac{bt}{6} = 
\frac{5bt}{96} 
\]
free credits. By construction, this free credit is 
$\Psi^- - \Psi^+ = -\Delta\Psi$, to which we add the $bw$ credits 
we gave $h$ upon insertion.
Hence, if $b$ is large enough, say at least $20$, we have
$bw - \Delta\Psi \ge t$. Since the real insertion cost is at most 
$t$ credits, it follows that the released credit suffices to pay 
for the real insertion cost (with change to spare), and the lemma 
follows.
\end{proof}

\begin{corollary}\label{lem:cost}
The overall cost of $n$ insertions is $O(n\log^3n)$.
\end{corollary}
\begin{proof}
As already mentioned, we ignore the changes in the size of the 
credits caused by doublings of $N$; as noted, this adds only 
$O(n\log^2n)$ to the overall cost.
We add up the inequalities~(\ref{stored:pot}), over all insertions, and get that the 
overall actual cost of the insertions plus 
$\Psi_{\text{final}} - \Psi_{\text{initial}}$,
is at most $O(nw\log^2n) = O(n\log^3n)$. Since
$\Psi_\text{final} - \Psi_\text{initial} \ge 0$, this also bounds 
the actual cost of $n$ insertions.
\end{proof}

We note that Chan's recent improvement~\cite{Chan19} follows from 
this corollary, simply be reducing the size of a single credit to 
$O(\log N)$, and leaving the rest of the analysis intact.
The following lemma summarizes the properties of our
insertion-only structure.
%
\begin{lemma}\label{lem:insert-only}
The deterministic amortized cost of an insertion
is $O(\log^3 n)$, and the deterministic worst-case
cost of a query is $O(\log^2 n)$.
\end{lemma}

\subsection{Handling deletions}\label{sec:delete}

Finally, we describe how to maintain the lower envelope of
a set $H$ of non-vertical planes
in $\reals^3$ under insertions and deletions. As before, we denote
the current number of planes in the data structure by $n$,
and we use $N$ to denote the power of $2$ with
$n \in [N, 2N)$. Now we add a \emph{global rebuilding}
mechanism: whenever the number of updates (insertions
and deletions) since the last global rebuild becomes $N/2$, 
we completely rebuild the data 
structure from scratch for the current set $H$, and we adjust 
(double or half) the value of $N$, if needed, to restore the 
size range invariant.\footnote{%
  Note that the actual value of $n$ does not have to change much,
  when the sequence of insertions and deletions is reasonably 
  balanced.} We will argue later that the overall 
cost of the rebuildings is subsumed in (and actually much smaller 
than) the cost of the other steps of the algorithm.

The basic organization of the data structure
is the same as in Section~\ref{sec:insert}, consisting
of a sequence of bins $\I = (\DD_{i_0},\DD_{i_1},\dots,\DD_{i_k})$,
where $0\le i_0 < i_1 < \dots < i_k $, occupied by substructures 
of some static structures.
For each such substructure, we continue to denote by $C(\DD_j)$ the 
set of planes
that it was constructed from, and by $S(\DD_j)\subseteq C(\DD_j)$ 
the subset of planes that survived (are stored) in it.\footnote{%
We remind the reader that $\DD_j$ may also contain non-stored, i.e., 
pruned planes.}

Insertions and queries are performed in much the same way 
as in Section~\ref{sec:insert},
although some aspects of their implementation and analysis are 
different; see below for details.
We delete a plane $h$ by visiting each substructure $\DD_j$ with 
$h \in C(\DD_j)$, and  marking $h$ as \emph{deleted} in each 
conflict list of $\DD_j$ that contains $h$ (note that this is done 
also for substructures in which $h$ has not survived the initial 
construction, because it was pruned at some level of the hierarchy).
Each plane can get marked, at the time it is deleted, up to 
$O(\log^2n)$ times, once for each conflict list that contains 
it at the time when the deletion takes place. (Recall that, at each
$\DD_j$ with $h\in C(\DD_j)$, $h$ appears in at most $O(\log n)$
(pruned) conflict lists, over the entire hierarchy constructed at 
$\DD_j$.)

Actual removal 
of $h$, albeit possibly only from some of the substructures, takes 
place during a global rebuild or when pieces of the data structure 
are rebuilt, either during an insertion of a new plane, or at 
certain steps of the procedure where conflict lists are purged and 
their elements are reinserted; see below for details.

For a substructure $\DD_i$, we denote by
$A(\DD_i) \subseteq S(\DD_i)$ the set of \emph{active} planes
in $\DD_i$, defined as those planes that (a) are in $S(\DD_i)$, 
(b) have not been (marked as) deleted, and (c) have not been 
reinserted into other substructures due to the lookahead deletion 
mechanism, which we describe next. When a substructure $\DD_i$ is 
created, we have $A(\DD_i) = S(\DD_i)$. Once $\DD_{i}$ is created, 
its associated sets $C(\DD_{i})$ and $S(\DD_i)$, as well as all 
non-purged conflict lists $\CL(\tau)$ of prisms $\tau$ in (any hierarchical 
stage within) $\DD_i$, remain fixed, until $\DD_i$ is destroyed
(in a rebuild triggered by an insertion or a reinsertion, or in a 
global rebuild).
On the other hand, conflict lists may be marked as
purged by the deletion-lookahead mechanism, and the set 
$A(\DD_i)$ of active planes in $\DD_i$ may get smaller, due to 
deletions and the purging of conflict lists.

\paragraph{Lookahead deletions.}
When too many planes in a conflict list $\CL(\tau)$, for some prism 
$\tau$, are (marked as) deleted, the real lower envelope of $H$ 
might rise too high, and the lowest (undeleted) plane over a query 
point $q$ (with $q$ lying in the projection of $\tau$) does no 
longer have to belong to $\CL(\tau)$. (Note that if $\tau$ belongs 
to the lowest-indexed cutting of some structure $\DD_j$, which is the
only kind of prisms we access when processing a query, its ceiling
lies above level $k_0$ of the corresponding set of planes, so, even 
after $k_0-1$ deletions from that set, the lowest plane over $q$ 
still belongs to $\CL(\tau)$, but a larger number of deletions may 
cause the difficulty just noted to arise.)

To avoid this situation, which might cause us to miss the correct 
answer to a query, we use the following \emph{lookahead deletion 
mechanism}. Suppose that, for
some prism $\tau$ in a cutting $\Lambda_i$ of some substructure
$\DD_j$, at least $|\CL(\tau)|/(2\alpha)$ planes in $\CL(\tau)$ 
have been marked 
as deleted,\footnote{Note that $\CL(\tau)$ can also contain planes 
from $C(\DD_i) \setminus S(\DD_i)$, and that these planes also count
for this test.}
where $\alpha>1$ is our cutting parameter (i.e., each prism of 
$\Lambda_i$ is intersected by at most $\alpha k_i = 2^i\alpha k_0$ 
planes, for any $i$). Then we \emph{purge} the conflict list 
$\CL(\tau)$, and we reinsert (only) the planes 
in $\CL(\tau) \cap A(\DD_i)$, one by one, using the standard 
insertion algorithm.  After this, $A(\DD_i)$ contains no more 
elements from $\CL(\tau)$, but elements from $\CL(\tau)$ may still 
appear in other conflict lists of $\DD_i$ and in $S(\DD_i)$.
We keep the purged prism $\tau$ in $\DD_i$; 
whenever a query accesses $\tau$ (when $\tau$ is a prism of the 
lowest-indexed cutting of some substructure), it realizes that 
$\CL(\tau)$ is purged and simply skips it.
A plane $h$ may be reinserted many times
due to the purging of a conflict list, but only once
for each substructure $\DD_i$ in which it is active (prior to the 
operation).  Also, the planes in $C(\DD_i)\setminus S(\DD_i)$,
and the planes marked as deleted will never be reinserted when a
conflict list is purged.

As mentioned, queries and insertions are handled as in 
Section~\ref{sec:insert}.
For queries, while processing a substructure $\DD_i$ and searching 
in the conflict list of some prism $\tau$ of the lowest-indexed 
cutting $\Lambda_0$ of $\DD_i$, we only consider planes in 
$\CL(\tau) \cap A(\DD_i)$, and report the lowest among them over the
query point. As we will show later, in Lemma~\ref{inv},
this suffices to retrieve the correct overall lowest plane.
That is, the plane that is lowest over $q$ among the reported planes
is the overall lowest plane over $q$.

To insert a plane $h$, we take, as in Section~\ref{sec:insert},
the largest contiguous prefix $\I'$ of occupied bins in $\I$,
of some length $j$, discard the existing structures in $\I'$, 
set\footnote{Note the difference between this step and the insertion 
in the insert-only case, discussed in Section~\ref{sec:insert}, where
$\B_j$ includes all the planes in the sets $S(\DD_i)$. In contrast, 
here only the active planes are considered.}
$\B_j := \left(\bigcup_{i=0}^{j-1} A(\DD_{i}) \right)\cup\{h\}$,
construct a new static structure for $\B_j$, and spread its 
substructures within some bins of $\I'$, according to the rules of 
Section~\ref{sec:insert}.  A plane $h\in \B_j$ is active after the
insertion (only) at the bin where it is stored.

Since we perform the reconstruction of the structure only on the 
active planes in the various destroyed substructures, the planes 
marked as deleted really disappear at this step, but only from
the structures stored at the bins of $\I'$; such a plane 
might still show up (marked as deleted) in substructures $\DD_\ell$ 
with larger bin indices $\ell$, which have not been touched by this 
instance of the insertion procedure.

It is easy to prove, by induction on the number of operations,
that the following invariants are maintained:
\begin{itemize}
\item[(D1)]\label{inv:inv1delete}
For each $i$, we have $2^{i-1} < |S(\DD_{i})| \leq 2^{i}$.
\item[(D2)]\label{inv:inv2delete}
The sets $A(\DD_{i})\subseteq S(\DD_{i})$ are pairwise disjoint, and 
their union is $H$.
\end{itemize}
Indeed, Invariant (D1) is the same as Invariant (I1), and its 
maintenance is argued as in Section~\ref{sec:insert}. Invariant (D2) 
follows because, by induction, a plane $h$ is active, before the 
current operation, at exactly one bin. If we delete $h$ then
the invariant continues to hold, as $h$ no longer belongs to $H$. 
If we purge $h$ from a conflict list in some substructure $\DD_i$, 
it is no longer 
active in $\DD_i$, but its reinsertion makes it active again, 
at the unique bin where it is stored. The same reasoning applies 
to all planes that were active at the bins that were
destroyed by the reinsertion, and a similar reasoning holds when we 
insert (rather than reinsert) a plane.

Note though that the lower bound (\ref{eq:eq1a}) of 
Section~\ref{sec:insert} does not have to hold now, 
as the number of active planes may be much smaller that the 
number of stored planes. The upper bound (\ref{eq:eq1b}) 
remains valid, and so does
Lemma~\ref{lem:lenI}.
The correctness of the data structure is a consequence of the 
following lemma.

\begin{lemma}\label{inv}
Let $q \in \reals^2$, and let $h \in H$ be a (non-deleted) plane on 
the lower envelope of $H$ over $q$ (for most queries, $h$ is unique).
Let $\DD_i$ be the unique substructure
for which $h \in A(\DD_i)$. Then $h$ belongs to the conflict list 
$\CL(\tau)$ of the prism $\tau$ of the lowest-indexed cutting 
$\Lambda_0$ of $\DD_i$, whose
$xy$-projection contains $q$, and $\tau$ has not been marked as purged.
\end{lemma}
\begin{proof}
The second part of the claim is an immediate consequence of the
first part, since $h$ is active in $\DD_i$.

Assume to the contrary that $h \not\in \CL(\tau)$.
Let $q^+$ be the point on $h$ over $q$.
By assumption, $q^+$ lies on the lower envelope of $H$, and
since $h \not\in \CL(\tau)$, but $h \in A(\DD_i) \subseteq S(\DD_i)$,
we have that $h$ does not cross $\tau$, so the point $q^+$ lies
above $\tau$. Let $t$ be the largest
index for which $q^+$ lies above the top terrain $\oLambda_t$
of the cutting $\Lambda_t$ of $\DD_i$; by what we have just argued,
such a $t$ exists, and it is possible that $t=m$ (the total 
number of real terrains in $\DD_i$), but $t\ne m+1$. Let $\tau'$ be the
prism of $\Lambda_t$ for which $q^+$ lies above the ceiling $\otau'$,
or, equivalently, $q$ lies in the $xy$-projection of 
$\tau'$.\footnote{We already noted, in Section~\ref{sec:static}, 
that the terrains $\oLambda_k$ are not necessarily monotone 
increasing in $z$. Nevertheless, the definition of $t$ means that 
$q^+$ lies below all the terrains $\oLambda_{t'}$, for $t' > t$.}
Since $\Lambda_t$ is a shallow cutting of $L_{\le k_t}(H_t)$, for
a suitable set of planes $H_t$, we have that $\oLambda_t$ lies fully
above $L_{k_t}(H_t)$. Thus, at least $k_t$ planes of $H_t$ pass
below $q^+$, and we denote the set of these planes as $\CL(q^+)$. 
Since $q^+$ now lies on the lower envelope of $H$, all the (at least 
$k_t$) planes of $\CL(q^+)$ must have been (marked as) deleted from 
$H$.

Note that it is possible that $\CL(\tau')$
has already been purged by the lookahead deletion mechanism.
Note also that we do not necessarily have 
$\CL(q^+) \subseteq \CL(\tau')$,
as some planes from $H_t$ may have appeared in too many conflict
lists after the construction of $\Lambda_t$ and may have been removed
by the pruning mechanism from the conflict list of
$\Lambda_t$.

Consider now the prism $\tau''$ of $\Lambda_{t+1}$ that contains 
$q^+$ (with $\tau'' = \reals^3$, if $t=m$). Since $\CL(q^+) 
\subseteq H_t$, none of its planes were pruned earlier, before 
constructing $\Lambda_{t}$.  Consequently, 
$\CL(q^+) \subseteq \CL(\tau'')$ and, by definition 
of $q^+$, $h\in \CL(\tau'')$.\footnote{%
  The planes in $\CL(q^+)$ certainly cross $\tau''$, since they 
  intersect the vertical downward ray emanating from $q^+$.
  Note also that some planes that pass below $q^+$ may have been
  pruned earlier, and thus do not belong to $\CL(q^+)$.}
In the extreme case $t=m$ we take $\tau''$, as just mentioned,
to be the entire 3-space,
and then $H_t\subseteq \CL(\tau'')$. 
If $t < m$, we have $|\CL(\tau'')| \leq \alpha k_{t+1} = 
2\alpha k_t$, and if $t = m$, we have
$|\CL(\tau'')| = n \leq 2\alpha k_m$.
Hence, by the time $q^+$ has reached the lower envelope of $H$, at 
least
\[
|\CL(q^+)|\ge k_t \ge \frac{|\CL(\tau'')|}{2\alpha}
\]
planes of $\CL(\tau'')$ have been marked as deleted. 
Thus, the lookahead-deletion mechanism should have
purged $\CL(\tau'')$, which contains $h$,
but $h$ is still in $A(\DD_i)$, a contradiction that
establishes the claim.
\end{proof}

\noindent
The following lemma analyzes the performance of the data structure.

\begin{lemma}
The amortized deterministic cost of an insertion is $O(\log^3 n)$,
the amortized deterministic cost of a deletion is $O(\log^5 n)$,
and the worst-case deterministic cost of a query is $O(\log^2 n)$.
\end{lemma}
\begin{proof} 
The bound on the query time follows as in 
Lemma~\ref{lem:insert-only}.

\paragraph{Insertions.}
Consider an insertion of a plane $h$. As in Section~\ref{sec:insert},
each plane that is 
in $H$ (i.e., has not been marked 
as deleted) holds $b(w-i)$ \emph{credits}, each worth $a\log^2 N$ 
\emph{units}, where $i$ is the current unique location (bin index) 
of the structure $\DD_i$ for which $h\in A(\DD_i)$, and 
$w:= \log N + 2$ bounds, by Lemma~\ref{lem:lenI}, the maximum 
length of $\I$. Here $a$ is as defined in Section~\ref{sec:insert}, 
and $b$ is another constant parameter that will be chosen later.

We modify the analysis in Lemma~\ref{lem:insert-only} as follows. 
Let $j$ be, as before, the index of the first empty bin
just before the insertion, and let
$\B_j := \left(\bigcup_{i=0}^{j-1} A(\DD_{i}) \right)\cup\{h\}$.
Define $s = |\B_j|$ and 
$t = \sum_{i=0}^{j-1} |S(\DD_{i})| + 1$.
As in (\ref{eq:eq1a},\ref{eq:eq1b}), using Invariant (D1) 
(which is identical to Invariant (I1)), we get that $2^{j-1} < t 
\le 2^j$.  On the other hand, as already remarked, it is possible 
that $s \ll t$, which may cause us to place the newly constructed 
substructures $\D^{(u)}$ in bins of rather small indices. The active
planes in these bins will then require a larger number of credits,
which the scheme in Section~\ref{sec:insert} cannot provide.

To cover the cost of an insertion in such a case, we observe that
$s \ll t$ means that most elements 
in the structures that are destroyed by the insertion are not active
in their respective structures.
That is, they either are not stored in their substructure (call it 
$\DD_i$), are (marked as) deleted, or were contained in a conflict
list of $\DD_i$ that has been purged.
To exploit this observation, we proceed as follows. 
Consider some substructure $\DD_i$, and let $\tau$ be
a prism in $\DD_i$ (at any level of the hierarchy). If the conflict 
list of $\tau$ has not been purged, we denote by $D(\tau)$ the 
number of planes in $\CL(\tau)$ that have been (marked
as) deleted.\footnote{Again, we also count planes in 
$C(\DD_i) \setminus S(\DD_i)$.}
We define the potential of $\DD_i$ to be 
\begin{equation} \label{pot-bi}
\Psi(\DD_i) =  \left(b' \sum_{\tau \in \Pi_{\neg p}^{(i)}} D(\tau) + 
b''  \sum_{\tau \in \Pi_p^{(i)}} |\CL(\tau)|\right) \log N 
\end{equation}
credits, where $b'$ and $b''$, with $b' > b''$ 
are suitable multiples of $b$, to be set later,
$\Pi_{\neg p}^{(i)}$ denotes the set of all non-purged prisms in $\DD_i$ and
$\Pi_p^{(i)}$ is the set of all purged prisms in $\DD_i$.
We emphasize once again that $|\CL(\tau)|$ denotes the original size
of the conflict list, as it was constructed.
The purpose of this potential is to provide the credit for the
increase in potential (when the number of active 
planes in the destroyed substructures is small) and to
pay for the reinsertions (when a conflict list is purged). 
The overall potential $\Psi^*$ of the structure, measured in
credits (rather than in units of credit), is defined as
\[
\Psi^* = \Psi + \sum_{i\ge 0} \Psi(\DD_i) ,
\]
where $\Psi$ is the overall number of credits held by the planes in 
$H$, as explained above.
As mentioned, every plane counted by $t$ but not by $s$ has either
been marked as deleted or was contained in a conflict list that has 
been purged. Thus
\[
t - s \leq  
\sum_{i = 0}^{j - 1} \left(\sum_{\tau \in \Pi_{\neg p}^{(i)}} D(\tau) + 
\sum_{\tau \in \Pi_p^{(i)}} |\CL(\tau)|\right) .
\]
The following two cases can arise:

\paragraph{Case (a):} $s < (1-\zeta)t$.
In this case $t-s > \frac{\zeta}{1-\zeta} s$; that is,
\[
\sum_{i = 0}^{j - 1} \left(\sum_{\tau \in \Pi_{\neg p}^{(i)}} D(\tau) + 
\sum_{\tau \in \Pi_p^{(i)}} |\CL(\tau)|\right)
 > \frac{s}{31} ,
\]
or, with a suitable choice of $b''$, and recalling that
$b' > b''$,
\[
\sum_{i=0}^{j-1} \Psi(\DD_i) > \frac{b''}{31} s\log N.
\]
Since all these substructures are now destroyed,
this potential will no longer appear in the full potential $\Psi^*$, 
and we can safely use it
to cover the cost of the insertion.
As described in Section~\ref{sec:insert}, 
for a sufficiently large constant $b''$, 
this will enable us to place the $s$ planes of $\B_j$ at the bins 
they are to be stored in, no matter where, with the correct amount 
of credit assigned to each of them (and with change to spare). 

\paragraph{Case (b):} $s \ge (1-\zeta)t$.
Then 
\[
|S(\D^{(1)})| \ge (1-\zeta)s \ge (1-\zeta)^2 t \ge 
(1-\zeta)^2 2^{j-1} > 2^{j-2} ,
\]
(which holds for $\zeta = 1/32$), implying that $\D^{(1)}$ is stored
at bin $j$ or $j-1$. (Note that, right after the reconstruction 
caused by an insertion, 
$A(\D^{(j)}) = S(\D^{(j)})$ for each of the newly constructed
substructures $\D^{(j)}$.) Furthermore, the lower bound in Invariants
(D1) and our assumption that $s \ge (1-\zeta)t$ show that at least 
$\sum_{i=1}^{j-2}2^{i-1} + 2-\zeta t = 2^{j-2}-\zeta t
\geq t/4-\zeta t$ planes in $\B_j$ were stored at bins 
$i=0,1,\dots, j-2$ prior to the insertion.  By 
Lemma~\ref{lem:discarded}, the reconstruction following the 
insertion passes at most $\zeta s \le \zeta t$ of them to bins of 
indices $0,1,\dots,j-2$, so at least 
$t/4-2\zeta t$ of these planes end up at $\DD_j$ or $\DD_{j-1}$. 
These planes release at least $bt(1/4 -2\zeta)$ credits.

The other planes, that are passed to lower-indexed bins, may 
require additional credits, but the total number of these credits
is bounded as in the proof of Lemma~\ref{lem:insert-only}.
It follows that for our choice of $\zeta = 1/32$ and for a 
somewhat larger choice of $b$ (than the one in 
Section~\ref{sec:insert}), the released credit suffices to pay 
for the real insertion cost.
In all the newly constructed substructures $\DD_i$, we have 
$S(\DD_i) = A(\DD_i)$, and no conflict list has been purged,
so, by definition, $\Psi(\DD_i) = 0$; in other words, no credits 
have to be reallocated for these potentials.

In conclusion, in either of the two cases, the actual cost of the 
insertion, plus the difference in $\Psi^*$, is upper bounded by the
amortized cost, which, as in Section~\ref{sec:insert}, is $bw$
credits, or $O(\log^3n)$ units of credit.

\paragraph{Deletions.}
Finally, we analyze the amortized deletion cost. When we
delete a plane $h$ from $H$, we give it $b'_0\log^3N$ credits,
where $b'_0$ is some sufficiently large multiple of $b'$,
that is, $\Theta(\log^5N)$ units of credit. 
To each conflict list $\CL(\tau)$ 
with $h \in \CL(\tau)$, that has not 
been purged yet, we allocate, from the credit given to $h$,
$b' \log N$ credits to account for the
increase in potential due to the deletion of $h$
(this increase is reflected in (\ref{pot-bi})).
There are at most $c\log N$ such lists
in each of the $O(\log N)$ substructures $\DD_i$, so
there are enough credits to account for the increase
in potential. 

The marking of $h$ as deleted may lead, via the lookahead 
deletion mechanism, to
the purging of several conflict lists containing $h$, 
and to the reinsertion of the active planes in these conflict lists.
Let $\CL(\tau)$ be a conflict list in some substructure $\DD_i$ 
that is purged when $h$ is deleted.
At this point there are at least
$\frac{1}{2\alpha} |\CL(\tau)|$ planes 
in $\CL(\tau)$ that have been marked as deleted, and 
the status of $\CL(\tau)$ switches
from non-purged to purged.
Thus, this switch releases (again, recall (\ref{pot-bi}))
\[
\left(b' D(\tau)
- b'' |\CL(\tau)|\right) \log N
\geq
\left(\frac{b'}{2\alpha} 
- b''\right) |\CL(\tau)|\log N
\]
credits. 
The reinsertion itself proceeds exactly as in the case of insertion.
With a suitable choice of $b'$ and $b''$, say $b'\ge 4\alpha b''$
and $b''$ sufficiently large,
the $\Theta(\log N)$ credits needed to support each of the at most
$|\CL(\tau)|$ reinsertions are thus available, 
including the $\Theta(\log N)$ credits that each reinserted 
plane has to bring along. 
This implies, as in Section~\ref{sec:insert}, that the amortized cost
of a deletion is indeed $O(\log^3N)$ credits, or $O(\log^5N)$ units.

Each global rebuilding occurs after $\Theta(N)$ updates 
(insertions and deletions),
which have occurred since the last global rebuilding. In the
rebuilding, all lingering (marked as) deleted planes are fully
removed from the structure. All other planes abandon 
their present status, and we simply build a static structure from 
the current planes, storing its substructures at a suitable
sequence of bins.  The cost of the rebuilding is $O(N\log^2N)$,
so the total cost of all rebuildings is $O(n\log^2n)$, where $n$ is
the total number of updates, plus the size of the initial set of 
planes (if nonempty). This is well subsumed by the overall 
amortized cost of the insertions and deletions.
\end{proof}

\paragraph{Storage.}
So far, the structure requires $O(n\log n)$ cells of storage:
$\I$ has $O(\log n)$ substructures, where the substructure $\DD_i$
at index $i$ is a hierarchy of cuttings, each approximating some
level in a geometric sequence of levels (of suitable subsets of $H$).
Put $n_i:= |C(\DD_i)|\le 2^i$. The number of prisms in the cutting 
for level $k$ is $O(n_i/k)$, and
the size of each conflict list is $O(k)$, so
the total storage for each level of $\DD_i$ is $O(n_i)$, for a total
storage of $O(n_i\log n_i)$. Summing over $i$, the total storage
is $O(n\log n)$. A similar analysis shows that the total storage,
excluding the conflict lists, is $O(n)$.

Using an idea that is credited to Afshani by Chan~\cite{Cha10}, we 
can improve the storage
to linear, if for each conflict list we store
only its initial size and the number of planes
that were deleted in it
(except for the lowest-indexed cuttings, where we keep the
conflict lists explicitly, but the size of any
such lowest-indexed list is only $O(1)$). Then, the storage for
$\DD_j$ is $O(n_j)$, making the overall storage $O(n)$.
To make this work, we need additional mechanisms to compensate
for the missing conflict lists.
Specifically, when we delete a plane $h$, we need to find the
conflict lists that contain $h$, and increment the deletion counter
of each corresponding prism. Naively, within a substructure $\DD_i$,
the plane $h$ has to find all the vertices of all the cuttings
that lie above it. Each such vertex is a vertex of some prism(s),
and $h$ belongs to the conflict list of each such prism.
However, $h$ might lie below a vertex $v$ and not belong to the
conflict lists of the incident prisms, because $h$ has been pruned
away while processing the current or a higher-indexed cutting.

Thus, we augment $\DD_j$ with a separate halfspace range reporting
data structure for the set of vertices of each of its cuttings. We
use the recent algorithm of Afshani and Chan~\cite{AC09} (which can
be made deterministic by using the shallow cutting construction 
of~\cite{CT15}), which preprocesses a set $V$ of points in 
$\reals^3$, in $O(|V|\log|V|)$ time, into a data structure of linear
size, so that the set of those points of $V$ that lie above a
querying plane $h$ can be reported in $O(\log |V| + t)$ time, with
$t$ being the output size.  The cost of augmenting $\DD_i$ with these
reporting structures is subsumed by the cost of building $\DD_i$
itself. Now, when deleting a plane $h$, we access each substructure
$\DD_i$ of $\I$ with $h \in C(\DD_i)$. For this, each plane $h$
stores pointers to all these structures. Since the overall size of
the sets $\C(\DD_i)$ is $O(n)$, the overall number of such pointers
is linear. For each substructure $\DD_i$ with $h \in C(\DD_i)$, we
find the prisms that contain $h$ in their conflict lists. To ensures
correctness of this step, $h$ also stores a second pointer, for each
$\DD_i$ containing it, to the level at which it was pruned; if $h$
was not pruned, we store a null pointer. Now $h$ accesses the
halfspace range reporting structures of all the levels higher that
the level at which $h$ was pruned, and retrieves from each of these
structures the prisms that contain it in their conflict lists. For
each such prism $\tau$, we increment its deletion counter by $1$. If
the counter becomes too large relative to the initial size,
as explained above, we purge the entire conflict list, and reinsert
its surviving active members into the structure (of course, this
step also requires a data structure to be performed efficiently, see
below).  The total cost of these steps, excluding the one that purges
conflict
lists that have become too small, is $O(\log^3n + t)$, where
$t$ is the overall number of prisms that store $h$ in their conflict
lists. The term $O(\log^3n)$ arises since we access up to $O(\log n)$
substructures $\DD_i$, access up to $O(\log n)$ halfspace range
reporting structures at each of them, and pay an
overhead of $O(\log n)$ for querying in each of them. Since, by
construction, $t=O(\log^2n)$, this modification,  so far,
adds $O(\log^3n)$ to the total of cost of a deletion.

As noted, we also require a mechanism to compute the active
members of the conflict lists that are purged in a substructure
$\DD_i$. 
To do so, we preprocess the planes of $S(\DD_i)$ into a (dual
version of a) halfspace reporting data structure that we keep with
$\DD_i$.  We query this structure with each of the at most four 
vertices of $\tau$, to obtain, in an output-sensitive manner, all the
planes of $S(\DD_i)$ that cross $\tau$. This structure takes space
linear in $|S(\DD_i)|$, $O(|S(\DD_i)| \log |S(\DD_i)|)$ time to
build, and can answer a query in $O(\log |S(\DD_i)| + t)$ time, where
$t$ is the output size. The cost of
answering such a query is subsumed by
the cost of reinserting the planes, and the cost of
constructing this reporting structure is subsumed by
the cost of constructing $\DD_i$.
We thus obtain the following main summary result of this section.
\begin{theorem}\label{main}
The lower envelope of a set of $n$ non-vertical planes in
three dimensions can be maintained dynamically, so as to support
insertions, deletions, and queries, so that each insertion takes
$O(\log^3n)$ amortized deterministic time, each deletion takes
$O(\log^5n)$ amortized deterministic time, and each query takes
$O(\log^2n)$ worst-case deterministic time, where $n$ is the
size of the set of planes at the time the operation is performed.
The data structure requires $O(n)$ storage.
\end{theorem}

\section{Dynamic Lower Envelopes for Surfaces}\label{sec:connect}

We finally show how to extend the data structure
from Section~\ref{sec:chan} for general surfaces. As
mentioned in the introduction, the key observation
is that Chan's technique (also with our improvement)
is ``purely combinatorial'': once we have,
as a black box, a procedure for efficiently constructing
vertical shallow cuttings, accompanied with efficient
procedures for the various geometric primitives that
are used by the algorithm (which are provided in our
algebraic model of computation---see Section~\ref{sec:prelim}
for details), the rest of the algorithm simply
organizes and manipulates the given
surfaces into standard data structures. Indeed, the
whole geometry needed for the deletion lookahead
mechanism is encapsulated in the proof of Lemma~\ref{inv},
which relies only on the properties of conflict lists in a
vertical shallow cutting, which hold for general well-behaved
surfaces too. We first show how to find a \emph{vertical} shallow
cutting \emph{with} conflict lists, as needed for Chan's
technique.

\begin{theorem}\label{thm:vert_shallow_cutting_cl}
Let $F$  be a set of $n$ continuous totally defined algebraic
functions of constant description complexity, so that the
complexity of the lower envelope of any $m$ functions in $F$ is
$O(m)$, and let $k \in \{1, \dots, n\}$.
Then there exists a vertical shallow cutting $\Lambda_k$ for 
$L_{\leq k}(F)$ with the following properties (where $s$ is
the vertical visibility parameter, introduced above, for $F$):
\begin{enumerate}
  \item The number of prisms in $\Lambda_k$ is $O((n/k)\log^2 n)$.
  \item Each prism $\tau$ in $\Lambda_k$ intersects at least $k$ and
  at most $2k$ functions in $F$, and the ceiling 
  of $\tau$ lies above $L_k(F)$.
  \item We can find $\Lambda_k$ and the conflict lists
    for its prisms in expected time
    $O(n \log^3 n \lambda_s(\log n))$,
    using expected space
    $O(n \log n \lambda_s(\log n))$, where $s$ is the parameter
    introduced in Section~\ref{sec:ric}.
\end{enumerate}
\end{theorem}
\begin{proof}
We combine the techniques from 
Sections~\ref{sec:levelapproximation},~\ref{sec:cutting} 
and~\ref{sec:ric}.
First, set $\lambda = 4c \log n$, for a suitable
constant $c$ as in Section~\ref{sec:levelapproximation}.
Pick $t$ randomly in 
$\big[\frac{7\lambda}{6}, \frac{5\lambda}{4}\big]$, and
let $S_k$ be a random subset of $F$ of size $r_k = 4c(n/k)\log n$.
If $r_k > n$, we set $r_k = n$ and we pick $t$ randomly in 
$[k,3k/2]$. Denote by $\oT_k$ the $t$-level in $\A(S_k)$.
By Lemma~\ref{lem:prop2} and as argued at the end of 
Section~\ref{sec:levelapproximation}, the expected complexity of 
$\oT_k$ is $O((n/k)\log^2 n)$.

We compute $\oT_k$ as follows: we perform the algorithm
from Section~\ref{sec:ric} on $F$ for the chosen level $t$, and
we stop the randomized incremental construction after
$r_k$ steps. The set of functions inserted during these
steps constitute the random sample $S_k \subseteq F$.
By Theorem~\ref{thm:ric}, this step takes expected
time $O(nt \lambda_s(t) \log (n/t)\log n) =
O(n \log^3 n \lambda_s(\log n))$, and expected space
$O(n t \lambda_s(t)) = O(n \log n \lambda_s(\log n))$,
where $s$ is as above.
As a result, we get the vertical decomposition $\VD_{\leq t}(S_k)$ of
$L_{\leq t}(S_k)$ together with the conflict lists
(with respect to $F$) of the prisms in $\VD_{\leq t}(S_k)$.
From this, we can extract $\oT_k$ by gluing together the
ceilings of all prisms that are met by $L_t(S_k)$.
By using the pointers that connect between adjacent prisms in 
$\VD_{\leq t}(S_k)$, and the pointers that connect the prisms 
with their vertices in $\VD_{\leq t}(S_k)$, we can do this in 
$O(|\oT_k|)$ steps. If the complexity of $\oT_k$ exceeds its
expectation by more than some preset threshold constant factor, 
we repeat the whole process with a new random level $t$.
By Markov's inequality, this happens
a constant number of times in expectation.

Next, we compute for each function $f \in F \setminus S_k$ the 
intersection between $f$ and $\oT_k$. For this, we inspect each
prism $\tau \in \VD_{\leq t}(S_k)$ that has $f$ in its conflict 
list and is incident to $\oT_k$, compute the intersection
between $f$ and the boundary of $\tau$, and keep the part of
this intersection that appears on $\oT_k$. Finally, we glue
together the resulting partial curves in order to obtain
$f \cap \oT_k$ (this intersection curve does not need to be
connected and can be fairly complex). The total time for this step
is proportional to the total size of the conflict lists of
$\VD_{\leq t}(S_k)$ times a logarithmic factor for the 
gluing operation. This is $O(n \log^2 n \lambda_s(\log n))$
in expectation, by Theorem~\ref{thm:ric}.

Finally, we construct the vertical decomposition $\oLambda_k$ of
$\oT_k$, in $O(|T_k|\log n) = O((n/k)\log^3 n)$ time, by
sweeping the $xy$-projection of $\oT_k$ with a $y$-vertical line.
By Lemma~\ref{lem:shallow-cutting}, the downward vertical extension
$\Lambda_k$ of $\oLambda_k$ is a shallow cutting
for the first $k$ levels of $\A(F)$, with high probability.
To find the conflict lists of the prisms of $\Lambda_k$, we build a 
planar point location structure for the $xy$-projection of 
$\oLambda_k$. Then, for each $f \in F  \setminus S_k$, we use the 
planar point location structure to locate the trapezoid of 
$\oLambda_k$ that contains an initial point on each connected 
component of $f\cap \oT_k$.  Then we use the planar point location 
structure again to trace each connected component through 
$\oLambda_k$, and pay $O(\log n)$ time for each trapezoid
that we cross (we use the planar point location structure to
cross through the $xz$-faces of the prisms, where there may
be a lot of prisms on the other side). Then, starting from these 
trapezoids, we perform
another traversal of $\oLambda_k$ to find all the trapezoids of
$\oLambda_k$ that lie fully above $f$. For all these trapezoids,
of both kinds, $f$ is in the conflict list of the corresponding 
vertical prism, and all members of the conflict lists arise in
this manner. Hence, the overall time for this step is proportional 
to the total size of the conflict lists of $\Lambda_k$
times a logarithmic factor
for the point locations. Thus, perhaps
somewhat pessimistically, the total expected running time 
for this step is $O(k (n/k) \log^3 n) = O(n \log^3 n)$, by 
Lemma~\ref{lem:shallow-cutting}. The functions 
$f \in F \setminus S_k$ for which $f \cap \oT_k = \emptyset$ 
either lie completely above or completely below $\oT_k$. 
In the former case, such a function is irrelevant, and we simply
discard it. In the latter case, it appears in all conflict lists
of $\Lambda_k$. In the last step, we check whether all prisms 
actually intersect between $k$ and $2k$ functions from $F$. 
If this is not the case, we repeat the whole construction. 
By the discussion in Section~\ref{sec:levelapproximation} 
and Markov's inequality, the expected number of attempts is constant.

The total expected running time and storage is dominated by
the randomized incremental construction, and hence the theorem 
follows.
\end{proof}

\paragraph{Remark.} The variant 
of Theorem~\ref{thm:vert_shallow_cutting_cl} for general lower 
envelope complexity is as follows:

\begin{theorem}\label{thm:vert_shallow_cutting_cl_general}
Let $F$  be a set of $n$ continuous totally defined algebraic
functions of constant description complexity, so that the
complexity of the lower envelope of any $m$ functions in $F$ is
at most $\psi(m)$, where $\psi(m)/m$ increases monotonically. 
Furthermore, let $k \in \{1, \dots, n\}$.
Then there exists a vertical shallow cutting $\Lambda_k$ for 
$L_{\leq k}(F)$ with the following properties:
\begin{enumerate}
  \item The number of prisms in $\Lambda_k$ is
    $O(\psi(n/k)\log^2 n)$.
  \item Each prism $\tau$ in $\Lambda_k$ intersects at least $k$ and
  at most $2k$ functions in $F$, and the ceiling of $\tau$ lies 
  above $L_k(F)$.
  \item We can find $\Lambda_k$ and the conflict lists 
    of its prisms in 
    $O(\psi(n/\log n) \log^4 n \lambda_s(\log n))$
    expected time, using expected space
    $O(\psi(n/\log n) \log^2 n \lambda_s(\log n))$.
\end{enumerate}
\end{theorem}
\begin{proof}
The argument is the same, with slightly adjusted bounds.
The remark at the end of Section~\ref{sec:levelapproximation}
yields the bound on the size of $\Lambda_k$.
The bound on the running time follows by using
Theorem~\ref{thm:ric_general} and the fact that $\psi(m)/m$ is 
monotone increasing.
\end{proof}

Now we can combine Theorem~\ref{thm:vert_shallow_cutting_cl}
with the construction in Section~\ref{sec:chan} to obtain
the desired data structure.
However, we need to adjust the bounds in our analysis to account for 
the fact that the cuttings that we construct are of slightly 
sub-optimal size ($O((n/k) \log^2 n)$ instead of $O(n/k)$), and that 
we need more (expected) time to construct them
($O(n \log^3 n \lambda_s(\log n))$ instead of $O(n\log n)$).
Specifically, we need to apply the following adjustments:
Since now the total size of the conflict lists is $O(n \log^2 n)$, 
when constructing the static data structure 
(Section~\ref{sec:static}), we prune a function only when it 
appears in $c \log^3 n$ conflict lists.  This increases the overall 
size of the static structure to
$O(n \log^3 n)$, and the construction time becomes
$O(n\log^4 n \lambda_s(\log n))$ (in expectation).
The query time remains $O(\log^2 n)$, since we can perform point
location in general minimization diagrams in $O(\log n)$ time
per query. Concerning insertions, the increased construction time 
implies that in Lemma~\ref{lem:insert-only}, we need to allocate
$\Theta(\log^4 N \lambda_s(\log N))$ units for each credit.
Then, the remaining analysis in the proof of 
Lemma~\ref{lem:insert-only} continues to hold, and we have an 
amortized insertion cost of $O(\log^5 N \lambda_s(\log N))$. Finally,
we analyze the deletion cost: since now each deleted element can 
appear in $O(\log^4 n)$ conflict lists, we must equip it
with $\Theta(\log^5 N)$ credits to pay for the reinsertions.
With the adjusted number of units per credit, this means that
each deleted element needs
$\Theta(\log^{9} N \lambda_s(\log N))$ units to pay 
for the reinsertions. Since the storage for the dynamic structure 
is proportional to the storage for the static structure, we
need $O(n \log^3 n)$ space overall.

Our efforts so far can thus be immediately reaped into the following
main result.
\begin{theorem}\label{dyn-surf}
The lower envelope of a set of $n$ totally defined continuous
bivariate functions of constant description complexity in three
dimensions, so that the lower envelope of any subset of the
functions has linear complexity, can be maintained dynamically,
so as to support insertions, deletions, and queries, so that
each insertion takes $O(\log^5 n \lambda_s(\log n))$
amortized expected time, each deletion takes
$O(\log^{9} n \lambda_s(\log n))$ amortized expected time, and
each query takes $O(\log^2 n)$ worst-case deterministic time, 
where $n$ is the number of functions currently in the data structure.
The data structure requires $O(n\log^{3} n)$ storage in expectation.
\end{theorem}

\paragraph{Remark.} With the obvious adjustment to the bounds,
we get the following theorem for general lower envelope complexity:
\begin{theorem}\label{thm:dyn-surf_general}
Let $F$ be a finite set of totally defined continuous bivariate 
functions of constant description complexity in three dimensions, 
so that the complexity of the lower envelope of any $m$ functions 
of $F$ is at most $\psi(m)$, where $m \mapsto \psi(m)/m$ is 
monotonically increasing. Then the lower envelope of $F$ can 
be maintained dynamically, so as to support insertions, deletions, 
and queries, so that each insertion takes
$O(\frac{\psi(n/\log n)}{n}\log^6 n \lambda_s(\log n))$ 
amortized expected time, each deletion takes
$O(\frac{\psi(n/\log n)}{n}\log^{10} n \lambda_s(\log n))$ 
amortized expected time, and each query takes
$O(\log^2 n)$ worst-case deterministic time, where $n$ is the
number of functions currently in the data structure.
The data structure requires $O(\psi(n)\log^{3} n)$ storage in 
expectation.
\end{theorem}

\section{Applications} \label{sec:application}

Let $S \subset \reals^2$ be a finite set of
pairwise disjoint \emph{sites}, each a simply-shaped convex planar 
region, e.g., points, line segments, disks, etc.
The problem of finding, for a point $q \in \reals^2$,
its nearest neighbor in $S$ under any
norm or convex distance function $\delta$~\cite{ChewDr85}
translates to ray shooting in the lower envelope of the 
set of functions $F = \{f_s(x) = \delta(x,s) \mid s \in S\}$.
Thus, if the lower envelope of $F$ has linear complexity,
Theorem~\ref{dyn-surf} yields a dynamic nearest neighbor data
structure for $S$. We note that the minimization diagram of the
lower envelope of $F$ is the Voronoi diagram of $S$ under
$\delta$~\cite{vor-book,EdSe}.

Dynamic nearest neighbor search has several applications that
we are going to mention, but first we introduce two classes of 
distance functions that are of particular interest.
\begin{itemize}
  \item \textbf{The $L_p$-metrics}:
    Let $p \in [1, \infty]$. We define, for
    $(x_1, y_1), (x_2, y_2) \in \reals^2$, the $L_p$ metric 
    \[
\delta_p((x_1, y_1), (x_2, y_2))=
\begin{cases}
      (|x_1 - x_2|^p + |y_1 - y_2|^p)^{1/p} & 
      \text{for $p < \infty$} \\
      \max \left\{ |x_1 - x_2|,\; |y_1 - y_2| \right\} & 
      \text{for $p = \infty$}.
\end{cases}
\]
    It is well known that $\delta_p$ is a metric, for any such $p$,
    and thus it induces lower envelopes of linear
    complexity, for any set of sites as above~\cite{Lee80}.
    The precise value for the 
    parameter $s$ defined in Section~\ref{sec:ric} depends on the 
    choice of $p$. To ensure a reasonable bound on $s$, we assume
    that $p$ is an integer (or $p=\infty$). For a finite integer
    value of $p$, we have $s=O(p^2)$, and for $p = \infty$,
    we have $s = 4$ (as two $L_\infty$-bisectors can intersect
    at most twice).
  \item \textbf{Additively weighted Euclidean metric}:
    Let $S \subset \reals^2$ be a set of point sites, and suppose
    that each $s \in S$ has an associated weight $w_s \in \reals$.
    We define a distance function
    $\delta: \reals^2 \times S \rightarrow \reals$ by
    $\delta(p, s) = w_s + |p s|$, where
    $| \cdot |$ denotes the Euclidean distance.
    This distance function also induces lower envelopes of
    linear complexity, i.e., the additively weighted
    Voronoi diagram of point sites has linear 
    complexity~\cite{vor-book}.  The bisectors for the 
    additively weighted Voronoi diagram are hyperbolic arcs, so
    each pair of bisectors intersects at most $4$ times.
    Thus, in this case, we have $s = 6$, for the 
    parameter $s$ defined in Section~\ref{sec:ric}.
\end{itemize}

\subsection{Direct applications of dynamic nearest neighbor search}

Now we can improve several previous results
by plugging our new bounds into known methods.

\paragraph{Dynamic bichromatic closest pair.}
Let $\delta: \reals^2 \times \reals^2 \rightarrow \reals$ be
a planar distance function,\footnote{For the definition 
of the bichromatic closest pair, $\delta$ can be an arbitrary 
function, but in applications, we usually assume that it is 
a metric.}
and let $R, B \subset \reals^2$
be two sets of point sites in the plane.
The \emph{bichromatic closest pair} of $R$ and $B$ with
respect to $\delta$ is a pair $(r, b) \in R \times B$
that minimizes $\delta(r, b)$.
We get the following improved version of Theorem~6.8 
in Agarwal \etal~\cite{AES}, which is obtained by combining
Eppstein's method~\cite{Eppstein95} with the  dynamic lower
envelope structure from Theorem~\ref{dyn-surf}.

\begin{theorem}\label{thm:dyn_bcp}
Let $R$ and $B$ be two sets of points in the plane, with
a total of at most $n$ points. We can store $R \cup B$ in a dynamic
data structure of size $O(n \log^3 n)$,
that maintains a closest pair in $R \times B$ under any $L_p$-metric
or any additively weighted Euclidean metric, in 
$O(\log^{10} n \lambda_s(\log n))$ amortized expected
time per insertion and $O(\log^{11}n \lambda_s(\log n))$ 
amortized expected time per deletion.
\end{theorem}

In fact, Chan~\cite{Chan19} recently showed how to adapt the data 
structure from Section~\ref{sec:chan} directly for the dynamic 
bichromatic closest pair problem, without incurring the 
polylogarithmic overhead that is inherent in Eppstein's 
method~\cite{Eppstein95}.  As in Section~\ref{sec:connect}, 
this improvement also carries over to the case of surfaces. 
For the interested reader, we explain how Chan's scheme
plays out in our presentation.

\begin{theorem}\label{thm:dyn_bcp_improved}
Let $R$ and $B$ be two sets of points in the plane, with
a total of at most $n$ points. We can store $R \cup B$ in a dynamic
data structure of size $O(n \log^3 n)$,
that maintains a closest pair in $R \times B$ under any $L_p$-metric,
with $p$ an integer or $p=\infty$,
or any additively weighted Euclidean metric, in 
$O(\log^{5} n \lambda_s(\log n))$ amortized expected
time per insertion and $O(\log^{9}n \lambda_s(\log n))$ 
amortized expected time per deletion, where $s$
is as defined in Section~\ref{sec:ric};
for additively weighted Euclidean distances we have $s=6$.
\end{theorem}

\begin{proof}
We maintain two copies of the dynamic structure from 
Sections~\ref{sec:chan} and~\ref{sec:connect}, one 
for the red points and one for the blue points. 
In addition,
we have a global min-heap $H$ that contains a set of some bichromatic 
pairs from the current set $R \times B$, where the key for a pair 
$(r, b)$ is the distance $\delta(r, b)$. 
More precisely, our data structure consists of $O(\log N)$ 
\emph{red bins} $\DD_0^R, \DD_1^R, \dots$ that store 
the (surfaces corresponding to) the red points, and of 
$O(\log N)$ \emph{blue bins} $\DD_0^B, \DD_1^B, \dots$ that store 
the (surfaces corresponding to) the blue points, satisfying 
the same invariants as in Section~\ref{sec:chan}. Here, 
as before, $N$ is an appropriate power of $2$ close to $n$. 
As we will see, it is 
possible to update $H$ such that it contains, at each point in time,
the current closest pair in $R \times B$, which we thus can retreive in
$O(1)$ time.

To insert a new red point 
$r$ into $R$,
we add $r$ into the red bins as 
in Section~\ref{sec:chan}. Then, we update the heap 
$H$ as follows: 
for each non-empty blue bin $\DD_i^B$,
we perform a point-location query to find the prism 
(or prisms) of the lowest-indexed cutting $\Lambda_0$ in $\DD_i^B$
whose $xy$-projection contains $r$. Next, 
for each active site $b$ in the conflict list of 
this prism (or prisms),  
we insert the pair $(r, b)$ into $H$, 
using $\delta(r, b)$ as the key.
An insertion into $B$ is symmetric.

To delete a red point $r$ from 
$R$, we remove from $H$ all pairs $(r, b')$ that involve $r$ 
(we can find them quickly by maintaining a list of back-pointers
with each element in $R \cup B$). After that, we perform the deletion
from the red bins as in Section~\ref{sec:chan}. If this causes a red
point $r'$ to be reinserted, we first remove all occurrences of 
$r'$ from $H$, as we did for $r$, and then we reinsert $r'$ into 
the structure as described in the preceding paragraph. 
A deletion from $B$ proceeds symmetrically.

These modifications do not affect the asymptotic amortized 
running times and the space requirement from 
Sections~\ref{sec:chan} and~\ref{sec:connect}. Indeed, 
the deletions from $H$ can be charged to the insertions 
into $H$. 
The cost for the insertions into $H$ (and the accompanying point 
locations in the $xy$-projections of the cuttings) 
can be covered with the units of credit that are spent for the
insertion or the reinsertion  (see Section~\ref{sec:chan}).
The space overhead for $H$ is $O(n \log N)$, since each insertion 
or reinsertion of a site $s$ can lead to only $O(\log N)$ new pairs 
in $H$, while removing all previous pairs involving $s$.

It remains to argue that the algorithm is correct. 
More precisely, we show that after each insertion and each
deletion, the current closest pair is 
in $H$ (and has the minimum key).\footnote{%
  We assume here that the closest pair is unique, but the
  argument easily carries over for the more general degenerate case.}

First, if an update does not change the current closest pair,
the invariant is clearly maintained: Indeed, an insertion does not
remove pairs from $H$. 
The first step of a 
deletion removes only pairs that involve the site that is to 
be deleted. If the deletion-lookahead mechanism causes  a member 
of the current closest pair to be reinserted, 
the closest pair will be added back to $H$ 
during this step, 
since by Lemma~\ref{inv}, it 
will be discovered by one of the queries into the  
lowest-indexed cuttings.

Second, if the current closest pair changes due to an insertion,
one of its two sites 
has just been added, and, as just discussed, Lemma~\ref{inv} 
ensures that the new closest pair 
is inserted into $H$.
Finally, suppose the current closest
pair changes due to a deletion. 
Let $(r^*, b^*)$ be the new closest
pair. Assume, without loss of generality, that $r^*$ was
inserted or last reinserted after $b^*$ was inserted or
last reinserted. This means that when $r^*$ was (re)inserted,
there was a (unique) blue bin $\DD^B_i$ that had the surface for
$b^*$ in $A(\DD^B_i)$, and 
$\DD^B_i$ was queried with $r^*$. If this query
found $(r^*, b^*)$, then this pair was inserted into $H$,
and the invariant follows.
Suppose this query did not find $(r^*, b^*)$. Then,
the surface for $b^*$ was not in the conflict list of the prism
in the lowest-indexed cutting $\Lambda_0$ of $\DD_i^B$ whose 
$xy$-projection contains $r^*$. However, by the time $(r^*, b^*)$
becomes the closest pair, Lemma~\ref{inv} ensures that the unique bin
$\DD_j^B$ where the surface for $b^*$ is active must have exactly
this property.  Thus, 
the bin $\DD_j^B$ in which $b^*$ is active must have changed since
the last (re)insertion of
$r^*$. However, the active bin of $b^*$ can only change due to
a reinsertion of $b^*$. 
This contradicts our assumption that $r^*$ was inserted
or last reinserted after $b^*$ was inserted or last
reinserted.
\end{proof}

\paragraph{Minimum Euclidean bichromatic matching.}
Let $R$ and $B$
be two sets of $n$ points in the plane (the \emph{red} and
the \emph{blue} points). A \emph{minimum Euclidean bichromatic
matching} $M$ of $R$ and $B$ is a set $M$ of $n$ line segments 
that go between $R$ and $B$ such that each point in $R \cup B$
is an endpoint of exactly one line segment in $M$ and such
that the total length of the segments in $M$ is minimum over
all such sets. 
Agarwal \etal~\cite[Theorem~7.1]{AES} show how to compute such a
minimum Euclidean bichromatic matching in total time 
$O(n^{2 + \eps})$, for any $\eps>0$, building on a trick by
Vaidya~\cite{Vaidya89}.  The essence of the algorithm lies in a
dynamic bichromatic closest pair data structure for a suitably
defined additively weighted Euclidean metric. The algorithm makes
$O(n^2)$ updates to this structure.  Thus, using
Theorem~\ref{thm:dyn_bcp_improved} (and the fact
that $s = 6$ for the case of the additively weighted
Euclidean metric), we get the following improvement:

\begin{theorem}
Let $R$ and $B$ be two sets of points in the plane, each with
$n$ points. We can find a minimum Euclidean bichromatic matching
for $R$ and $B$ in 
$O(n^2\log^{9}n \lambda_\swed(\log n))$ expected time.
\end{theorem}

\paragraph{Dynamic minimum spanning trees.}
Following Eppstein~\cite{Eppstein95}, 
Theorem~\ref{thm:dyn_bcp_improved}
immediately gives a data structure for the dynamic maintenance
of the edge set of a  minimum spanning tree for any $L_p$-metric,
under insertions and deletions of points. As described by
Agarwal \etal~\cite[Theorem~6.9]{AES}, it suffices to maintain
a suitable collection of bichromatic closest pair structures
such that each point appears in $O(\log^2 n)$ such structures, and
such that each insertion or deletion of a point 
requires $O(\log^2 n)$ updates of bichromatic closest pairs.
We thus get the following improved version of Theorem~6.9
in Agarwal \etal~\cite[Theorem~6.9]{AES}.

\begin{theorem}
Let $p$ be an integer. We can maintain a minimum spanning tree of a set of at most $n$ points in the
plane, under the $L_p$-metric, 
such that each insertion and  deletion 
takes $O(\log^{11} n \lambda_s(\log n))$
amortized expected time, using
$O(n \log^5 n)$ space, where $s=O(p^2)$.
\end{theorem}

\paragraph{Maintaining the intersection of unit balls in three
dimensions.}
Agarwal \etal~\cite{AES} show how to use dynamic lower envelopes to
maintain the intersection of unit balls in three dimensions, so that
certain queries on the union can be supported. Their algorithm uses
parametric search on the query algorithm in a black box fashion.
Thus, we obtain the following improvement over Theorem~8.1 
in Agarwal \etal~\cite{AES}. Since the $xy$-projection of the
intersection of the boundaries of two unit balls in $\reals^3$ 
is a quadratic curve in plane, and since two such
curves can intersect in at most four points,
we have $s = 6$ where $s$ is the parameter from Section~\ref{sec:ric}.

\begin{theorem}
The intersection $B^\cap$ of a set $B$ of at most $n$ unit balls 
in $\reals^3$ can be maintained dynamically by a data structure of
size $O(n \log^3 n)$, so that each insertion (resp., deletion)
takes $O(\log^5 n \lambda_\sballs(\log n))$ (resp.,
$O(\log^{9} n \lambda_\sballs(\log n))$) amortized expected time,
and the following queries can be answered:
(a) for any query point $p \in \reals^3$, we can determine,
in $O(\log^2 n)$ deterministic worst-case time,
whether $p \in B^\cap$, and
(b) after performing each update, we can determine, in
$O(\log^5 n)$ deterministic worst-case time, whether 
$B^\cap \neq \emptyset$.
\end{theorem}

\paragraph{Maintaining the smallest stabbing disk.}
Let $\mathcal{C}$ be a family of simply shaped compact
strictly-convex sets in the plane. We wish to dynamically maintain
a finite subset $C \subseteq \mathcal{C}$, under insertions and
deletions, such that, after each update, 
we have a smallest disk that intersects all the sets of $C$ (see
Agarwal \etal~\cite[Section~9]{AES} for precise definitions).
Our structure yields the following improved version of Theorem~9.3 
in~\cite{AES} (the precise value of $s$ depends on the choice of 
$\mathcal{C}$):

\begin{theorem}
A set $C$ of at most $n$ (possibly intersecting) simply shaped
compact convex sets in the plane can be stored in a data structure
of size $O(n \log^3 n)$, so that a smallest stabbing disk for $C$
can be computed in $O(\log^5 n)$ additional deterministic
worst-case time after each insertion
or deletion. An insertion takes $O(\log^{5} n \lambda_s(\log n))$
amortized expected time and a deletion takes
$O(\log^{9} n \lambda_s(\log n))$ 
amortized expected time.
\end{theorem}

\paragraph{Shortest path trees in unit disk graphs.}
Let $S \subset \reals^2$ be a set of $n$ point sites.
The \emph{unit disk graph} $\UD(S)$ of $S$ has vertex set $S$
and an edge between two distinct sites $s,t \in S$ if and only
if $|st| \leq 1$.
Cabello and Jej\^ci\^c~\cite{CabelloJejcic15} show how to compute
a shortest path tree in
$\UD(S)$ for any given root vertex $r\in S$, in time $O(n^{1+\eps})$,
for any $\eps > 0$,
using the bichromatic closest pair structure for the weighted 
Euclidean distance from
Agarwal \etal~\cite[Theorem~6.8]{AES}.
With our improved Theorem~\ref{thm:dyn_bcp_improved}, we get
the following result.
\begin{theorem}\label{thm:ud_sp}
Let $S \subset \reals^2$ be a set of $n$ sites. For any $r \in S$,
we can compute
a shortest path tree with root $r$ in $\UD(S)$ in expected time 
$O(n \log^{9}n \lambda_\swed(\log n))$.
\end{theorem}
We note that Theorem~\ref{thm:ud_sp} has very recently been 
improved by Wang and Xue~\cite{WangXu19}. They show how to 
compute a shortest path tree in a weighted unit disk graph  
with $n$ sites in $O(n \log^2 n)$ deterministic time.

\subsection{Dynamic disk graph connectivity}\label{sec:dyn_connect}

Next, we describe three further applications of our data
structure with improved bounds for problems on disk graphs:
Let $S \subset \reals^2$ be a finite set of
point sites, each with an assigned weight $w_s \geq 1$. 
Every $s \in S$ corresponds to a disk with center $s$
and radius $w_s$.  The \emph{disk graph} $D(S)$
is the intersection graph of these disks , i.e.,
$D(S)$ has vertex set $S$ and an edge connects two sites $s, t$ if
and only if $|st| \leq w_s + w_t$. 
In this section, we assume that all weights lie in the
interval $[1, \Psi]$, for some $\Psi \geq 1$, and we call
$\Psi$ the  \emph{radius ratio}.
First, we show how to dynamically maintain $D(S)$ under insertions 
and deletions of weighted vertices (i.e., disks), such that we can
answer \emph{reachability queries} efficiently: given
$s,t \in S$, is there a path in $D(S)$ from $s$ to $t$? 
The amortized expected update time is
$O(\Psi^2 \log^{8} n)$ for insertions and
$O(\Psi^2 \log^{12} n)$ for deletions, and the worst-case 
cost of a query is $O(\log n/ \log\log n)$.
Previous results have update time
$O(n^{20/21})$ and query time $O(n^{1/7})$ for general disk graphs,
and update time $O(\log^{10} n)$ and query time 
$O(\log n /\log\log n)$ for the unit disk case~\cite{ChanPaRo11}.

Our approach is as follows:
let $\G$ be a planar grid whose cells are pairwise openly disjoint
axis-aligned squares with diameter (i.e., diagonal) $1$.
For any grid cell $\sigma \in \G$, since $w_s \geq 1$ for every
$s\in S$, the sites of $\sigma \cap S$ induce a clique in $D(S)$.
The \emph{neighborhood} $\N(\sigma)$ of a cell $\sigma \in \G$ is the
$\big(\lceil 4\sqrt{2}\Psi\rceil + 1\big)\times
\big(\lceil 4\sqrt{2}\Psi\rceil + 1\big)$ 
block of cells in $\G$ with $\sigma$ at its center. 
We call two cells \emph{neighboring} if they are in each other's
neighborhood. By construction, the endpoints of any edge in $D(S)$
lie in neighboring cells: 
the side length of a cell is $\sqrt{2}/2$, and there 
can be at most $\lceil 2\Psi/(\sqrt{2}/2)\rceil = 
\lceil 2\sqrt{2}\Psi\rceil$ cells between two cells that 
contain adjacent sites. We define an abstract graph $G$ whose
vertices are the \emph{nonempty} cells $\sigma \in \G$, i.e.,
the cells with $\sigma \cap S \neq \emptyset$.  We pick the following
edges for $G$: consider any pair of neighboring grid cells 
$\sigma,\tau \in \G$. We have an edge between $\sigma$ and $\tau$ if
and only if there are two sites $s \in \sigma \cap S$ and
$t \in \tau \cap S$ with $|st| \leq w_s + w_t$.  By construction, and
since the sites inside each cell form a clique, the connectivity
between two sites $s, t$ in $D(S)$ is the same as for the
corresponding containing cells in $G$:

\begin{lemma}\label{lem:gridreachability}
Let $s,t \in S$ be two sites and let $\sigma$ and $\tau$ be
the cells of $\G$
containing $s$ and $t$, respectively. There is an $s$-$t$-path
in $D(S)$ if and only if there is a path between $\sigma$ and $\tau$
in $G$.
\end{lemma}

To maintain $G$, we use the following result by
by Holm, De Lichtenberg and Thorup, that supports dynamic
connectivity queries with respect to \emph{edge}
updates~\cite{HolmEtAl01}.
\begin{theorem}[Holm \etal, Theorem~3]
\label{thm:dynamicspanningtree}
Let $G$ be a graph with $n$ vertices. There exists a deterministic
data structure such that (i) we can insert or delete edges into/from
$G$ in amortized time $O(\log^2 n)$, and (ii) we can answer
reachability queries in worst-case time $O(\log n / \log \log n)$.
\end{theorem}

Even though Theorem~\ref{thm:dynamicspanningtree} assumes that
the number of vertices is fixed, we can use a standard
rebuilding method to maintain $G$ dynamically
within the same asymptotic amortized
time bounds, by creating a new data structure
whenever the number of nonempty grid cells changes by a factor of 
$2$.  When a site $s$ is inserted into or deleted
from $S$, at most $O(\Psi^2)$ edges in $G$ change, since
only the neighborhood of the cell of $s$ is affected.
Thus, once this set $E$ of changing edges is determined, we can
update $G$ in amortized time $O(\Psi^2 \log^2 n)$, by 
Theorem~\ref{thm:dynamicspanningtree}.
It remains to describe how to find $E$.
For this, we maintain a \emph{maximal bichromatic matching} (MBM)
between the sites in each pair of non-empty neighboring cells,
similar to Eppstein's method~\cite{Eppstein95}.
The definition is as follows:
Let $R \subseteq S$ and $B \subseteq S$ be two sets of sites.
An MBM $M$ between $R$ and $B$ is a maximal set of edges in 
$(R \times B) \cap D(S)$ that form a matching. Using 
Theorem~\ref{dyn-surf}, we can easily maintain MBMs.

\begin{lemma}\label{lem:mbm}
Let $R,B \subseteq S$ be two sets with a total of at most $n$ sites.
There exists a dynamic data structure that maintains a maximal
bichromatic matching of the disk graph $D(R \cup B)$, allowing the
insertion or deletion of sites in expected amortized time
$O(\log^{9} n \lambda_\swed(\log n))$ per update.
\end{lemma}

\begin{proof}
We have two dynamic lower envelope structures, one for $R$ and one
for $B$, as in Theorem~\ref{dyn-surf}, with the weighted distance
function $\delta(p, s) = |ps| - w_s$, that allow us to perform
nearest neighbor search with respect to $\delta$ (i.e., vertical ray
shooting at the corresponding lower envelope). We denote by $\NN_R$
the structure for $R$ and by $\NN_B$ the structure for $B$. We store
in $\NN_R$ the currently unmatched points in $R$, and in $\NN_B$ the 
currently unmatched points in $B$.  When inserting a site $r$ into 
$R$, we query $\NN_B$ with $r$ to get an ummatched point $b \in B$ 
that minimizes $|rb| - w_b$. If $|rb| \leq w_r + w_b$, we add the 
edge $rb$ to $M$, and we delete $b$ from $\NN_B$. Otherwise we 
insert $r$ into $\NN_R$.  By construction, if there is an edge 
between $r$ and an unmatched site in $B$, then there is also an edge 
between $r$ and $b$.  Hence, with a symmetric procedure for 
insertions into $B$, the insertion procedure maintains an MBM.
Now suppose we want to delete a site $r$ from $R$. If $r$
is unmatched, we simply delete $r$ from $\NN_R$. Otherwise, we remove
the edge $rb$ from $M$, and we reinsert $b$ as above, looking for
a new unmatched site in $R$ for $b$.
A symmetric procedure handles deletions from $B$. 
Since each insertion and deletion of a site requires $O(1)$ insert, 
delete and query operations in $\NN_R$ or $\NN_B$, the lemma follows.
\end{proof}

We create a data structure as in Lemma~\ref{lem:mbm}
for each pair of non-empty neighboring grid cells.
Whenever we insert or delete a
site $s$ in a grid cell $\sigma$, we update the MBMs for $\sigma$ and
all neighboring cells.
Observe that there is an edge between $\sigma$ and $\tau$
if and only if their MBM is not empty.
Thus, if $s$ is inserted, we add to $G$ an edge between any pair
$\sigma,\tau$ whose MBM changes from empty to non-empty.
If $s$ is deleted, we delete all edges between pairs of cells whose
MBM changes from non-empty to empty.
We thus obtain the following theorem:

\begin{theorem}
\label{thm:dynamicdiskgraph}
Let $\Psi \geq 1$. We can dynamically maintain the disk graph of a 
set $S$ of at most $n$ sites in the plane
with weights in $[1, \Psi]$ such that
(i) we can insert or delete sites in expected amortized time
$O(\Psi^2 \log^{9} n \lambda_\swed(\log n))$,
and (ii) we can determine for any pair of sites $s,t$ whether
they are connected by  a path
in $D(S)$, in determistic worst-case time $O(\log n / \log \log n)$.
\end{theorem}

As stated above, polylogarithmic bounds
were previously known only for the case of unit disk graphs.  More 
precisely, Chan, P{\v{a}}tra{\c{s}}cu, and Roditty mention that
one can derive from known results an update time of
$O(\log^{10} n)$~\cite{ChanPaRo11}. An
extension of our method leads to significantly
improved bounds for this case, too. 
Namely, for unit disks, we can obtain amortized expected update time
$O(\log^2 n)$ with worst-case query time $O(\log n/\log\log n)$,
and amortized expected update time $O(\log n \log\log n)$ with
worst-case query time $O(\log n)$~\cite{KaplanEtAl16}.
Very recently, Kauer and Mulzer~\cite{KauerMu19} 
showed how to improve the dependence on $\Psi$ in 
Theorem~\ref{thm:dynamicdiskgraph}, at the cost of 
slightly slower queries.

\subsection{Breadth-first-search in disk graphs}

As observed by Roditty and Segal~\cite{RodittySe11} in
the context of unit disk graphs, a dynamic nearest neighbor
structure can be used for computing exact BFS-trees in
disk graphs. More precisely, let $D(S)$ be
a disk graph with $n$ sites as in 
Section~\ref{sec:dyn_connect}, and let
$r \in S$. To compute a BFS-tree with root
$r$ in $D(S)$, we build a dynamic nearest neighbor
data structure for the weighted Euclidean distance 
(the weights correspond to the radii) and
we insert all points from $S \setminus \{r\}$.
At each point in time, the dynamic nearest
neighbor data structure contains those sites that
are not yet part of the BFS-tree. To find the new
neighbors of a site $p$ of the partial BFS-tree $T$, 
we repeatedly find and delete a nearest neighbor of $p$
in $S \setminus T$, 
until the next nearest neighbor is not adjacent to $p$
in $D(S)$. By construction, the other farther disks
are also not neighbors of $p$. The successful queries are charged 
to the BFS-edges, and the last unsuccessful query is
charged to $p$. Thus, the total number of operations
on the data structure is $O(n)$.
We get the following theorem:

\begin{theorem}
Let $S$ be a set of $n$ weighted sites in the plane, and let
$r \in S$. Then, we can compute a BFS-tree in $D(S)$ with
root $r$ in total expected time 
$O(n \log^{9} n\, \lambda_\swed(\log n))$.
\end{theorem}

\subsection{Spanners for disk graphs}

Finally, we discuss how to use our data structure
to efficiently  compute \emph{spanners} 
in disk graphs; see also Seiferth's thesis~\cite{Seiferth16} for
more details. Let $D(S)$ be a disk graph with $n$ sites as in
Section~\ref{sec:dyn_connect}, and let $\eps > 0$.
A $(1+\eps)$\emph{-spanner} for $D(S)$ is a subgraph
$H \subseteq D(S)$ such that, for any $s,t \in S$, the 
shortest-path distance $d_H(s,t)$ between $s$
and $t$ in $H$ is at most $(1+\eps) d(s,t)$, where
$d(s,t)$ is the shortest-path distance in $D(S)$.
F\"urer and Kasiviswanathan~\cite{FurerKa12} show that a simple
construction based on the \emph{Yao graph}~\cite{Yao82}
yields a $(1+\eps)$-spanner for $D(S)$ with
$O(n/\eps)$ edges. It goes as follows:
let $\C$ be a set of $k = O(1/\eps)$ cones, each with
opening angle $2\pi/k$ that
partition the plane. For each site $t \in S$, we translate $\C$ to
$t$, and for each translated cone $C$, we select a
site $s \in S$ with the following properties (if it exists):
(i) $s$ lies in $C$ and $st$ is an edge of $D(S)$;
(ii) we have $w_s \geq w_t$; and (iii) among all sites
with properties (i) and (ii), $s$ minimizes the distance
to $t$.  We add the edge $st$ to $H$.
F\"urer and Kasiviswanathan show that
this construction yields a $(1+\eps)$-spanner
for $D(S)$~\cite[Lemma~1]{FurerKa12}.
However, it is not clear how to implement this construction
efficiently. Therefore, F\"urer and Kasiviswanathan show
that it is sufficient to relax property~(iii) and to require only an
\emph{approximate} shortest edge in each cone. Using this,
they 
construct
such a relaxed spanner in time
$O(n^{4/3+\delta}\eps^{-4/3} \log^{2/3} \Psi)$, where
$\delta > 0$ can be made arbitrarily small and all
radii lie in the interval $[1, \Psi]$.

We can improve this running time by combining our new dynamic
nearest neighbor structure with techniques that we have developed for
\emph{transmission graphs}~\cite{KaplanEtAl15}.
Let $S$ be a set of $n$ weighted point sites as above.
The \emph{transmission graph} of $S$ is a \emph{directed} graph on 
$S$ with an edge from  $s$ to $t$ if and only if $|st| \leq w_s$, 
i.e., $t$ lies in the disk of $s$. A similar Yao-based
construction as above yields a $(1+\eps)$-spanner for directed 
transmission graphs: take for each site $t\in S$ and each cone $C$ 
the shortest incoming edge to $t$ in $C$.  Again, it is not
clear how to obtain this spanner efficiently.
To solve this problem, Kaplan \etal~\cite{KaplanEtAl15} proceed as 
F\"urer and Kasiviswanathan and describe strategies to compute 
relaxed versions of this spanner using only an approximate shortest
edge in each cone.

One strategy to obtain a running time of $O(n \log^4 n)$
is as follows:\footnote{The original result claims a running 
time of $O(n \log^5 n)$, but with Chan's recent improvement 
for dynamic Euclidean nearest neighbors~\cite{Chan19},
this result also improves.}
we compute a \emph{compressed quadtree} $T$ for 
$S$~\cite{sariels-book}.  Let $\sigma$ be a cell in $T$, and let 
$|\sigma|$ be the diameter of $\sigma$.  We augment $T$ such that 
for every edge $st$ in the transmission
graph, there exist cells $\sigma, \tau$ in $T$ with
diameters $|\sigma| = |\tau| = \Theta(\eps |st|)$ and
with $s \in \sigma$, $t \in \tau$. In particular, if $st$ is
the shortest edge in a cone with apex $t$, then
any edge $s't$ with $s' \in \sigma$ is sufficient for our relaxed
spanner, see Figure~\ref{fig:approximateedge}.
Kaplan \etal show that this augmentation requires adding $O(n)$ 
additional nodes to $T$, which can be found in $O(n \log n)$ time.
Furthermore, we compute for each cell $\sigma$ the set
$W_\sigma = 
\{ s \in \sigma \cap S \mid w_s = \Theta(|\sigma|/\eps) \}$.
Our strategy is to select spanner edges between sites in cells 
$\sigma,\tau \in T$ with $|\sigma| = |\tau|$ whose distance is 
$\Theta(|\sigma|/\eps)$.  Since a site can be contained in many cells
of $T$, we consider for each pair $\sigma,\tau$ only the sites in 
$W_\sigma$ for outgoing edges.  This avoids checking sites in 
$\sigma$ whose radius is too small to form an edge with sites in 
$\tau$. Sites whose radius is too large to be in $W_\sigma$ can be 
handled easily; see below. By definition of the sets $W_\sigma$ each
site appears in a constant number of such sets, which is
crucial in obtaining an improved running time.

\begin{figure}[htb]
\centering
\includegraphics[scale=0.5]{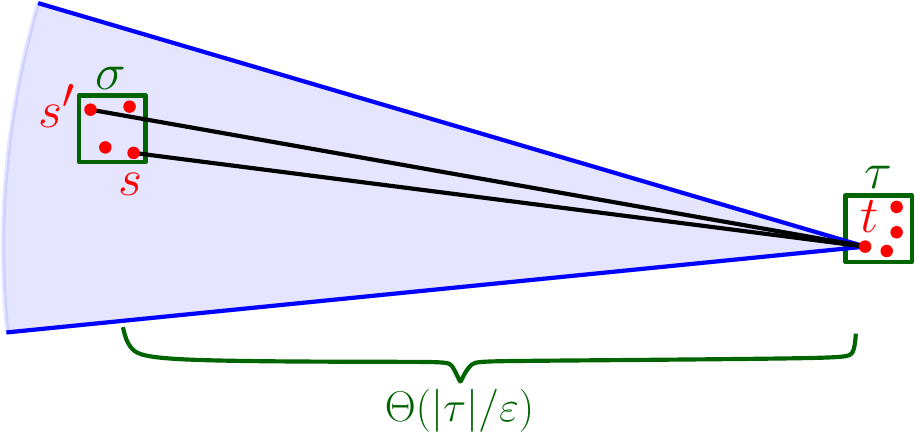}
\caption{A cone $C$ (blue) and the shortest edge $st$ in this cone.
  Any edge $s't$  with $s' \in \sigma$ has approximately the same
  length as $st$.}\label{fig:approximateedge}
\end{figure}

Now we can sketch the construction algorithm for the spanner $H$.
We go through all cones $C \in \mathcal{C}$.
For each $C$, we perform a level order traversal of the
cells in $T$, starting with the lowest level. For each cell
$\tau$ in $T$, we find the approximate incoming edges of length
$\Theta(|\tau|/\eps)$
with respect to $C$ that go into
the \emph{active} sites in $S \cap \tau$, i.e., those sites in
$\tau$ for which no such edge has been found in
a previous level. See Algorithm~\ref{alg:NNedgeselection} for
pseudocode how to
process a pair $C,\tau$.
To do this, we consider the cells of $T$ that have diameter $|\tau|$
and distance $\Theta(|\tau|/\eps)$ from $\tau$ and that
intersect the translated copy of $C$ whose center is in the
center of $\tau$. For each
such cell $\sigma$, we first check whether the disk corresponding to
the largest site $s$ in $\sigma$ contains $\tau$ completely. If so, 
we add edges from $s$ to all active sites in $t \in \tau$. This step 
covers all edges from sites in $\sigma$ whose radius is too large to 
be in $W_\sigma$.  Otherwise, we check for each site
in $W_\sigma$ whether it has edges to active sites in $\tau$.
This test is performed with a dynamic Euclidean nearest neighbor
data structure that stores the active sites in $\tau$: while the 
nearest neighbor $t$ in $\tau$ for the current site $s \in W_\sigma$ 
has $|st| \leq w_s$, we add the edge $st$ to $H$, and we remove $t$ 
from the nearest neighbor structure. Otherwise, we proceed to the 
next site in $W_\sigma$.  The resulting graph $H$ has $O(n/\eps^2)$ 
edges and contains for each site $t$ and for each cone $C$ attached 
to $t$ an approximately shortest incoming edge for $t$.

\begin{algorithm}[htb]
let $\gamma$ be the child of $\tau$ whose nearest neighbor
structure $\NN_\gamma$ contains the most sites\\
\label{line:NNpreproccesing}
for each child $\gamma' \neq \gamma$ of $\tau$, insert all
sites in $\gamma'$ into
$\NN_\gamma$; let $\NN_\tau = \NN_\gamma$ \\
\ForEach{$\sigma \in T$ \textnormal{with $|\sigma| = |\tau|$ and
distance}
$O(|\tau|/\eps)$ \textnormal{from} $\tau$
\textnormal{that intersects $C$}}{
  \If{\textnormal{disk of site} $s \in \sigma$ \textnormal{with
largest weight contains} $\tau$}{
    for each $t \in \NN_\tau$ add the edge $st$ to $H$
  }
  \Else{
    \ForEach{$s \in W_{\sigma}$}{
\label{alg:line:edgeselectionstart}
      $t \gets \NN_\tau(s)$  \tcp*{query NN structure of $\tau$ 
        with $s$}
      \label{line:firstNNquery}
      \While{$|st| \leq  w_s$ \textnormal{and} $t \neq \emptyset$}{
      \label{line:NNwhile}
    add the edge $st$ to $H$; delete $t$ from $\NN_\tau$; $t \gets
\NN_\tau(s)$
      }
    }
    reinsert all deleted points into $\NN_\tau$
\label{alg:line:edgeselectionend}
  }

}
delete all sites $t$ from $\NN_\tau$ for which at least one edge 
$st$ was found (i.e., make them inactive)
\caption{Selecting incoming edges for the points of a node $\tau$ 
  of $T$ and a cone $C$.}\label{alg:NNedgeselection}
\end{algorithm}

The nearest neighbor structures can be maintained
with logarithmic overhead throughout the level-order traversal: we
initialize them at the leaves of $T$, and when going to the next
level, we obtain the nearest neighbor structure for each cell
by inserting the elements of the smaller child structures into the
largest child structure. For more details, we refer to 
Kaplan~\etal~\cite{KaplanEtAl15} and to the thesis of 
Seiferth~\cite{Seiferth16}. They prove that the running time is 
dominated by the $O(n \log n)$ insertions and $O(n/\eps^2)$ 
deletions to the dynamic nearest neighbor structure.

A similar strategy works for disk graphs.
Given $S$, we compute an augmented quadtree $T$ for $S$ as
above in order to obtain an approximate representation of the 
distances in $S$. Furthermore, we compute for each cell $\sigma$ 
in $T$ an appropriate set $W_\sigma$ of assigned sites $s$ from 
$S \cap \sigma$ with $w_s = \Theta(|\sigma|/\eps)$, as above.
To construct the spanner, we perform the level order traversals of
the cells in $T$ as before, going through all cones 
$C \in \mathcal{C}$ and through all cells in $T$ from bottom to top.
Now, suppose we  visit a cell $\tau$ of $T$,
and let $\sigma$ be a cell of $T$ with diameter $|\sigma| = |\tau|$
and distance $\Theta(|\tau|/\eps)$ from $\tau$ that
intersects the translated copy of $C$ with apex in the middle of 
$\tau$.  As in Algorithm~\ref{alg:NNedgeselection},
our goal is to find all ``incoming'' edges from $W_\sigma$
for the active sites in $\tau$,
where an \emph{incoming} edge for $\tau$ now is an edge $st$
with $t \in \tau$ and $w_s \geq w_t$
(recall property (ii) from the original construction of
F\"urer and Kasiviswanathan).
We store the active sites of $\tau$ in a dynamic
nearest neighbor data structure $\NN_\tau$ for the metric 
$\delta(s,t) = |st| - w_t$, instead of the Euclidean metric.
To ensure that we find only edges from
larger to smaller disks, we sort the disks in $W_\sigma$ by radius,
and besides $\NN_\tau$, we also maintain a list $L_\tau$ of all 
sites in $\tau \cap S$ sorted by radius during the traversal of $T$.

We change
lines~\ref{alg:line:edgeselectionstart}--\ref{alg:line:edgeselectionend} in
Algorithm~\ref{alg:NNedgeselection} as follows:
we query the sites from $W_\sigma$ in order from small to large. 
Before querying a site $s$, we use $L_\tau$ to insert into 
$\NN_\tau$ all active sites
with weight at most $w_s$ that are not yet in $\NN_\tau$.
We keep querying $\NN_\tau$ with $s$ as long
as the resulting nearest neighbor $t$ corresponds to a disk that 
intersects the disk of $s$, and we add these edges $st$ to $H$. 
After that, we proceed to the site $s' \in W_\sigma$ with the next 
larger radius, and we again insert all remaining active
sites with weight at most $w_{s'}$ from $L_\tau$ into $\NN_\tau$ 
(all these sites $t$ have $w_s \leq w_t \leq w_{s'}$).
After processing $W_\sigma$, we proceed with the next cell
$\sigma'$.  To
ensure that our nearest neighbor queries still return only
smaller disks, we
need to delete all sites in $\NN_\tau$ whose weight is larger than 
the smallest weight in $W_{\sigma'}$. This can be done by deleting all
sites in $W_\tau$ from $\NN_\tau$, and reinserting only the 
relevant ones. By definition of $W_\tau$, for 
each site the
additional insertions and deletions to maintain $\NN_\tau$ occur 
only for a constant  number of pairs $\sigma,\tau$, accounting for 
an additional $O(n)$ insertions and deletions per site.
An analysis similar to the one performed by 
Kaplan~\etal~\cite{KaplanEtAl15} for transmission graphs now 
shows that $H$ can be constructed in time $O(n \log n)$ plus the
time for $O(n \log n)$ insertions and $O(n/\eps^2)$ deletions in 
the dynamic nearest neighbor structure.  By Theorem~\ref{dyn-surf}, 
we thus obtain the following result:
\begin{theorem}
Let $S$ be a set of $n$ weighted sites in the plane, and let
$\eps > 0$. Then, we can construct a $(1 + \eps)$ spanner for $D(S)$
in expected time $O((n/\eps^2)\log^{9} n \lambda_\swed(\log n))$.
\end{theorem}


\bibliographystyle{plain}
\bibliography{literature}

\newcommand{\SortNoop}[1]{}
\begin{thebibliography}{10}

\bibitem{AC09}
Peyman Afshani and Timothy~M. Chan.
\newblock On approximate range counting and depth.
\newblock {\em Discrete Comput. Geom.}, 42(1):3--21, 2009.

\bibitem{Agarwal_etal98}
Pankaj~K. Agarwal, Mark de~Berg, Ji{\v{r}}{\'{\i}} Matou{\v{s}}ek, and Otfried
  Schwarzkopf.
\newblock Constructing levels in arrangements and higher order {V}oronoi
  diagrams.
\newblock {\em SIAM J. Comput.}, 27(3):654--667, 1998.

\bibitem{AES}
Pankaj~K. Agarwal, Alon Efrat, and Micha Sharir.
\newblock Vertical decomposition of shallow levels in 3-dimensional
  arrangements and its applications.
\newblock {\em SIAM J. Comput.}, 29(3):912--953, 1999.

\bibitem{AM95}
Pankaj~K. Agarwal and Ji{\v{r}}{\'{\i}} Matou{\v{s}}ek.
\newblock Dynamic half-space range reporting and its applications.
\newblock {\em Algorithmica}, 13(4):325--345, 1995.

\bibitem{vor-book}
Franz Aurenhammer, Rolf Klein, and Der{-}Tsai Lee.
\newblock {\em {V}oronoi Diagrams and {D}elaunay Triangulations}.
\newblock World Scientific Publishing, Singapore, 2013.

\bibitem{BPR}
Saugata Basu, Richard Pollack, and Marie-Fran{\c{c}}oise Roy.
\newblock {\em Algorithms in Real Algebraic Geometry}, volume~10 of {\em
  Algorithms and Computation in Mathematics}.
\newblock Springer-Verlag, Berlin Heidelberg, second edition, 2006.

\bibitem{BS}
Jon~Louis Bentley and James~B. Saxe.
\newblock Decomposable searching problems. {I}. {S}tatic-to-dynamic
  transformation.
\newblock {\em J. Algorithms}, 1(4):301--358, 1980.

\bibitem{dBCvKO}
Mark {\SortNoop{Berg}}de~Berg, Otfried Cheong, Marc van Kreveld, and Mark~H.
  Overmars.
\newblock {\em Computational Geometry: Algorithms and applications}.
\newblock Springer-Verlag, Berlin Heidelberg, third edition, 2008.

\bibitem{ECG}
Jean-Daniel Boissonnat and Monique Teillaud.
\newblock {\em Effective Computational Geometry for Curves and Surfaces}.
\newblock Mathematics and Visualization. Springer-Verlag, Berlin Heidelberg,
  2007.

\bibitem{CabelloJejcic15}
Sergio Cabello and Miha Jej\^ci\^c.
\newblock Shortest paths in intersection graphs of unit disks.
\newblock {\em Comput. Geom.}, 48(4):360--367, 2015.

\bibitem{Chan2000}
Timothy~M. Chan.
\newblock Random sampling, halfspace range reporting, and construction of $(\le
  k)$-levels in three dimensions.
\newblock {\em SIAM J. Comput.}, 30(2):561--575, 2000.

\bibitem{Cha05}
Timothy~M. Chan.
\newblock Low-dimensional linear programming with violations.
\newblock {\em SIAM J. Comput.}, 34(4):879--893, 2005.

\bibitem{Cha10}
Timothy~M. Chan.
\newblock A dynamic data structure for 3-{D} convex hulls and 2-{D} nearest
  neighbor queries.
\newblock {\em J. ACM}, 57(3):16:1--15, 2010.

\bibitem{Chan19}
Timothy~M. Chan.
\newblock Dynamic geometric data structures via shallow cuttings.
\newblock In {\em Proc. 35th Int. Sympos. Comput. Geom. (SoCG)}, page to
  appear, 2019.

\bibitem{ChanPaRo11}
Timothy~M. Chan, Mihai P{\v{a}}tra{\c{s}}cu, and Liam Roditty.
\newblock Dynamic connectivity: connecting to networks and geometry.
\newblock {\em SIAM J. Comput.}, 40(2):333--349, 2011.

\bibitem{CT15}
Timothy~M. Chan and Konstantinos Tsakalidis.
\newblock Optimal deterministic algorithms for 2-d and 3-d shallow cuttings.
\newblock {\em Discrete Comput. Geom.}, 56(4):866--881, 2016.

\bibitem{Ch}
Bernard Chazelle.
\newblock Cutting hyperplanes for divide-and-conquer.
\newblock {\em Discrete Comput. Geom.}, 9(2):145--158, 1993.

\bibitem{Chazelle00}
Bernard Chazelle.
\newblock {\em The Discrepancy Method}.
\newblock Cambridge University Press, 2000.

\bibitem{CEGS}
Bernard Chazelle, Herbert Edelsbrunner, Leonidas~J. Guibas, and Micha Sharir.
\newblock A singly-exponential stratification scheme for real semi-algebraic
  varieties and its applications.
\newblock In {\em Proc. 31st Internat. Colloq. Automata Lang. Program.
  (ICALP)}, pages 179--193, 1989.

\bibitem{ChewDr85}
L.~Paul Chew and Robert L. (Scot)~Drysdale III.
\newblock {V}oronoi diagrams based on convex distance functions.
\newblock In {\em Proc. 1st Annu. Sympos. Comput. Geom. (SoCG)}, pages
  235--244, 1985.

\bibitem{Chvatal79}
Vasek Chv{\'{a}}tal.
\newblock The tail of the hypergeometric distribution.
\newblock {\em Discrete Mathematics}, 25(3):285--287, 1979.

\bibitem{Clarkson88}
Kenneth~L. Clarkson.
\newblock A randomized algorithm for closest-point queries.
\newblock {\em SIAM J. Comput.}, 17(4):830--847, 1988.

\bibitem{CS89}
Kenneth~L. Clarkson and Peter~W. Shor.
\newblock Applications of random sampling in computational geometry. {II}.
\newblock {\em Discrete Comput. Geom.}, 4(5):387--421, 1989.

\bibitem{EdSe}
Herbert Edelsbrunner and Raimund Seidel.
\newblock Voronoi diagrams and arrangements.
\newblock {\em Discrete Comput. Geom.}, 1(1):25--44, 1986.

\bibitem{Eppstein95}
David Eppstein.
\newblock Dynamic {E}uclidean minimum spanning trees and extrema of binary
  functions.
\newblock {\em Discrete Comput. Geom.}, 13:111--122, 1995.

\bibitem{Er}
Jeff Erickson.
\newblock Static-to-dynamic transformations.
\newblock Lecture notes.
\newblock \url{
  http://jeffe.cs.illinois.edu/teaching/datastructures/notes/01-statictodynamic.pd
  f}.

\bibitem{FurerKa12}
Martin F{\"{u}}rer and Shiva~Prasad Kasiviswanathan.
\newblock Spanners for geometric intersection graphs with applications.
\newblock {\em J. Comput. Geom.}, 3(1):31--64, 2012.

\bibitem{sariels-book}
Sariel Har-Peled.
\newblock {\em Geometric Approximation Algorithms}, volume 173 of {\em
  Mathematical Surveys and Monographs}.
\newblock American Mathematical Society, Providence RI, 2011.

\bibitem{HKS}
Sariel Har-Peled, Haim Kaplan, and Micha Sharir.
\newblock Approximating the $k$-level in three-dimensional plane arrangements.
\newblock In M.~Loebl, J.~Ne\v{s}et\v{r}il, and R.~Thomas, editors, {\em
  Journey through Discrete Mathematics: A Tribute to Ji\v{r}\'{\i}
  Matou\v{s}ek}, pages 467--504. Springer-Verlag, Berlin Heidelberg, 2017.

\bibitem{HPS}
Sariel Har-Peled and Micha Sharir.
\newblock Relative {$(p,\epsilon)$}-approximations in geometry.
\newblock {\em Discrete Comput. Geom.}, 45(3):462--496, 2011.

\bibitem{HolmEtAl01}
Jacob Holm, Kristian de~Lichtenberg, and Mikkel Thorup.
\newblock Poly-logarithmic deterministic fully-dynamic algorithms for
  connectivity, minimum spanning tree, 2-edge, and biconnectivity.
\newblock {\em J. ACM}, 48(4):723--760, 2001.

\bibitem{KaplanEtAl15}
Haim Kaplan, Wolfgang Mulzer, Liam Roditty, and Paul Seiferth.
\newblock Spanners and reachability oracles for directed transmission graphs.
\newblock In {\em Proc. 31st Int. Sympos. Comput. Geom. (SoCG)}, pages
  156--170, 2015.

\bibitem{KaplanEtAl16}
Haim Kaplan, Wolfgang Mulzer, Liam Roditty, and Paul Seiferth.
\newblock Dynamic connectivity for unit disk graphs.
\newblock In {\em Proc. 32nd European Workshop Comput. Geom. (EWCG)}, pages
  183--186, 2016.

\bibitem{KaplanMuRoSe18}
Haim Kaplan, Wolfgang Mulzer, Liam Roditty, and Paul Seiferth.
\newblock Spanners for directed transmission graphs.
\newblock {\em SIAM J. Comput.}, 47(4):1585--1609, 2018.

\bibitem{KauerMu19}
Alexander Kauer and Wolfgang Mulzer.
\newblock Dynamic disk connectivity.
\newblock In {\em Proc. 35th European Workshop Comput. Geom. (EWCG)}, pages
  50:1--6, 2019.

\bibitem{Kol}
Vladlen Koltun.
\newblock Almost tight upper bounds for vertical decompositions in four
  dimensions.
\newblock {\em J. ACM}, 51(5):699--730, 2004.

\bibitem{Lee80}
D.~T. Lee.
\newblock Two-dimensional {V}orono\u\i\ diagrams in the {$L\sb{p}$}-metric.
\newblock {\em J. ACM}, 27(4):604--618, 1980.

\bibitem{LLS}
Yi~Li, Philip~M. Long, and Aravind Srinivasan.
\newblock Improved bounds on the sample complexity of learning.
\newblock {\em J. Comput. System Sci.}, 62(3):516--527, 2001.

\bibitem{Ma:rph}
Ji{\v{r}}{\'{\i}} Matou{\v{s}}ek.
\newblock Reporting points in halfspaces.
\newblock {\em Comput. Geom.}, 2(3):169--186, 1992.

\bibitem{Mulmuley91}
Ketan Mulmuley.
\newblock On levels in arrangements and {V}oronoi diagrams.
\newblock {\em Discrete Comput. Geom.}, 6:307--338, 1991.

\bibitem{mulmuley1994computational}
Ketan Mulmuley.
\newblock {\em Computational Geometry: An Introduction Through Randomized
  Algorithms}.
\newblock Prentice-Hall, 1994.

\bibitem{Mulzer18}
Wolfgang Mulzer.
\newblock Five proofs of {C}hernoff's bound with applications.
\newblock {\em Bulletin of the {EATCS}}, 124, 2018.

\bibitem{OvL}
Mark~H. Overmars and Jan van Leeuwen.
\newblock Two general methods for dynamizing decomposable searching problems.
\newblock {\em Computing}, 26(2):155--166, 1981.

\bibitem{Ram99}
Edgar~A. Ramos.
\newblock On range reporting, ray shooting and {$k$}-level construction.
\newblock In {\em Proc. 15th Annu. Sympos. Comput. Geom. (SoCG)}, pages
  390--399 (electronic), 1999.

\bibitem{RodittySe11}
Liam Roditty and Michael Segal.
\newblock On bounded leg shortest paths problems.
\newblock {\em Algorithmica}, 59(4):583--600, 2011.

\bibitem{SS2}
Jacob~T. Schwartz and Micha Sharir.
\newblock On the ``piano movers'' problem. {II}. {G}eneral techniques for
  computing topological properties of real algebraic manifolds.
\newblock {\em Adv. in Appl. Math.}, 4(3):298--351, 1983.

\bibitem{Seiferth16}
Paul Seiferth.
\newblock {\em Disk Intersection Graphs: Models, Data Structures, and
  Algorithms}.
\newblock PhD thesis, Freie Universit\"at Berlin, 2016.

\bibitem{SA}
Micha Sharir and Pankaj~K. Agarwal.
\newblock {\em {D}avenport-{S}chinzel sequences and their geometric
  applications}.
\newblock Cambridge University Press, 1995.

\bibitem{Vaidya89}
Pravin~M. Vaidya.
\newblock Geometry helps in matching.
\newblock {\em SIAM J. Comput.}, 18(6):1201--1225, 1989.

\bibitem{WangXu19}
Haitao Wang and Jie Xue.
\newblock Near-optimal algorithms for shortest paths in weighted unit-disk
  graphs.
\newblock In {\em Proc. 35th Int. Sympos. Comput. Geom. (SoCG)}, page to
  appear, 2019.

\bibitem{Yao82}
Andrew~Chi{-}Chih Yao.
\newblock On constructing minimum spanning trees in k-dimensional spaces and
  related problems.
\newblock {\em SIAM J. Comput.}, 11(4):721--736, 1982.

\end{thebibliography}

\appendix

\section{The randomized incremental construction}\label{app:ric}

We provide the details of the randomized incremental
construction from Section~\ref{sec:ric}.

\subsection{Setup}

Each three-dimensional cell $\tau$ of $\VD_{\le t}(F_i)$ 
is a \emph{pseudo-prism}, with up to six faces.
We have already used ``floor'' and ``ceiling'' 
to refer, respectively, to the bottom and the top face of $\tau$;
we also call them the \emph{$xy$-faces} of $\tau$,
as they extend (in a suitable sense) in the $x$- and $y$-directions.
The \emph{forward $xz$-face} and the \emph{backward $xz$}-face are 
the  ``curtains'' that extend in the $x$- and $z$-directions and 
bound $\tau$ in the respective 
positive and negative $y$-directions. They are each erected from 
an edge of $\tau$ that is a portion of an intersection curve 
between the surface supporting the floor or ceiling of $\tau$ 
with another surface. We collectively refer to these curtains 
as the \emph{$xz$-faces} of $\tau$.
Similarly, we refer to the left and right faces collectively as 
\emph{$yz$-faces}. They originate from lifting the 
vertical edges of the planar vertical decomposition of the
stage-1 cells; see Figure~\ref{fig:xywalls} for an illustration. 
A similar notation applies to the edges of $\tau$. There are three 
kinds of edges: (i) an \emph{$x$-edge}, which is the common edge of 
an $xy$-face 
and an $xz$-face of $\tau$. It is either (a portion of) 
a \emph{real} intersection curve, or a \emph{shadow edge}, that lies 
vertically below or above a real intersection edge on the other 
(floor or ceiling) side of $\tau$; (ii) a \emph{$y$-edge}, which 
is the common edge of an $xy$-face and a
$yz$-face of $\tau$; and (iii) a $z$-edge, a straight $z$-parallel 
segment, which is a common edge of an $xz$-face and a $yz$-face.

\begin{figure}[htb]
\begin{center}
\includegraphics{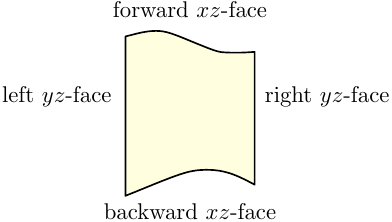}
\caption{A top view of a prism $\tau$ and its 
$xz$-faces and $yz$-faces.}\label{fig:xywalls}
\end{center}
\end{figure}

To navigate in $\VD_\leq(F_i)$, each prism $\tau \in \VD_\leq(F_i)$
maintains pointers to the \emph{adjacent} 
prisms that share (part of) a $yz$-face
with it.
By general position, there are only $O(1)$ such
prisms. Moreover, for each vertex $v$ of a stage-1 cell
(i.e., an intersection of three surfaces in $F_i$; the projection
of such an intersection onto a vertically visible surface in
$F_i$ above or below; or the projection of a vertical visibility
between two edges onto the top or the bottom edge of the pair), we 
maintain pointers to all \emph{incident} prisms that 
have $v$ as a vertex, and in each prism,
we maintain reverse pointers to the incident stage-1 vertices. 
Again by general position, we need only $O(1)$
pointers with each vertex and with each prism.
These pointers allow us to perform a BFS-search within the prisms of 
a stage-1 cell, and to switch from one stage-1 cell to another one
with a common vertex.
Finally, we maintain a pointer
to a vertex that lies on the $t$-level $L_t(F_i)$
of the current arrangement. This vertex will serve as a 
starting point when we need to traverse the $t$-level.

\subsection{Inserting a surface}

Suppose we add a new surface $f=f_{i+1}$ to 
some prefix $F_i$ of $F$. We have the vertical decomposition 
$\VD_{\le t}(F_i)$, where each prism $\tau \in \VD_{\le t}(F_i)$ has an 
associated conflict list $\CL(\tau)$, consisting of the surfaces 
of $F\setminus F_i$ that cross $\tau$. Our task is to obtain 
$\VD_{\le t}(F_{i+1})$ with the conflict lists of its prisms, 
each consisting of surfaces in $F\setminus F_{i+1}$.
We do this in three steps. First, we find the vertical 
decomposition of the part of the arrangement $\A(F_{i+1})$ 
that is below the $t$-level of $\A(F_{i})$; we denote this 
portion by $\A^{+f}_{\le t}(F_{i})$. 
Second, we construct the conflict lists of the new prisms.
Third, we discard prisms that
lie above $\A_{\le t}(F_{i+1})$ from 
the vertical decomposition obtained in the first step.

\paragraph{First step: Constructing the vertical 
decomposition of $\A^{+f}_{\le t}(F_{i})$.}
When $f$ is inserted, we retrieve the set $\Pi_f$ of all 
\emph{old} prisms of 
$\VD_{\le t}(F_i)$ that $f$ crosses, using back pointers to the
conflict lists that contain $f$. 
Note that any old prism 
$\tau \not\in \Pi_f$ remains valid in the 
vertical decomposition of $\A^{+f}_{\le t}(F_{i})$, but, in case
$f$ passes fully below $\tau$, the level of $\tau$ in $\A(F_{i+1})$
(compared to its level in $\A(F_{i})$) increases by $1$. 
Hence, $\tau$ may find itself above the $t$-level, in which case
we discard it in the third step. On the other hand, 
any old prism $\tau \in \Pi_f$ is destroyed,
being split into several \emph{fragments}. 
More precisely, we
compute the vertical decomposition for $f$ locally within 
$\tau$ (considered as a closed set). 
This splits $\tau$ into $O(1)$ smaller
pseudo-prisms (the fragments), and we compute the conflict
list for each fragment, by brute-force inspection of $\CL(\tau)$.
This takes $O\big(|\Pi_f| + \sum_{\tau \in \Pi_f} |\CL(\tau)|\big)$ 
time.  Our strategy is to use these fragments to
find the region of 
$A^{+f}_{\le t}(F_{i})$
that is affected by $f$ and to 
construct the \emph{new} prisms in this region from scratch.
The conflict lists of the fragments will be useful in 
finding the conflict lists for the new prisms.

Each new prism $\tau'$ must involve $f$ as one of its (up to) 
six defining surfaces. That is, it must contain a bounding feature 
that lies on $f$. This feature could be a face (when $f$ forms the 
floor or ceiling of $\tau'$); an $x$-edge (where $f$ intersects 
the floor or ceiling of $\tau'$ at a boundary edge, but does not
meet the interior of $\tau'$); or a vertex (which is either an 
intersection point of $f$ with an edge of $\tau'$ at its endpoint, 
or a locally $x$-extremal point of an intersection curve on $f$, 
which defines a $yz$-face of $\tau'$; see Figure~\ref{fig:xext}).

\begin{figure}[htb]
\begin{center}
\includegraphics{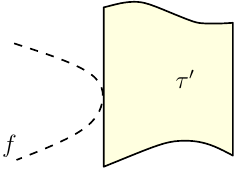}
\caption{A $yz$-face of a new prism $\tau'$ (seen here as the left side 
of this top view) formed by a locally $x$-extremal point of some intersection 
edge with $f$.}\label{fig:xext}
\end{center}
\end{figure}

\paragraph{New prisms with an $xy$-face on $f$.}
We first consider the construction of new prisms of this type.
We begin by finding and tracing 
the $x$-edges of $\VD(\A^{+f}_{\le t}(F_{i}))$ along $f$. 
More precisely, we collect the edges along $f$ of stage~1 
of the vertical decomposition of $\A^{+f}_{\le t}(F_{i})$.
Following the terminology just introduced, each such edge 
is either a real intersection edge between $f$ and an older surface, 
or a shadow edge, i.e., the vertical projection of the portions of 
some real intersection edge that are vertically visible from $f$.

We think of $f$ as two-sided, having a top
side and a bottom side, appearing as two disjoint copies of 
$f$, infinitesimally separated from one another.
The real intersection edges are drawn on both sides 
of $f$, but each shadow edge is drawn only on one side (top or bottom) 
of $f$; it is the top (resp., bottom) side of $f$ if the edge is a
shadow of a real edge above (resp., below) $f$. Thus, 
we obtain two different maps on $f$, called $M^t_f$ and $M^b_f$, drawn on 
the top and bottom sides of $f$, respectively.
The maps $M^t_f$ and $M^b_f$ have three kinds of vertices. The first kind
arises when an intersection edge $e$ between two other surfaces $f'$ and $f''$ 
crosses $f$, generating a real vertex $v$ of 
$\A^{+f}_{\le t}(F_{i})$.
We break $e$ into two subedges $e^+$, $e^-$, where $e^+$ lies above $f$ and 
$e^-$ lies below $f$, locally near $v$. Then, the vertical projection of $e^+$ 
on $f$ is drawn only in $M^t_f$, and that of $e^-$ only in
$M^b_f$.
Both projections are arcs emanating from $v$. 
See, e.g., Figure~\ref{fig:vdins1}, where the intersection curve between 
the surfaces $a$ and $b$ intersects $f$ in a vertex $v$.
The second kind of vertex arises from a real vertex $w$ of 
$\A^{+f}_{\le t}(F_{i})$, incident to three surfaces $f_1$, $f_2$, $f_3$, 
so that $w$ is vertically visible from $f$ and lies, say, above $f$.
Then, $w$ is incident to three intersection edges of pairs from   
$\{f_1, f_2, f_3\}$.
Each such edge is split at $w$ into two portions, one visible from $f$ 
and one invisible (hidden by the third function), so we draw in
$M^t_f$
three respective projected arcs, all emanating from the projection 
of $w$. 
See Figure~\ref{fig:vdins1} for an illustration, where 
the real vertex $w$ defined by $a$, $d$, and $g$ and the real vertex $w'$ 
defined by $d$, $h$, and $k$ form the shadow vertices $\tilde{w}$ and 
$\tilde{w}'$ on $f$.
\begin{figure}[htb]
\begin{center}
\includegraphics{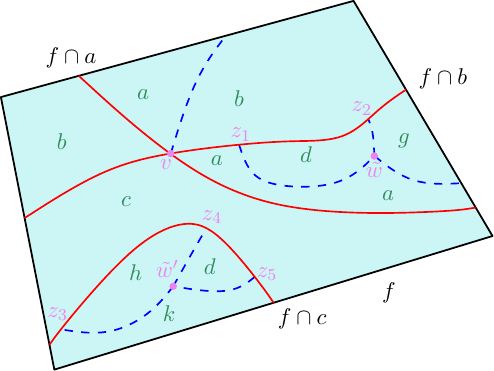}
\caption{The first stage of the vertical decomposition on (the top side of) 
the 
newly inserted surface $f$. Red (solid) arcs depict real intersection edges
and blue (dashed) arcs depict shadow arcs. The label of each face denotes 
the surface that appears vertically above $f$ and is visible from $f$
over that face. The three types
of vertices are depicted: $v$ is a real vertex, $\tilde{w}$ and $\tilde{w}'$
are shadow vertices, and $z_1,\ldots,z_5$ represent vertical 
visibilities.}\label{fig:vdins1}
\end{center}
\end{figure}
The third kind of vertex arises from a vertical visibility between
a real edge $e$ on $f$ and another edge $e'$ lying, say, above $f$. Let $z$
be the point on $e$ where this vertical visibility occurs. Then $e'$ is 
visible from $f$ only on one side of $z$, as the other surface forming 
with $f$ the edge $e$ rises above $f$ on the other side and hides $e'$.
We draw, on $M^t_f$ only, the projection of $e'$ on $f$, and terminate 
it at $z$.
See Figure~\ref{fig:vdins1}, where 
the vertices $z_1,\dots,z_5$ represent such pairs of vertical visibility.
Finally,
there might be shadow edges that do not cross any real edge (so they 
are not involved in any vertically visible pair). These edges are 
projections of full (and fully visible) edges of $\A(F_i)$, which 
are either closed or unbounded Jordan curves.

To construct $M_f^t$ and $M_f^b$, we go through all
old prisms $\tau \in \Pi_f$, 
and we compute the intersection of $f$ with
the boundary $\partial\tau$. The pieces of the real intersection 
edges are obtained from 
the intersections of $f$ with the floor and/or 
the ceiling of $\tau$. Each such intersection consists of $O(1)$ 
connected subarcs. 
When such an edge $e$ leaves $\tau$ (at an endpoint of some 
intersection subarc), it can do so either through a $y$-edge or through 
an $x$-edge. In the former case (crossing a $y$-edge), we must glue 
$e$ to a suitable portion of its continuation into the appropriate
old prism adjacent to $\tau$, whereas in the latter case (crossing an 
$x$-edge) the crossing point $v$ is either a real vertex of 
$\A^{+f}_{\le t}(F_{i})$ on $e$ (an endpoint of $e$), or part 
of a vertically visible pair in $\A^{+f}_{\le t}(F_{i})$ consisting 
of (a suitable extension of) $e$ and the real intersection edge of 
$\tau$ on its other side (floor or ceiling). When $v$ is a real vertex,
it delimits two edges of $\A^{+f}_{\le t}(F_{i})$, lying on the same
intersection curve, one of which is $e$ (within $\tau$) and the
other enters an adjacent prism; in this case $v$ is a feature 
of both $M^t_f$ and $M^b_f$. When $v$ encodes a vertical visibility 
pair, it appears only in one of the maps $M^t_f$ or $M^b_f$;
see Figure~\ref{fig:etrace}. 
\begin{figure}[htb]
\begin{center}
\includegraphics{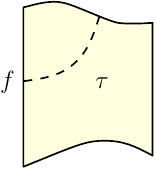}
\caption{The intersection curve of $f$ with the floor or ceiling 
of a prism $\tau$ (view from above). At the left exit point of the curve 
(through a $y$-edge) we glue it to its continuation within an adjacent prism.
The top exit point (through an $x$-edge) is a real 
feature (vertex or part of a vertically visible pair) of the new 
decomposition.}\label{fig:etrace}
\end{center}
\end{figure}

To obtain the shadow edges, we consider the intersections
of the $xz$-faces of $\tau$ with $f$, and we use
the local information at $\tau$.
For example, if an intersection edge $e$ (between two surfaces of $F_i$) 
lies on the top side of $\tau$, 
we take the $xz$-face $\varphi$ bounded from above by $e$, intersect $f$ 
with $\varphi$, and for each connected arc $e'$ of that intersection, 
we draw the shadow of $e$ in $M^t_f$ over $e'$ as $e'$ 
itself; see Figure~\ref{fig:eonet}. The other portions of $e$ (e.g.,
the middle portion in Figure~\ref{fig:eonet}) are not 
handled in $\tau$, since the information at $\tau$ does not 
let us know whether these pieces are at all visible (in their entirety) 
from $f$; these pieces are handled within other nearby prisms from $\Pi_f$ 
that have $e$ (or an edge overlapping $e$) as an $x$-edge. 
(Note that, in Figure~\ref{fig:eonet}, any visible part (to $f$) of the middle portion 
of $e$ is drawn on the bottom map $M^b_f$.)
Real intersection edges on the bottom side of $\tau$ 
are handled in a fully symmetric manner, and their relevant 
portions are drawn as shadow edges in $M^b_f$. 

\begin{figure}[htb]
\begin{center}
\includegraphics{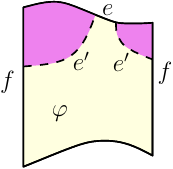}
\caption{Creating shadow edges (projections of suitable portions of $e$) on the top 
side of $f$ within a vertical curtain 
$\varphi$ (a side view).}\label{fig:eonet}
\end{center}
\end{figure}

We perform two planar sweeps to assemble the pieces of real and
shadow edges on the two sides of $f$,  and to obtain DCEL-representations
of $M^t_f$ and of $M^b_f$. 
For each stage-1 cell of the vertical decomposition of 
$\A^{+f}_{\le t}(F_{i})$ (i.e., the cells obtained
by erecting vertical curtains through the edges of the
arrangement) that has its floor or its ceiling on $f$,
there is a face 
in $M^t_f$ (resp., in $M^b_f$). Note that $M^t_f$ and $M^b_f$
may have additional faces that correspond to patches of $f$ that
lie above $\A_{\le t}(F_{i})$.
Consider for specificity only $M^t_f$.
Recall that each such stage-1 cell has a unique floor (a portion of $f$ in 
this case) and a unique ceiling (a portion of another surface), but
its complexity can be arbitrarily large (its floor and ceiling need 
not even be simply connected, although, by construction, they are 
always connected). Denote by $\B$ the collection of these stage-1 cells
that have their floor on $f$. Each cell $B \in \B$ is the union of fragments 
of prisms in $\Pi_f$ (each obtained as we decompose prisms 
of $\Pi_f$ into smaller ``local'' 
prisms) with their floor on $f$. Each prism from $\Pi_f$ generates $O(1)$ 
such fragments, and we use
a planar point location structure for $M^t_f$ 
to find, for each
fragment that has its floor on $f$, the cell of $\B$ that contains
the fragment.
Thus, we can, for each $B \in \B$, compute the conflict list
$\CL(B)$ of $B$ by taking the union of the conflict lists of  
fragments that make up $B$.
The bottom map $M^b_f$ is handled similarly.
The total time for this whole step 
is $O\big(|\Pi_f|\log n + \sum_{\tau \in \Pi_f} |\CL(\tau)|\big)$.

\paragraph{Constructing the new stage-2 prisms.}
Next, we perform stage~2 of the vertical 
decomposition within each cell $B$ of $\B$, by constructing the 
$yz$-faces that partition it 
prisms of constant complexity. 
To do this, we take the floor $B_f$ of $B$ (which is a 
portion of $f$), project it onto the $xy$-plane, and compute the 
vertical decomposition, denoted $\VD(B_f)$, of $B_f$ using a planar sweep. 
The $y$-edges of $\VD(B_f)$, when lifted to three dimensions and 
intersected with $B$, define the $yz$-faces of the desired vertical 
decomposition of $B$, which we denote as $\VD(B)$.
In this step, we also compute the navigational pointers between 
adjacent prisms within a stage-2 cell and between 
the prisms and their incident stage-1 vertices.
The cells for the bottom map $M_f^b$ are handled similarly.
The total running time for this step is $O\big(|\Pi_f| \log n \big)$.

\paragraph{Prisms for which $f$ defines an $x$-edge or a
$yz$-face (that passes through a vertex on $f$).}
Consider the new prisms of this kind whose bottom faces lie on 
some $g\in F_i$. To construct these prisms, 
we draw on $g$ the intersection edges of $g$ with $f$ (which we 
have already computed). Let $e$ be such an intersection edge.
Vertically above $g$, on one side of $e$, we have prisms whose 
ceiling is on $f$, which we have already computed. On the 
other side of $e$, the surface $g$ is the floor of new prisms that are
obtained from suitable fragments of old prisms that were cut 
by $f$ and intersect $e$. See, e.g., the middle portion of the prism
depicted in Figure~\ref{fig:eonet}, and see also below.

Thus, we collect all fragments of prisms from $\Pi_f$ that 
have their floor on $g$, draw their 
projections on $g$, and trace these projections with a vertical sweep 
to obtain the 
old real or shadow edges that appear on these prisms, as done 
in the previous case. This produces new (partial) stage-1 cells
that have their floor on $g$ and for which $f$ only defines
an edge or a vertex, together with their conflict lists. 
We then construct the planar vertical
decomposition of their projections onto the $xy$-plane, and lift
each resulting trapezoid into a stage-2 prism, which collectively
yield the set of new prisms of this kind. This proceeds very much 
like the processing in the former case, and we also handle 
the conflict lists in a similar way.
We repeat this process for each side of each surface intersected 
by $f$. The total running time  
is $O\big(|\Pi_f| \log n +  \sum_{\tau \in \Pi_f} |\CL(\tau)|\big)$;
see Figure~\ref{fig:vdins2}.

\begin{figure}[htb]
\begin{center}
\includegraphics{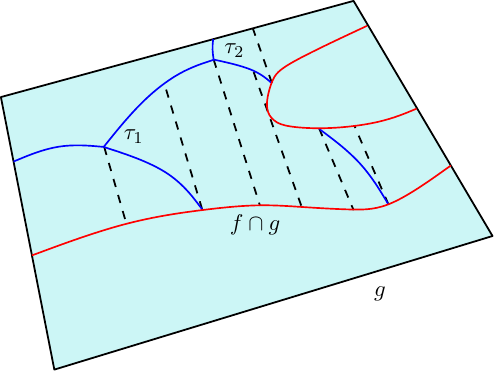}
\caption{Collecting prisms to which $f$ 
contributes only an $x$-edge or a vertex.
The figure depicts floors of such prisms on an older surface $g$. The red arcs 
depict $f\cap g$, and the blue arcs depict (portions of) older real
or shadow arcs. For the prisms $\tau_1$ and $\tau_2$, $f$ only contributes
a vertex. For all the other prisms, $f$ contributes an $x$-edge. $f$ passes 
below $g$ in the region decomposed into prisms, and above $g$ in the 
complementary regions, which are decomposed into prisms that have $f$ as a ceiling,
and have already been constructed.}\label{fig:vdins2}
\end{center}
\end{figure}

\paragraph{Second step: Constructing the conflict lists.}
We next find, for each new stage-1 cell $B$,
the conflict lists of the new stage-2 prisms $\tau \in \VD(B)$.
Assume, without loss of generality, that the floor of 
$B$ is on $f$, 
and let $f^+$ denote the surface on the ceiling of $B$.
Let $B_0$ denote the $xy$-projection of $B$.
Let $g$ be a surface that crosses $B$, and let $B(g)$ denote
the portion of $B_0$ over which $g$ lies between $f$ and $f^+$
(that is, $g$ is in $B$). Clearly, $g\in\CL(\tau)$ precisely for
those prisms $\tau \in \VD(B)$ whose $xy$-projections intersect $B(g)$. 
Note that it is possible that $B(g) = B_0$, in which
case $g$ does not cross $f$ or $f^+$ at all over $B_0$. Then 
$g$ belongs to the conflict lists of all new prisms in $\VD(B)$;
see Figure~\ref{fig:bg}.

\begin{figure}[htb]
\begin{center}
\includegraphics{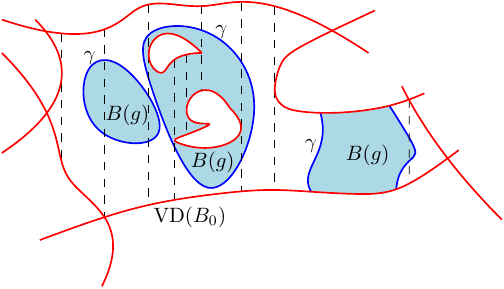}
\caption{The $xy$-projection $B_0$ of a new stage-1 cell $B$ (involving $f$), 
and its interaction with a crossing surface $g$. The shaded region is
$B(g)$ and its boundary $\gamma$ consists of projections of portions
where $g$ intersects either  the floor $f$, or the ceiling $f^+$.}%
\label{fig:bg}
\end{center}
\end{figure}

In general, though, $\gamma:=\bd B(g) \cap \text{int}(B_0)$ is 
a collection of algebraic arcs and closed or unbounded curves, 
each of which belongs to either 
$g \cap f$ or $g \cap f^+$. Note that 
$\gamma \cup \bd B_0$ partitions
$B_0$ into connected regions, over each of which $g$ either `floats'
between $f$ and $f^+$, or lies below $f$, or lies above $f^+$.
We need to place $g$ in the conflict list of precisely those prisms
whose projections overlap a region of the first kind (the union 
of these regions is $B(g)$).

Our strategy is to trace $\gamma$ through $\VD(B_0)$ (which is the 
$xy$-projection of $\VD(B)$). Clearly, $g$ belongs to $\CL(\tau)$ 
for every prism $\tau$ 
whose projection meets $\gamma$. We then perform a search over the 
adjacency graph of the prisms (which connects those pairs of prisms 
that have overlapping $yz$-faces), and `broadcast' $g$ to all the prisms
that we reach, placing $g$ in their conflict list. This last step is
easy to implement, in total time proportional to the total size 
of the conflict lists,
so we focus on the step of tracing $\gamma$.

Consider the task of tracing the portion of $\gamma$ formed by 
the $xy$-projection of $g\cap f$.
For simplicity of presentation, continue to refer to it as $\gamma$.
Each connected component of $\gamma$ is either a closed (or unbounded) curve,
or an arc whose endpoints lie on $\bd B_0$. We compute a point on each
connected component, and then locate, for each point, the trapezoid
of $\VD(B_0)$ that 
contains it. We then trace $\gamma$ from these points and trapezoids,
in both directions, within $\VD(B_0)$, through the adjacency graph of 
the trapezoids, and place
$g$ in the conflict list of each prism whose projection 
forms a trapezoid that we reach. 
To obtain a point on each component, we go over all prisms
$\tau \in \Pi_f$ for which $g \in \CL(\tau)$, and we 
check whether $g$ meets $f$ 
(the floor of $B$) within $\tau$. If so, we take one point in 
each component of the intersection $g \cap f\cap\tau$, and use it
as a starting point. If no such $\tau$ is found, we know that 
$B(g)=B_0$, and we proceed as above (i.e., place $g$ in the conflict 
lists of all new prisms). Note that this may 
generate several starting points along the same component of $\gamma$.
To avoid tracing the component multiple times, we find, for each
point, the new prism that contains it, and mark that prism as already
visited (by the current component of $\gamma$). 
In this way, the tracing of the component from some starting
point terminates when we reach such a marked prism.

The cost of this procedure is 
$O\big(\sum_{\tau \in \Pi_f} |\CL(\tau)|\log n + \sum_{\tau' \in \Pi_f'} 
|\CL(\tau')|\big)$, where $\Pi_f'$ is the set of new prisms.
The $\log n$-factor comes from the 
cost of locating the starting points in the new prisms. The number
of starting points is at most the total conflict size of the old 
prisms.
We find the conflict lists of the new prisms whose ceiling
is on $f$ analogously, using $M_f^b$ instead of $M_f^t$. The 
same strategy works also for those new prisms that have $f$ 
as defining only an edge or a vertex.

\paragraph{Third step: Removal of prisms that are above $\A_{\le t}(F_{i+1})$.}
Consider an old prism $\tau$ whose ceiling was
part of $L_t(F_{i})$.
If $f$ passes fully below $\tau$, the level of $\tau$ increases
to $t+1$, and $\tau$ has to be removed. Similarly,
some new prisms may be at level $t+1$, and we also need to remove them.
If we could explicitly keep track of the level of each prism,
and update these counters after each insertion, the 
removal of these ``overhanging'' prisms would be trivial.
However, there may be many prisms that lie 
fully above $f$, and broadcasting to all of them that their level 
has increased by $1$ is too expensive in general.

Instead, we discover the prisms to be discarded on the fly.
There are three cases.
The first case occurs when  $f$ lies completely above 
$L_t(F_{i})$. Then, we have $\Pi_f = \emptyset$, and 
there is nothing to do. 
In the second case, $f$ lies completely below $L_t(F_{i})$, and 
we must discard all prisms whose ceiling touches the top boundary of
$\A^{+f}_{\le t}(F_{i})$. This case will be 
very similar to the next case, so we defer it for now.
In the third case, $f$ intersects the $t$-level 
$L_t(F_{i})$. This intersection manifests itself as 
a collection of real intersection edges on 
$M_f^t$ between $f$ and surfaces from $F_i$,
such that each edge is incident on only one side to a stage-1 cell 
$B \in \B$ that has its floor on $f$ (and on the 
other side, there are no prisms in $\Pi_f$ that intersect $f$).
We can identify these intersection edges by inspecting 
all real intersection edges $e$ on $M_f^t$ and by checking 
whether there are incident cells from $\B$ on only one side 
of $e$.
The resulting intersection edges bound the connected regions on 
the top boundary of $\A^{+f}_{\le t}(F_{i})$ where $f$
dips below the original $t$-level 
$L_t(F_{i})$. In each such region,
we must remove all prisms whose ceiling touches the  
top boundary of $\A^{+f}_{\le t}(F_{i})$.

To better understand this 
process, we note that each stage-1 cell $C$ of the 
vertical decomposition is such that all prisms in 
$\VD(C)$ are at the same level of $\A^{+f}_{\le t}(F_{i})$ and are 
adjacent to each other only through overlapping $yz$-faces. 
Hence, if one prism from $\VD(C)$ has its ceiling on the 
top boundary of $\A^{+f}_{\le t}(F_{i})$, they all do. Furthermore,
if we have one such prism in $\VD(C)$, we can find all these 
prisms in $O(|\VD(C)|)$ time, using the navigational pointers between 
adjacent stage-2 prisms. Thus, we can start the process by 
identifying all stage-1 cells $B \in \B$ on $M_f^t$ whose ceiling 
touches the top-boundary of $\A^{+f}_{\le t}(F_{i})$, and 
by collecting the prisms in $\VD(B)$ through a simple traversal. 
However, this does not suffice, because there could still be stage-1 
cells to be removed that are not visible on $M_f^t$, 
i.e., stage-1 cells whose floor does not lie on $f$.  These stage-1 cells 
can be found through a traversal from the stage-1 cells that 
we have already identified. Recall that we also  
maintain pointers between the vertices of the stage-1 cells 
and the prisms that are incident to them. These cells can 
be used to pass from one stage-1 cell $B$ to adjacent stage-1 cells
in the same stage-0 cell (i.e., stage-1 cells that share 
a (partial) vertical curtain with $B$), and also to adjacent stage-1
cells (that share an $x$-edge where the floor and 
the ceiling meet) in another stage-0 cell. During the traversal,
we mark all prisms that we encounter, to avoid multiple 
visits to the same prism. Thus, the total running time is 
proportional to the number of prisms that we discard.

The third case is very similar, only that we do not start the 
traversal from $M_f^t$, but from the vertex on 
$L_t(F_{i})$ that that we maintain throughout the algorithm. 
Clearly, after discarding the superfluous prisms, 
we can update this vertex with no additional overhead.
To conclude, the overall expected running time of the above 
procedure is proportional to the overall size of all the 
conflict lists that have been generated during 
the incremental process plus the overall number of generated prisms 
times a logarithmic factor.

\end{document}